%% file: main.tex
\documentclass{article}[11pt]
\usepackage{titling}
\usepackage{ amssymb }
\usepackage{amsfonts}
\usepackage{amsthm}
\usepackage{amsmath}
\usepackage{bbm}
\usepackage{setspace}
\usepackage[normalem]{ulem}

\usepackage{anyfontsize}

\usepackage{float}

\usepackage[utf8]{inputenc}

\def\showauthornotes{0}

\def\showkeys{0}
\def\showdraftbox{0}
\def\showcolorlinks{1}
\def\usemicrotype{1}
\def\showfixme{0}

\usepackage{csquotes}

\usepackage[
backend=biber,
maxcitenames=50,
maxbibnames=99,
style=alphabetic,
]{biblatex}
 \usepackage{titlesec}
\title{Low Acceptance Agreement Tests\\ via Bounded-Degree Symplectic HDXs}
\author{Anonymous}
\author{Yotam Dikstein\thanks{Institute for Advanced Study, USA. email: yotam.dikstein@gmail.com.}, \;Irit Dinur\thanks{Weizmann Institute of Science, ISRAEL. email: irit.dinur@weizmann.ac.il. Supported by ERC grant 772839, and ISF grant 2073/21.}, and Alexander Lubotzky\thanks{Weizmann Institute of Science, ISRAEL. email: alex.lubotzky@mail.huji.ac.il. Supported by the ERC grant 882751, and by a research grant from the Center for New Scientists at the Weizmann Institute of Science.}}
\date{\today}

\input{macros.tex}

\begin{document}\clearpage\thispagestyle{empty}

\maketitle
\begin{abstract}
We solve the derandomized direct product testing question in the low acceptance regime, by constructing new high dimensional expanders that have no small connected covers. We show that our complexes have swap cocycle expansion, which allows us to deduce the agreement theorem by relying on previous work.

Derandomized direct product testing, also known as agreement testing, is the following problem. Let $X$ be a family of $k$-element subsets of $[N]$ and let $\sett{f_s:s\to\Sigma}{s\in X}$ be an ensemble of local functions, each defined over a subset $s\subset [N]$. Suppose that we run the following so-called agreement test: choose a random pair of sets $s_1,s_2\in X$ that intersect on $\sqrt k$ elements, and accept if $f_{s_1},f_{s_2}$ agree on the elements in $s_1\cap s_2$. We denote the success probability of this test by $\agr(\set{f_s})$. Given that $\agr(\set{f_s})=\eps>0$, is there a global function $G:[N]\to\Sigma$ such that $f_s = G|_s$ for a non-negligible fraction of $s\in X$ ? 

We construct a family $X$ of $k$-subsets of $[N]$ such that $|X| = O(N)$ and such that it satisfies the low acceptance agreement theorem. Namely, %
   \begin{equation*} 
    \agr (\set{f_s}) > \eps  \quad \Longrightarrow \quad \exists G:[N]\to\Sigma,\quad \Pr_s[f_s\overset{0.99}{\approx} G|_s]\geq\poly(\eps).
\end{equation*}
A key idea is to replace the well-studied LSV complexes by symplectic high dimensional expanders (HDXs). The family $X$ is just the $k$-faces of the new symplectic HDXs. The later serve our needs better since their fundamental group satisfies the congruence subgroup property, which implies that they lack small covers. We also give a polynomial-time algorithm to construct this family of symplectic HDXs.
\end{abstract}

\section{Introduction}
Any function $f:[N]\to\Sigma$ can be encoded by specifying its restrictions to certain subsets $s_1,s_2,\ldots \subset [N]$. The direct product encoding specifies the restriction of $f$ to all $k$-element subsets. While certainly redundant, this encoding allows an algorithm such as a PCP verifier to access $k$ inputs of $f$ with only a single query to the direct product encoding. But, this is only useful as long as the encoding is valid. An agreement test is a property tester for the validity of this encoding. The natural two-query test, called the V-test, is as follows: choose a random pair of sets $s_1,s_2$ with prescribed intersection size ($\sqrt k$ in our case) and accept if $f_{s_1},f_{s_2}$ agree on the elements in $s_1\cap s_2$. 

Direct product tests were introduced by Goldreich and Safra in \cite{GolSaf97} as an abstraction of PCP low degree tests. A sequence of works analyzed direct product tests, \cite{GolSaf97,DinurR06, DinurG2008,ImpagliazzoKW2012, DinurS14, DinurL2017}. 
There are two main parameter regimes of interest. The ``$99\%$'' or high acceptance regime which is natural in the world of property testing, and the ``$1\%$'' or low acceptance regime, which is the main regime of interest in the world of PCPs, and is often more challenging. In this regime the goal is to show that given an ensemble of local functions $\set{f_s}_{s\in X}$, where $X$ is some family of subsets, even if the test accepts with a small but non-negligible probability $\varepsilon$, the given encoding is $\varepsilon'$ correlated to a valid one, for some $\varepsilon'$ that depends favorably on $\varepsilon$.
Denoting the success probability of the test by $\agr (\set{f_s})$, 
\begin{equation}\label{eq:LA}
    \agr (\set{f_s}) > \varepsilon  \quad \Longrightarrow \quad \exists G:[N]\to\Sigma,\quad \Pr_s[f_s {\approx} G|_s]\geq \varepsilon'.
\end{equation}

The high redundancy of the direct product encoding, taking $N$ symbols to $\binom N k\approx N^k$ symbols, has lead researchers to look for {\em derandomized} direct product tests. Derandomization means, in this context, coming up with a family of $k$-sets that is much smaller than $\binom{N}{k}$, and with an appropriate testing distribution, such that the new encoding supports an agreement test just like in the fully redundant direct product case. Goldreich and Safra \cite{GolSaf97} showed a certain derandomized direct product test in the $99\%$ regime, and Impagliazzo, Kabanets, and Wigderson \cite{ImpagliazzoKW2012} showed a derandomized direct product test in the $1\%$ regime, with a family of $k$-sets whose size is $N^{c}$ for some $c>2$ that is independent of $k$. 

The recent emergence of the area of high dimensional expansion gave hope for a derandomized direct product test using a family of subsets that is based on high dimensional expanders. One of the early works that involve high dimensional expanders within theoretical computer science \cite{DinurK2017} has shown that indeed, local spectral expanders provide an optimal derandomization for direct product tests in the 99\% regime (the work \cite{DinurK2017} has some restrictions which were later removed in \cite{DiksteinD2019}). 

The works of \cite{DinurK2017,DiksteinD2019} combine spectral techniques from high dimensional expansion with combinatorial machinery from the non-derandomized setup. For a while it wasn't clear why these techniques fail to go beyond the 99\% regime and into the 1\% regime.

It turns out that unlike the case of the 99\% regime, {\em not every high dimensional expander supports a 1\% agreement test}. The obstacle was discovered recently, in two concurrent works \cite{DiksteinD2023agr, BafnaM2023}: complexes $X$ that have small covers do not support an agreement test (for a definition of a covering map, see \pref{sec:building-background}), since the cover itself provides a counterexample. 

Both \cite{BafnaM2023} and \cite{DiksteinD2023agr} show that an agreement theorem holds for any high dimensional expander that lacks small covers\footnote{\cite{BafnaM2023} don't use the language of covers but rather of {\em UG coboundary expansion}, with respect to non-abelian group coefficients. It is known, see e.g. \cite{DinurM2019}, that coboundary expansion with non abelian coefficients is related to covers and their stability.}. Moreover, embracing the cover obstacle, \cite{DiksteinD2023agr} proved a modified agreement test theorem which shows that for high dimensional expanders, covers are the only obstacle. The theorem, whose details are given in full below in \pref{thm:agrcover}, holds under an additional condition on the complex, which is called {\em swap cocycle expansion}. Swap cocycle expansion of a complex $X$ pertains to cocycle expansion of a related complex, called the faces complex, and denoted $FX$ (see \pref{def:face-complex}). It was further shown in \cite{DiksteinD2023swap} that the spherical building $A_n$ associated with $SL_n(\F_p)$ is a swap cocycle expander\footnote{In previous works, this was termed a swap \emph{cosystolic} expander.}. Spherical buildings lack {\em any} connected covers, so the above theorems suffice for deducing a derandomized direct product test (see \cite[Corollary 1.6]{DiksteinD2023agr}) with parameters improving upon the previously best known from \cite{ImpagliazzoKW2012}. Nevertheless, the complex $A_n$ is not bounded-degree, so it still does not solve the main problem, which is to find a family of subsets of $[N]$ whose size is {\em linear} in $N$, and which supports a 1\% agreement test.

\subsection{Results} 
We construct new high dimensional expanders with no small covers, and show that they support 1\% agreement tests. 
Our complexes are constructed as quotients of the affine Bruhat-Tits building associated with the symplectic group  \(Sp(2g,\mathbb{Q}_p)\), denoted Let \(\tilde{C}_g=\tilde{C}_g(\mathbb{Q}_p)\). We provide some background in \pref{sec:building-background}. Our main theorem is as follows,
\begin{theorem} \label{thm:main}
    For every $\eps>0$, there exist $c>0$ and large enough integers $k<g$ and a prime $p$ such that the following holds. There exists an infinite family of constant degree connected $g+1$-partite simplicial complexes \(\set{X_N}_N\) that are finite quotients of \(\tilde{C}_g\) such that $X_N$ has $N$ vertices and $O_g(N)$ faces, and such that the following holds. Let $\Sigma$ be a finite alphabet, and let $\sett{f_s:s\to\Sigma}{s\in X(k)}$ be an ensemble of local functions on $X(k)$ (where $X$ is any complex in the family $\set{X_N}$). Let $\agr(\set{f_s})$ denote the success probability of the agreement test with respect to the distribution that chooses $s_1,s_2\in X(k)$ conditioned on $|s_1\cap s_2| = \sqrt k$. Then
\begin{equation}\label{eq:CLA}
    \agr (\set{f_s}) > \eps  \quad \Longrightarrow \quad \exists G:X(0)\to\Sigma,\quad \Pr_s[f_s \approx G|_s]\geq \eps^c. 
\end{equation} 
where $f\approx f'$ indicates that the two functions agree on $0.99$ of their domain.
\end{theorem}
This theorem was also concurrently proved by \cite{BafnaLM2024}. 

In \pref{sec:poly-const} we will show that one can choose the complexes in \pref{thm:main} to be \emph{polynomially constructible}.

In previous work, \cite{DiksteinD2023agr} showed a similar conclusion under the assumption that $X$ is a quotient of the building associated with $SL_n(\mathbb{Q}_p)$, and assuming that it has no small covers. Unfortunately, it is not known how to construct such a quotient. The novel idea of the current paper is to replace $SL_n(\mathbb{Q}_p)$ by $Sp(2g,\mathbb{Q}_p)$ and to construct complexes that are analogous to those of \cite{LubotzkySV2005a,LubotzkySV2005b} with the additional property of having no small cover. 
\begin{theorem} \label{thm:CLhighdim}
    Let \(m\geq 2\) and \(g \geq 100 \sqrt{m\log (m)}\). For every prime \(p\), there exists an infinite family of connected simplicial complexes \(X\) that are finite quotients of \(\tilde{C}_g(\mathbb{Q}_p)\) such that every \(X\) has no connected \(m'\)-covers for any \(1< m' \leq m\). 
\end{theorem}
The case of $m=2$ follows from \cite{ChapmanL2023}. Given \pref{thm:CLhighdim}, the remaining work is to show swap coboundary expansion (a key requirement for the agreement theorem to hold, see \pref{def:swap}) for buildings of symplectic type. This was previously shown \cite{DiksteinD2023swap} for buildings associated with $SL_n$. We adapt those techniques for the symplectic case, and show,
\begin{theorem}\torestate{ \label{thm:swap-coboundary-expansion}
Let \(d\) be an integer. There is some \(p_0=p_0(d)\) such that for all primes \(p > p_0\) the following holds. Let \(X\) be a quotient of $\tilde C_g(\Q_p)$, the affine symplectic building associated with \(Sp(2g,\mathbb{Q}_p)\), for \(g \geq d^5\). Then \(X\) is a \((d,\exp(-O(\sqrt{d})))\)-swap cocycle expander.}
\end{theorem}

We rely on the following low soundness agreement theorem from \cite{DiksteinD2023agr}.

\begin{theorem}[Informal version of \pref{thm:agrcover-technical}]\label{thm:agrcover} Let $k\in \mathbb{N}$, and let $\eps >\Omega(1/\log k)$. There exists \(\lambda > 0\) and a sufficiently large \(d > k\) such that the following holds. Let $X$ be a $d$-dimensional \(\lambda\)-high dimensional expander with \((d,\exp(-O(\sqrt{d}))\)-swap-cocycle-expansion. 
    Let $\sett{f_s:s\to\Sigma}{s\in X(k)}$ be an ensemble of local functions on $X(k)$.
\begin{equation}\label{eq:CLA2}
    \agr (\set{f_s}) > \eps  \quad \Longrightarrow \quad \exists Y\xrightarrowdbl{\rho} X, \exists G:Y(0)\to\Sigma,\quad \Pr_s[f_s\hbox{ is explained by }G]\geq\poly(\eps).
\end{equation} 
where $\rho:Y\to X$ is a $\ell=\poly(1/\epsilon)$ covering map. 
\end{theorem}
Here by ``explained'' we formally mean that there exists some \(\tilde{s} \in \rho^{-1}(s)\) such that \(f_s \approx G \circ \rho|_{\tilde{s}}\). We remark that a similar result was shown in \cite{BafnaM2023} for the simplified case where $X$ has no connected covers of size $\poly(1/\varepsilon)$, which still suffices to prove \pref{thm:main}.

Since we construct in \pref{thm:CLhighdim} complexes $X$ with no connected $m$-covers, for $m \leq \poly(1/\eps)$, we deduce that $Y$ must be a collection of disjoint copies of $X$ and this proves our main result.

\subsection{Proof overview}

Our main theorem, \pref{thm:main}, is stated in terms of agreement testing. We rely on \pref{thm:agrcover} from \cite{DiksteinD2023agr} which translates the problem into constructing a complex that is a swap cocycle expander with no non-trivial covers of size \(\leq m = \poly(\varepsilon)\). If $X$ is such a complex, then the cover $Y$ appearing in \eqref{eq:CLA2} is necessarily disconnected into \(m\) disjoint copies of \(X\). It follows easily that in such a situation there exists at least one disjoint copy of \(X\) such that the restriction of the function \(G\) to that copy also satisfies \eqref{eq:CLA}, thus proving \pref{thm:main}.

Previous work on the quotients of \(\tilde{A}_n\), the \(SL_n\)-affine building, showed that they are sufficient swap cocycle expanders on which one can apply \pref{thm:agrcover}, \cite{DiksteinD2023swap}. These complexes include the \cite{LubotzkySV2005b} complexes. Unfortunately, it is not known whether there exist such complexes that have no small covers. In fact, by \cite{KaufmanKL2014} some \cite{LubotzkySV2005b} complexes are known to have \(2\)-covers. Therefore we searched for complexes that are similar enough to those of \cite{LubotzkySV2005b}, so that we could argue about the swap cocycle expansion, but distinct enough so that we could also show that they have no small covers.

The solution to our problem is \emph{going symplectic}. An \cite{LubotzkySV2005b} complex \(X\) is a quotient of \(\tilde{A}_n(\mathbb{Q}_p)\), the affine \(SL_n\)-building by some lattice \(\Gamma \leq SL_n\) (i.e. a discrete cocompact subgroup). As \(\tilde{A}_n(\mathbb{Q}_p)\) is simply connected, a classical fact from topology is that $X$ has a connected \(m\)-cover if and only if $\Gamma$ has a subgroup \(\Gamma' \leq \Gamma\) of index $m$. This observation appeared in \cite{KaufmanKL2014}, which showed that \cite{LubotzkySV2005b} complexes with no \(2\)-covers exist, relying on a special case of Serre's conjecture \cite{Serre1970}, namely, the congruence subgroup property on \(SL_n\). Recent work \cite{ChapmanL2023}, used the fact that this conjecture is known to hold in \emph{the symplectic group}, \(Sp(2g,\mathbb{Q}_p)\), to construct such complexes with no connected \(2\)-covers and to show that they are coboundary expanders over \(\mathbb{F}_2\). These complexes are quotients of \(\tilde{C}_g\), the affine symplectic building. Our work extends this idea to any \(m > 1\) by choosing carefully different parameters. That is, for any fixed $m$, we construct quotients of \(\tilde{C}_g\) that have no connected \(m'\)-covers for all $m'<m$.

\paragraph*{Constructing the complexes}We construct a family of subgroups \(\Gamma \leq Sp(2g,\mathbb{Q}_p)\) and quotient \(\tilde{C}_g\) by them to construct our family of agreement testers. 
As a first step, we use \cite{LubotzkyD1981} which says that there is a \(1-1\)-correspondence between \(m\)-index subgroups of \(\Gamma\) and \(m\)-index open subgroups of \(\widehat{\Gamma}\), the profinite completion of \(\Gamma\) (see \pref{sec:no-small-covers} for the precise definition of the profinite completion). For \(\Gamma \leq Sp(2g,\mathbb{Q}_p)\) as above, the congruence subgroup problem has an affirmative solution (proven by Rapinchuk \cite{Rapinchuk1989}), from this solution one obtains a clear structure of the profinite completion from which one can describe all small index open subgroups.

Using the identification of \(Sp(2g,\mathbb{Q}_p)\) with a unitary group over a suitable quaternion algebra, we construct a discrete cocompact arithmetic subgroup of \(\Gamma_0 \subseteq Sp(2g,\mathbb{Q}_p)\). We call this lattice \(\Gamma_0\) and search for our subgroups inside \(\Gamma_0\). Inside this \(\Gamma_0\), we are able to show existence of normal subgroups \(\Gamma_0 \geq \Gamma_1 \geq \Gamma_2 \geq \dots\) whose profinite completion is well structured - it is a product \(\widehat{\Gamma_i} = H_i \times \prod_{q} Sp(2g,\mathbb{Z}_q)\) for some primes \(q\). The group \(H_i\) is pro-\(\ell\) for a prime $\ell$ so it has no open subgroups of index less than \(\ell\) (here we can take \(\ell > m\)).
We show by using \cite{Weigel1996} that the groups \(Sp(2g,\mathbb{Z}_q)\) do not have open subgroups of index less than \(m\), and that this implies that the product also has no such subgroups. This is enough to obtain \pref{thm:CLhighdim}.

\paragraph*{Swap cocycle expansion}Our work is not done. We still need to argue that the complexes we constructed are swap cocycle expanders to apply \pref{thm:agrcover}. Fortunately, in a previous work, swap cocycle expansion was proven for quotients of \(\tilde{A}_n\) \cite{DiksteinD2023swap}. We extend the proof of \cite{DiksteinD2023swap} to quotients of \(\tilde{C}_g\). But in order to do this, we need to show that color restrictions (certain subcomplexes) of links of \(\tilde{C}_g\) are coboundary expanders. For an overview of the extension itself, we refer the reader to \pref{sec:proof-of-faces-complex-lower-bound}; for the rest of the overview we focus on the main new component, which is showing coboundary expansion of the links of \(\tilde{C}_g\).

The color restrictions we need to analyze are the following family of three partite complexes \(C_g^{i_1,i_2,i_3}\). To define this complex we consider \(\mathbb{F}_p^{2g}\) together with an asymmetric bilinear form denoted by \(\iprod{\cdot , \cdot}\). An isotropic subspace is a subspace \(v \subseteq \mathbb{F}_p^{2g}\) such that for every two vectors \(x,y \in v\), \(\iprod{x,y} = 0\). The complex \(C_g^{i_0,i_1,i_2}\) has three types of vertices, 
\[C_g^{i_0,i_1,i_2} [i_j] = \sett{v \subseteq \mathbb{F}_p^{2g}}{\dim(v) = i_j \ve v \text{ is isotropic}}.\]
We have a triangle \(\set{v_0,v_1,v_2} \in C_g^{i_0,i_1,i_2}(2)\) if \(v_0 \subset v_1 \subset v_2\). Henceforth we denote by \(I = \set{i_0,i_1,i_2}\) and the complex \(C_g^I\) for short. A technique in \cite{DiksteinD2023swap} allows us to reduce to the case where \(i_0 \ll i_1 \ll i_2 \ll g\) so we assume so.

To show coboundary expansion of this complex, we use the well known cone method. This method first appeared implicitly in \cite{Gromov2010}, and by now is one of the main tools for showing coboundary expansion \cite{LubotzkyMM2016, KozlovM2019,KaufmanO2021,KaufmanM2018,DiksteinD2023cbdry,DiksteinD2023swap}. This method is based on the observation that isoperimetric inequalities in symmetric complexes imply coboundary expansion. In other words, if one can show that \(O(1)\)-length cycles in \(C_g^I\) have triangle tilings with \(O(1)\)-triangles, then link is a coboundary expander. See \pref{sec:preliminaries} for the precise definitions and statement.

Let \(\cont=(v_0,v_1,\dots,v_t=v_0)\) be a cycle where \(t\) is a constant. For simplicity we assume our cycles only contain vertices in \(C_g^I[i_0]\) and \(C_g^I[i_1]\) (the proof reduces to this case). The tiling is simple so we describe it here:
\begin{enumerate}
        \item We find an isotropic subspace \(u^*\) of dimension $i_1$ that is perpendicular to every subspace participating in \(\cont\). Here we use the assumption that the sum of all subspaces in the cycle has low enough dimension so one such subspace exists by dimension considerations. To do so we rely on the fact that \(v_i + u^*\) is isotropic.
        \item We fix an arbitrary ``middle vertex'' \(u^{**} \subseteq u^{*}\) such that \(u^{**}\in C_g^I[i_0]\), and decompose the \(t\)-cycle to many \(t\)-cycles touching \(u^{**}\). These six cycles have three of the original subspaces in the cycle \((v_i \subseteq v_{i+1} \supseteq v_{i+2})\), and three `new' subspaces: \(u^{**}\) and two \(u_i, u_{i+2}\) that are formed by adding to \(v_i + u^{**}\) new vectors from \(u^*\) until reaching an \(i_1\)-dimensional isotropic subspace.
        \item We tile every such \(6\)-cycle separately. We observe that the sum of spaces in a \(6\)-cycle as above is contained in an isotropic subspace \(u^* + v_{i+1}\). A classical fact is that every isotropic subspace is contained in a maximal isotropic subspace, and in particular one can find some \(x \in C_g[i_2]\) that contains \(u^* + v_{i+1}\). This implies that \(x\) contains all vertices in the \(6\)-cycle, and we can therefore tile the cycle with triangles connected to \(x\).
\end{enumerate}
We illustrate this in \pref{fig:subspace-contraction}.
    \begin{figure}
        \centering
        \includegraphics[scale=0.5]{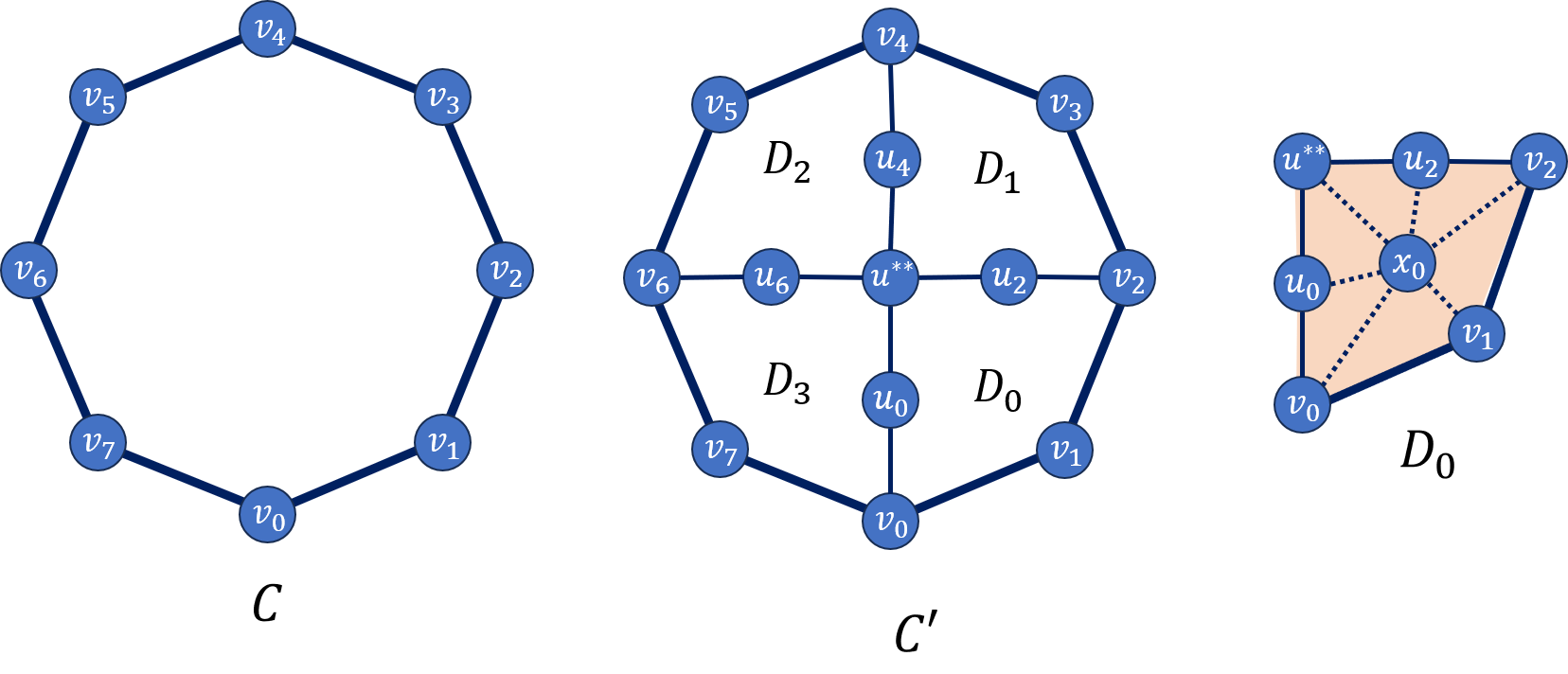}
        \caption{The contraction}
        \label{fig:subspace-contraction}
    \end{figure}
Local spectral expansion of these complexes is also needed so that we can apply \pref{thm:agrcover}. This property of the quotients is already known to experts in the field, but we include a proof here for completeness.

\subsection{Related works}
This work touches upon different fields of study.
\paragraph*{Agreement testing and PCPs} A main motivation behind agreement tests in the low acceptance ``$1\%$'' regime, is for constructing PCPs. Works such as \cite{RaSa} and \cite{ArSu} showed agreement in this regime that was translated into PCPs with low soundness. The most popular application of agreement tests is in relation to the parallel repetition theorem \cite{Raz-parrep}, where an agreement theorem on the complete complex is used to construct a large gap PCP.  Work by \cite{ImpagliazzoKW2012} gave a black-box conversion from an agreement-test theorem to a parallel repetition theorem, see also \cite{DinurS14, DinurS14-parrep}. The work \cite{ImpagliazzoKW2012} constructed a more derandomized family $X$ that satisfies \eqref{eq:CLA}; and later work by \cite{DinurM11} transformed it into an length-efficient PCP with a large gap. Agreement tests on subspaces played a role in the proof of the  2-to-1 theorem \cite{DKKMS,KMS}. Although in this setup the set system at hand did not have \eqref{eq:CLA}, a weaker guarantee sufficed to make these tests useful in the 2:2 inner verifier.
Agreement tests on high dimensional expanders were studied first by \cite{DinurK2017}. Work by \cite{GotlibK2022} suggested a connection between list agreement and coboundary expansion. The works \cite{BafnaM2023} and \cite{DiksteinD2023agr} independently showed that swap coboundary expanding HDXs are good for agreement tests as in \pref{thm:agrcover}.

\paragraph*{Coboundary expansion}Coboundary expansion was defined by Linial, Meshulam and Wallach \cite{LinialM2006}, \cite{MeshulamW09}, and indpendently by Gromov \cite{Gromov2010}. Kaufman, Kazhdan and Lubotzky \cite{KaufmanKL2014} introduced a local to global argument for proving cocycle expansion of $1$-cochains in the \emph{bounded-degree} complexes of \cite{LubotzkySV2005a,LubotzkySV2005b}, solving a \(2\)-dimensional case of the Gromov topological overlap problem. This was generalized later by Evra and Kaufman \cite{EvraK2016} to cocycle expansion in all dimensions, thus solving the problem in all dimensions. Following ideas that implicitly appeared in Gromov's work, Lubotzky Mozes and Meshulam analyzed the coboundary expansion spherical buildings \cite{LubotzkyMM2016}, using what was later known as the cone method. Works by \cite{KozlovM2019} and \cite{KaufmanO2021} abstracted this notion of cones. Techniques for lower bounding coboundary expansion were further developed in \cite{DiksteinD2023cbdry} and \cite{DiksteinD2023swap}. 
Dinur and Meshulam observed the connection between cocycle expansion and cover-stability. Later on, this connection was used by \cite{GotlibK2022} to analyze the problem of list-agreement on coboundary expanders. Work by Chapman and Lubotzky implemented a version of cocycle expansion to group stability \cite{ChapmanL2023stability}.
Another work by \cite{ChapmanL2023} constructed bounded degree coboundary expanders over \(\mathbb{F}_2\), by using a similar construction to the one presented in this work.

\paragraph*{Buildings and their quotients}The theory of buildings is extremely important for understanding the structure of classical groups. They were introduced by Jacques Tits \cite{Tits1974}. For more about buildings see \cite{AbramenkoB2008} or \cite{Weiss2008} and references therein. In the early 00's a lot of work was done to construct quotients of the \(SL_n\)-Bruhat Tits building to get Ramanujan complexes such as \cite{Ballantine2000}, \cite{CartwrightSZ2003}, \cite{Li2004}, \cite{LubotzkySV2005a} and \cite{LubotzkySV2005b}. After the work of \cite{KaufmanKL2014}, these complexes caught the attention of the TCS community due to their local properties. The local spectral expansion of these objects was put to a good use in the seminal work of Garland \cite{Garland1973}. This was further used in \cite{KaufmanM2017high} and \cite{EvraK2016}.
The notion of local spectral expansion we use today due to \cite{DinurK2017}, was tailored to suit these complexes.

\medskip

Bafna, Lifshitz and Minzer concurrently proved a theorem similar to \pref{thm:main} \cite{BafnaLM2024}. Their construction also uses quotients of \(\tilde{C}_g\) and relies on a theorem similar to \pref{thm:agrcover} (which appeared in independent previous work \cite{BafnaM2023}, concurrent to \cite{DiksteinD2023agr}) and on swap cocycle expansion of the quotients of \(\tilde{C}_g\).
 
\subsection{Open Questions}
We constructed a family of bounded degree complexes with sound agreement tests in the \(1\%\)-regime. These tests sometimes appear as a combinatorial gadget in PCP constructions. However, there is still no construction of a PCP that uses high dimensional expanders. It seems that constructing such a PCP may improve the tradeoff between set size, degree and soundness, possibly leading to an advancement on the famous sliding scale conjecture (see e.g. \cite{MR10}). However to do so, one must first improve upon the somewhat modest parameters in \pref{thm:agrcover}, and figure out a way to embed hard constraint satisfaction problems into high dimensional expanders.

Most prior work on bounded degree high dimensional expanders use quotients of \(SL_n\)-buildings (or other constructions such as the one suggested by \cite{KaufmanO181}). The complexes we construct, coming from a different infinite object \(\tilde{C}_g\), have similar expansion properties but a different link structure. Are there other desirable properties of these complexes, that quotients of \(\tilde{A}_n\) do not possess? For instance, can other buildings provide a better degree-to-expansion tradeoff?

Finally, we mention that this work showed cocycle expansion, and got coboundary expansion with respect to small groups (it is not a priori clear, but this is an equivalent statement to saying that the complex has no small covers). However, we do not know whether there exists bounded degree local spectral expanders \(X\) that are coboundary expanders with respect to \emph{all possible groups}. In particular, \(X\) is simply connected \cite{Surowski1984}. Current constructions of bounded degree local spectral expanders are all quotients of some infinite object by some group acting on the object. Any such complex is not a coboundary expander with respect to the acting group. Hence such a construction is an interesting open problem in this area. As an intermediate step, could one construct a complex that is a coboundary expander with respect to all finite groups?

\subsection{Organization}
\pref{sec:preliminaries} contains the necessary preliminaries for the paper. We describe the buildings we work with in the paper in \pref{sec:buildings}. \pref{thm:main} is proven in \pref{sec:mainproof}, and follows directly from \pref{cor:v-z-soundness}. 
\pref{sec:no-small-covers} contains the new construction of expanders with no small covers, proving  \pref{thm:CLhighdim}. We prove the swap cocycle expansion theorem, \pref{thm:swap-coboundary-expansion}, in \pref{sec:expansion-of-symplectic-building} where we also bound the coboundary expansion and the local spectral expansion of buildings of type \(C\). 

\subsection*{Acknowledgements}
We thank Alan Reid and Andrei Rapinchuk for helpful discussions. We thank Gil Melnik for his help with the figures.

\input{preliminaries}
\input{agreement}
\input{nosmallcovers}
\section{Expansion Properties of the Symplectic Building} \label{sec:expansion-of-symplectic-building}
In this section we prove that the symplectic spherical building is a \(\Omega(1)\)-coboundary expander and use this to show that quotients of the affine symplectic building are \((\d_1,\exp(-O(\sqrt{\d_1})))\)-swap cocycle expanders.

\input{coboundary}

\input{swap}
\input{spectral}
\printbibliography

\appendix
\input{appendix}
\end{document}

%% file: macros.tex
\providecommand{\RR}{\mathbb{R}}

\usepackage[l2tabu, orthodox]{nag}

\usepackage{xspace,enumerate}

\usepackage[dvipsnames]{xcolor}

\usepackage[T1]{fontenc}
\usepackage[full]{textcomp}

\usepackage[american]{babel}

\usepackage{mathtools}

\usepackage{amsthm}
\usepackage{amsmath}

\newtheorem{theorem}{Theorem}[section]
\newtheorem*{theorem*}{Theorem}

\newtheorem{proposition}[theorem]{Proposition}
\newtheorem*{proposition*}{Proposition}
\newtheorem{lemma}[theorem]{Lemma}
\newtheorem*{lemma*}{Lemma}
\newtheorem{corollary}[theorem]{Corollary}
\newtheorem*{corollary*}{Corollary}
\newtheorem*{conjecture*}{Conjecture}
\newtheorem{fact}[theorem]{Fact}
\newtheorem*{fact*}{Fact}

\newtheorem*{hypothesis*}{Hypothesis}

\theoremstyle{definition}
\newtheorem{definition}[theorem]{Definition}
\newtheorem*{definition*}{Definition}

\newtheorem{example}[theorem]{Example}

\newtheorem{algorithm}[theorem]{Algorithm}

\theoremstyle{remark}
\newtheorem{claim}[theorem]{Claim}
\newtheorem*{claim*}{Claim}

\newtheorem*{remark*}{Remark}
\newtheorem{observation}[theorem]{Observation}
\newtheorem*{observation*}{Observation}

 \usepackage[letterpaper,
 top=1in,
 bottom=1in,
 left=1in,
 right=1in]{geometry}

\usepackage[varg]{pxfonts} 

\usepackage[sc,osf]{mathpazo}
\usepackage{lmodern}

\ifnum\showkeys=1
\usepackage[color]{showkeys}
\fi

\ifnum\showcolorlinks=1
\usepackage[
colorlinks=true,
urlcolor=blue,
linkcolor=blue,
citecolor=OliveGreen,
]{hyperref}
\fi

\ifnum\showcolorlinks=0
\usepackage[
colorlinks=false,
pdfborder={0 0 0}
]{hyperref}
\fi

\usepackage{prettyref}

\newcommand{\savehyperref}[2]{\texorpdfstring{\hyperref[#1]{#2}}{#2}}

\newrefformat{eq}{\savehyperref{#1}{\textup{(\ref*{#1})}}}
\newrefformat{lem}{\savehyperref{#1}{Lemma~\ref*{#1}}}
\newrefformat{def}{\savehyperref{#1}{Definition~\ref*{#1}}}
\newrefformat{thm}{\savehyperref{#1}{Theorem~\ref*{#1}}}
\newrefformat{cor}{\savehyperref{#1}{Corollary~\ref*{#1}}}
\newrefformat{cha}{\savehyperref{#1}{Chapter~\ref*{#1}}}
\newrefformat{sec}{\savehyperref{#1}{Section~\ref*{#1}}}
\newrefformat{app}{\savehyperref{#1}{Appendix~\ref*{#1}}}
\newrefformat{tab}{\savehyperref{#1}{Table~\ref*{#1}}}
\newrefformat{fig}{\savehyperref{#1}{Figure~\ref*{#1}}}
\newrefformat{hyp}{\savehyperref{#1}{Hypothesis~\ref*{#1}}}
\newrefformat{alg}{\savehyperref{#1}{Algorithm~\ref*{#1}}}
\newrefformat{rem}{\savehyperref{#1}{Remark~\ref*{#1}}}
\newrefformat{item}{\savehyperref{#1}{Item~\ref*{#1}}}
\newrefformat{obs}{\savehyperref{#1}{Observation~\ref*{#1}}}
\newrefformat{step}{\savehyperref{#1}{step~\ref*{#1}}}
\newrefformat{conj}{\savehyperref{#1}{Conjecture~\ref*{#1}}}
\newrefformat{fact}{\savehyperref{#1}{Fact~\ref*{#1}}}
\newrefformat{prop}{\savehyperref{#1}{Proposition~\ref*{#1}}}
\newrefformat{prob}{\savehyperref{#1}{Problem~\ref*{#1}}}
\newrefformat{claim}{\savehyperref{#1}{Claim~\ref*{#1}}}
\newrefformat{relax}{\savehyperref{#1}{Relaxation~\ref*{#1}}}
\newrefformat{red}{\savehyperref{#1}{Reduction~\ref*{#1}}}
\newrefformat{part}{\savehyperref{#1}{Part~\ref*{#1}}}
\newrefformat{ex}{\savehyperref{#1}{Example~\ref*{#1}}}
\newrefformat{para}{\savehyperref{#1}{Paragraph~\ref*{#1}}}
\newrefformat{ass}{\savehyperref{#1}{Assumption~\ref*{#1}}}
\newrefformat{question}{\savehyperref{#1}{Question~\ref*{#1}}}

\newcommand{\Sref}[1]{\hyperref[#1]{\S\ref*{#1}}}

\usepackage{nicefrac}

\ifnum\usemicrotype=1
\usepackage{microtype}
\fi

\ifnum\showauthornotes=1
\newcommand{\Authornote}[2]{{\sffamily\small\color{red}{[#1: #2]}}}
\newcommand{\Authornotecolored}[3]{{\sffamily\small\color{#1}{[#2: #3]}}}
\newcommand{\Authorcomment}[2]{{\sffamily\small\color{gray}{[#1: #2]}}}
\newcommand{\Authorstartcomment}[1]{\sffamily\small\color{gray}[#1: }

\newcommand{\Authorfnote}[2]{\footnote{\color{red}{#1: #2}}}
\newcommand{\Authorfixme}[1]{\Authornote{#1}{\textbf{??}}}
\newcommand{\Authormarginmark}[1]{\marginpar{\textcolor{red}{\fbox{\Large #1:!}}}}
\else
\newcommand{\Authornote}[2]{}
\newcommand{\Authornotecolored}[3]{}
\newcommand{\Authorcomment}[2]{}
\newcommand{\Authorstartcomment}[1]{}

\newcommand{\Authorfnote}[2]{}
\newcommand{\Authorfixme}[1]{}
\newcommand{\Authormarginmark}[1]{}
\fi

\ifnum\showfixme=0

\fi

\usepackage{boxedminipage}

\newcommand{\Brac}[1]{\left[#1\right]}

\newcommand{\abs}[1]{\lvert#1\rvert}
\newcommand{\Abs}[1]{\left\lvert#1\right\rvert}

\newcommand{\card}[1]{\lvert#1\rvert}

\newcommand\sett[2]{\left\{ #1 \left| \; \vphantom{#1 #2} \right. #2  \right\}}
\newcommand{\set}[1]{\{#1\}}

\newcommand{\norm}[1]{\lVert#1\rVert}

\newcommand{\iprod}[1]{\langle#1\rangle}

\def\dim{\mathrm{ dim}}
\def\sp{\mathrm{ span}}

\newcommand{\F}{\mathbb{F}}
\newcommand{\Esymb}{\mathbb{E}}
\newcommand{\Psymb}{\mathbb{P}}

\DeclareMathOperator*{\E}{\Esymb}

\DeclareMathOperator*{\ProbOp}{\Psymb}

\renewcommand{\Pr}{\ProbOp}

\def\wt{{\mathrm{wt}}}

\newcommand{\prob}[1]{\Pr \left[ {#1} \right] }
\newcommand{\Prob}[2][]{\Pr_{{#1}}\left[#2\right]} 

\newcommand{\Ex}[2][]{\E_{{#1}}\Brac{#2}}

\newcommand{\ve}{\;\hbox{and}\;}

 \usepackage{dsfont}
\usepackage{mathrsfs}

\newcommand{\textparen}[1]{\text{(#1)}}

\ifx\because\undefined
\newcommand{\because}[1]{\textparen{because #1}}
\else
\renewcommand{\because}[1]{\textparen{because #1}}
\fi

\newcommand\bdot\bullet

\DeclareMathOperator{\poly}{poly}

\DeclareMathOperator{\dist}{dist}
\DeclareMathOperator{\sign}{sign}

\newcommand{\N}{\mathbb N}

\newcommand{\C}{\mathcal C}

\def\Col{\mathsf{C}}
\def\cont{f}

\renewcommand{\leq}{\leqslant}

\renewcommand{\geq}{\geqslant}

\ifnum\showdraftbox=1

\else

\fi

\let\epsilon=\varepsilon

\numberwithin{equation}{section}

\newcommand{\MYstore}[2]{%
  \global\expandafter \def \csname MYMEMORY #1 \endcsname{#2}%
}

\newcommand{\MYload}[1]{%
  \csname MYMEMORY #1 \endcsname%
}

\newcommand{\MYnewlabel}[1]{%
  \newcommand\MYcurrentlabel{#1}%
  \MYoldlabel{#1}%
}

\newcommand{\MYdummylabel}[1]{}

\newcommand{\torestate}[1]{%
  \let\MYoldlabel\label%
  \let\label\MYnewlabel%
  #1%
  \MYstore{\MYcurrentlabel}{#1}%
  \let\label\MYoldlabel%
}

\newcommand{\restatetheorem}[1]{%
  \let\MYoldlabel\label
  \let\label\MYdummylabel
  \begin{theorem*}[Restatement of \prettyref{#1}]
    \MYload{#1}
  \end{theorem*}
  \let\label\MYoldlabel
}

\newcommand{\restatelemma}[1]{%
  \let\MYoldlabel\label
  \let\label\MYdummylabel
  \begin{lemma*}[Restatement of \prettyref{#1}]
    \MYload{#1}
  \end{lemma*}
  \let\label\MYoldlabel
}

\newcommand{\restateprop}[1]{%
  \let\MYoldlabel\label
  \let\label\MYdummylabel
  \begin{proposition*}[Restatement of \prettyref{#1}]
    \MYload{#1}
  \end{proposition*}
  \let\label\MYoldlabel
}

\newcommand{\restateclaim}[1]{%
  \let\MYoldlabel\label
  \let\label\MYdummylabel
  \begin{claim*}[Restatement of \prettyref{#1}]
    \MYload{#1}
  \end{claim*}
  \let\label\MYoldlabel
}

\newcommand{\restatecorollary}[1]{%
  \let\MYoldlabel\label
  \let\label\MYdummylabel
  \begin{corollary*}[Restatement of \prettyref{#1}]
    \MYload{#1}
  \end{corollary*}
  \let\label\MYoldlabel
}

\newcommand{\restatefact}[1]{%
  \let\MYoldlabel\label
  \let\label\MYdummylabel
  \begin{fact*}[Restatement of \prettyref{#1}]
    \MYload{#1}
  \end{fact*}
  \let\label\MYoldlabel
}

\newcommand{\restatedefinition}[1]{
\let\MYoldlabel\label 
\let\label\MYdummylabel 
\begin{definition*}[Restatement of \prettyref{#1}] 
    \MYload{#1} 
\end{definition*} 
\let\label\MYoldlabel 
} 

\newcommand{\restate}[1]{%
  \let\MYoldlabel\label
  \let\label\MYdummylabel
  \MYload{#1}
  \let\label\MYoldlabel
}

\newcommand{\eps}{\epsilon}

\allowdisplaybreaks

\newcommand{\dunion}{\mathbin{\mathaccent\cdot\cup}}

\let\pref=\prettyref

\newcommand{\dir}[1]{\overset{\to}{#1}}
\newcommand{\bproof}[1]{\begin{proof}[Proof of \pref{#1}]}
\newcommand{\eproof}{\end{proof}}
\newcommand{\coboundary}{{\delta}}

\newcommand{\Img}{{\textrm {Im}}}
\newcommand{\agr}{{\textrm {Agree}}}

\def\Q{\mathbb{Q}}

\def\S{\mathcal{S}}
\def\J{\mathcal{J}}
\def\d_1{d}
\def\X{\mathcal{Z}}
\newcommand{\FA}[2][]{\mathrm{F}^{{#1}}\!{#2}}
\newcommand{\FX}[1][r]{\mathrm{F}^{#1}\!X}

\newcommand{\FS}[1][r]{\mathrm{F}^{#1}\!\S}
\newcommand{\FD}[1][r]{\mathrm{F}^{#1}\!\Delta}

\newcommand{\xrightarrowdbl}[2][]{%
  \xrightarrow[#1]{#2}\mathrel{\mkern-14mu}\rightarrow
}

%% file: preliminaries.tex
\section{Preliminaries} \label{sec:preliminaries}
Most of this section follows definitions used in previous works \cite{DiksteinD2023agr}, \cite{DiksteinD2023cbdry} and \cite{DiksteinD2023swap}.

A pure \(d\)-dimensional simplicial complex \(X\) is a hypergraph that consists of an arbitrary collection of sets of size \((d+1)\) together with all their subsets. The $i$-faces are sets of size \(i+1\) in \(X\), denoted by \(X(i)\). 

A simplicial action of a group \(\Gamma\) on a complex \(X\) is a homomorphism \(\phi: \Gamma \to Aut(X)\). Sometimes instead of writing \(\phi(\gamma)(v)\) we write \(\gamma . v\) or even \(\gamma v\) for short.

We denote the set of {\em oriented} \(k\)-faces in \(X\) by \(\dir{X}(k) = \sett{(v_0,v_1,...,v_k)}{\set{v_0,v_1,...,v_k} \in X(k)}\).
We denote by \(diam(X)\) the {\em diameter} of the graph underlying $X$.
A {\em \((d+1)\)-partite} \(d\)-dimensional simplicial complex is a complex \(X\) such that one can decompose \(X(0) = A_0 \dunion A_1 \dunion \dots \dunion A_d\) such that for every \(s \in X(d)\) and \(0\leq i \leq d\) it holds that \(\abs{s \cap A_i} = 1\). The \emph{color} of a vertex is $col(v)=i$ such that \(v \in A_i\). More generally, the color of a face \(s\) is \(c = col(s) = \sett{col(v)}{v \in s}\). We denote by \(X[c]\) the set of faces of color \(c\) in \(X\), and for a singleton \(\set{i}\) we sometimes write \(X[i]\) instead of \(X[\set{i}]\).
We also denote by $X^I$, for $I\subset \set{0,\ldots ,d}$, the complex induced on vertices whose colors are in $I$.

\begin{definition}[Join of complexes]
    Let \(S_i\) be an \(\ell_i\) dimensional complex, for $i=1,...,k$. Let \(n=(\sum_{i=1}^k \ell_i )+k-1\). The join \(\S = \bigvee S_i\) is the \(n\)-dimensional complex whose faces are all the \(s_1 \dunion s_2 \dunion \dots s_k\) so that \(s_i \in S_i\). The distribution over top-level faces is to (independently) choose \(s_i \sim S_i(\ell_i)\) and output \(s_1 \dunion s_2 \dunion \dots\dunion s_k\).
\end{definition}
Observe that if each \(S_i\) is $(\ell_i+1)$-partite, then \(S\) is $(n+1)$-partite. Moreover, if we restrict \(S\) to a set $I$ of colors so that every two distinct colors \(j_1,j_2 \in I\) come from different complexes \(S_{i_1}\) and $S_{i_2}$, then \(S^I\) is a complete $(|I|+1)$-partite complex.

\begin{definition}[local spectral expander]
    Let \(X\) be a \(d\)-dimensional simplicial complex and let \(\lambda \in (0,1)\). We say that \(X\) is a \emph{\(\lambda\)-one sided local spectral expander} if for every \(s \in X^{\leq d-2}\) it holds that \(\lambda(X_s) \leq \lambda\). We say that \(X\) is a \emph{\(\lambda\)-two sided local spectral expander} if for every \(s \in X^{\leq d-2}\) it holds that \(\abs{\lambda}(X_s) \leq \lambda\).

    Here \(\lambda(X_s)\) is the normalized second largest eigenvalue of the adjacency operator of the graph \(X_s^{\leq 1}\), and  \(\abs{\lambda}(X_s)\) is the second largest eigenvalue in absolute value.
\end{definition}
We stress that this definition includes \(s= \emptyset\), which also implies that the graph \(X^{\leq 1}\) should have a small second largest eigenvalue.

We will use the `trickle down' theorem \cite{Oppenheim2018}. 
\begin{theorem}[\cite{Oppenheim2018}] \label{thm:trickle-down}
    Let \(\lambda,\tau \geq 0\). Let \(X\) be a connected simplicial complex and assume that for any vertex \(v \in X(0)\) it holds that the non-trivial eigenvalues of \(X_v\) are between \([-\tau,\lambda]\). Then the non-trivial eigenvalues of \(X\) are between  \(\left [ \frac{\tau}{1+\tau},\frac{\lambda}{1-\lambda} \right ]\).
    In particular, \(X\) is a \(\frac{\lambda}{1-\lambda}\)-one or two sided high dimensional expander (respectively).
\end{theorem}
As a corollary of reiterating this theorem, one gets the following.
\begin{corollary} \label{cor:skel-two-sided-hdx}
    Let \(X\) be a \(d\)-dimensional \(\lambda\)-one sided local spectral expander. Then \(X^{\leq k}\) is a \(\max \set{\lambda, \frac{1}{d-k+1}}\)-local spectral expander
\end{corollary}

\subsubsection*{Walks on local spectral expanders}
Let \(k,\ell,d\) be integers such that \(\ell+k\leq d-1\). The \(k,\ell\)-swap walk \(S_{k,\ell}=S_{k,\ell}(X)\) is the bipartite graph whose vertices are \(L=X(k), R=X(\ell)\) and whose edges are all \((t,s)\) such that \(t \dunion s \in X\). The probability of choosing such an edge is the probability of choosing \(u \in X(k+\ell+1)\) and then uniformly at random partitioning it to \(u=t\dunion s\). This walk has been defined and studied independently by \cite{DiksteinD2019} and by \cite{AlevFT2019}, who bounded its spectral expansion.

\begin{theorem}[\cite{DiksteinD2019,AlevFT2019}]
    Let \(X\) be a \(\lambda\)-two sided local spectral expander. Then the second largest eigenvalue of \(S_{k,\ell}(X)\) is upper bounded by \(\lambda(S_{k,\ell}(X)) \leq (k+1)(\ell+1)\lambda\).
\end{theorem}

For a \(d\)-partite complex and two disjoint sets of colors \(J_1,J_2 \subseteq [d]\) one can also define the \emph{colored swap walk} \(S_{J_1,J_2}\) as the bipartite graph whose vertices are \(L=X[J_1],R=X[J_2]\). and whose edges are all \((s,t)\) such that \(t \dunion s \in X[J_1 \dunion J_2]\). The probability of choosing this edge is \(\Prob[{X[J_1 \dunion J_2]}]{t \dunion s}\).

\begin{theorem}[\cite{DiksteinD2019}]
    Let \(X\) be a \(d\)-partite \(\lambda\)-one sided local spectral expander. Then the second largest eigenvalue of \(S_{J_1,J_2}(X)\) is upper bounded by \(\lambda(S_{J_1,J_2}(X)) \leq |J_1|\cdot |J_2|\cdot \lambda\).
\end{theorem}
We note that this theorem also make sense even when \(J_1 = \set{i}, J_2 = \set{i'}\), and the walk is between \(X[i]\) and \(X[i']\) that are subsets of the vertices.

We mention that subsequent work proved \cite{GurLL22} a tighter bound on the spectral expansion of this walk.

\subsubsection*{Coboundary and Cocycle Expansion}
For a more thorough introduction, we refer the reader to \cite{DiksteinD2023cbdry}.

Let \(X\) be a \(d\)-dimensional simplicial complex for \(d \geq 2\) and let \(\Gamma\) be any group. For \(i=-1,0\) let 
\(C^i(X,\Gamma) = \set{f:X(i) \to \Gamma}\). We sometimes identify \(C^{-1}(X,\Gamma) \cong \Gamma\). For \(i=1,2\) let
\[C^1(X,\Gamma) = \sett{f:\dir{X}(1) \to \Gamma}{f(u,v)=f(v,u)^{-1}}\]
and
\[C^2(X,\Gamma) = \sett{f:\dir{X}(i) \to \Gamma}{\forall \pi \in Sym(3), (v_0,v_1,v_2) \in \dir{X}(2) \; f(v_{\pi(0)},v_{\pi(1)},v_{\pi(2)}) = f(v_0,v_1,v_2)^{\sign(\pi)}}.\]
be the spaces of so-called {\em anti symmetric} functions on edges and triangles. For \(i=-1,0,1\) we define the coboundary operators \(\coboundary_i : C^i(X,\Gamma) \to C^{i+1}(X,\Gamma)\) by
\begin{enumerate}
    \item \(\coboundary_{-1}:C^{-1}(X,\Gamma)\to C^{0}(X,\Gamma)\) is \(\coboundary_{-1} h (v) = h(\emptyset)\).
    \item \(\coboundary_{0}:C^{0}(X,\Gamma)\to C^{1}(X,\Gamma)\) is \(\coboundary_{0} h (v,u) = h(v)h(u)^{-1}\).
    \item \(\coboundary_{1}:C^{1}(X,\Gamma)\to C^{2}(X,\Gamma)\) is \(\coboundary_{1} h (v,u,w) = h(v,u)h(u,w)h(w,v)\).
\end{enumerate}
Let \(1 = 1_i \in C^i(X,\Gamma)\) be the constant function that always outputs the identity element. It is easy to check that \(\coboundary_{i+1} \circ \coboundary_i h \equiv 1_{i+2}\) for \(i=-1,0\) and \(h \in C^{i}(X,\Gamma)\). Thus we denote by
\[Z^i(X,\Gamma) = \ker \coboundary_{i} \subseteq C^i(X,\Gamma),\]
\[B^i(X,\Gamma) = \Img \coboundary_{i-1} \subseteq C^i(X,\Gamma),\]
and have that \(B^i(X,\Gamma) \subseteq Z^i(X,\Gamma)\). 

Henceforth, when the dimension $i$ of the cochain $f$ is clear from the context we denote $\coboundary_i f$ by $\coboundary f$.

Let \(f,g \in C^i(X,\Gamma)\). Then
\begin{equation} \label{eq:def-of-dist}
    \dist(f,g) = \Prob[s \in \dir{X}(i)]{f(s) \ne g(s)}.
\end{equation}
We also denote the weight of the function \(\wt(f) = \dist(f,1)\).
\begin{definition}[Cocycle expansion] \label{def:def-of-cosyst-exp}
    Let \(X\) be a \(d\)-dimensional simplicial complex for \(d \geq 2\). Let \(\beta >0\). We say that \(X\) is a \(\beta\)-cocycle expander if for every group \(\Gamma\), and every \(f \in C^1(X,\Gamma)\) there exists some \(g \in Z^1(X,\Gamma)\) such that
    \begin{equation} \label{eq:def-of-cosyst-exp}
        \beta \dist(f,g) \leq \wt(\coboundary f).
    \end{equation}
    In this case we denote \(h^1(X) \geq \beta\).
\end{definition}
We note that in previous works this definition was referred to as \emph{cosystolic} expansion. 
\begin{definition}[Coboundary expansion] \label{def:def-of-cob-exp}
    Let \(X\) be a \(d\)-dimensional simplicial complex for \(d \geq 2\). Let \(\beta >0\). We say that \(X\) is a \(\beta\)-coboundary expander if it is a \(\beta\)-cocycle expander and in addition \(Z^1(X,\Gamma) = B^1(X,\Gamma)\) for every group \(\Gamma\).
\end{definition}

\subsubsection*{Cones}
The cone method, appearing in \cite{Gromov2010}, was used in many previous works for bounding coboundary expansion (see e.g. \cite{LubotzkyMM2016}, \cite{KozlovM2019}, \cite{KaufmanO2020}, \cite{DiksteinD2023cbdry}). The following adaptation to non-abelian cones is due to \cite{DiksteinD2023swap}. We follow similar notation and definitions.

Fix \(X\), a simplicial complex and some \(v_0 \in X(0)\). We define two symmetric relations on loops around \(v_0\):
\begin{enumerate}
    \item[(BT)]~ We say that \(P_0 \overset{(BT)}{\sim} P_1\) if \(P_i = Q_0 \circ (u,v,u) \circ Q_1\) and \(P_{1-i} = Q_0 \circ (u) \circ Q_1\) for \(i=0,1\) (i.e. going from \(u\) to \(v\) and then backtracking is trivial).
    \item[(TR)]~ We say that \(P_0 \overset{(TR)}{\sim} P_1\) if \(P_{i} = Q_0 \circ (u,v) \circ Q_1\) and \(P_{1-i} = Q_0 \circ (u,w,v) \circ Q_1\) for some triangle \(uvw \in X(2)\) and \(i=0,1\).
\end{enumerate}

Let \(\sim\) be the smallest equivalence relation that contains the above relations (i.e. the transitive closure of two relations).

We denote by \(P \sim_1 P'\) if there is a sequence of loops \((P_0=P,P_1,...,P_m=P')\) and \(j \in [m-1]\) such that:
\begin{enumerate}
    \item \(P_j \overset{(TR)}{\sim} P_{j+1}\) and
    \item For every \(j' \ne j\), \(P_{j'} \overset{(BT)}{\sim} P_{j'+1}\).
\end{enumerate}
I.e. we can get from \(P\) to \(P'\) by a sequence of equivalences, where exactly one equivalence is by \((TR)\).

Let \(P = (u_0,u_1,...,u_m)\) be a walk in \(X\). We denote by $P^{-1}$ the walk $(u_m,\ldots,u_1,u_0)$.

\begin{definition}[Decoding cone]
    A decoding cone is a triple \(C=(v_0,\set{P_u}_{u \in X(0)}, \set{T_{uw}}_{uw \in X(1)})\) such that
\begin{enumerate}
    \item \(v_0 \in X(0)\).
    \item For every \(v_0 \ne u \in X(0)\) \(P_{u}\) is a walk from \(v_0\) to \(u\). For \(u = v_0\), we take \(P_{v_0}\) to be the loop with no edges from \(v_0\).
    \item For every \(uw \in X(1)\), \(T_{uw}\) is a sequence of loops \((P_0,P_1,...,P_m)\) such that:
    \begin{enumerate}
        \item \(P_0 = P_u \circ (u,w) \circ P_w^{-1}\), 
        \item For every \(i=0,1,...,m-1\), \(P_i \sim_1 P_{i+1}\) and
        \item \(P_m\) is equivalent to the trivial loop by a sequence of \((BT)\) relations.
    \end{enumerate}
    We call \(T_{uw}\) a \emph{contraction}, and we denote $|T_{uw}| = m$.
\end{enumerate}
\end{definition}
The definition of \(T_{uw}\) depends on the direction of the edge \(uw\). We take as a convention that \(T_{wu}\) has the sequence of loops $(P_0^{-1},P_1^{-1},\ldots,P_m^{-1})$, and notice that $P_0^{-1} = (P_u \circ (u,w) \circ P_w^{-1})^{-1} = P_w \circ (w,u) \circ P_u^{-1}$. Thus for each edge it is enough to define one of \(T_{uw},T_{wu}\). The diameter of the cone is
\[diam(C) = \max_{uw \in \dir{X}(1)} \abs{T_{uw}}.\]
Equivalently, this is the maximal number of triangle relations required in the contraction of any \(T_{uw}\).

There is a direct connection between the diameter of the cone to coboundary expansion.
\begin{lemma} \label{lem:group-and-cones}
    Let \(X\) be a simplicial complex such that \(Aut(X)\) is transitive on \(k\)-faces. Suppose that there exists a cone \(C\) with diameter \(R\). Then \(X\) is a \(\frac{1}{\binom{k+1}{3}\cdot R}\)-coboundary expander.
\end{lemma}

\subsubsection*{Coboundary expansion of joins}
We record the following claim on the coboundary expansion of joins of complexes.
\begin{claim} \label{claim:cob-exp-of-join}
    Let \(Z = A_1 \vee A_2\) be a join of two complexes \(A_1,A_2\) of dimensions \(d_1,d_2\) respectively. Assume that there is a group that acts on \(Z\) so that the action on \(Z(d_1+d_2+1)\) is transitive. Then there exists a constant \(\beta = \beta(d_1,d_2)\) such that \(h^1(Z) \geq \frac{\beta}{diam(A_1)}\).
\end{claim}

\begin{proof}[Proof of \pref{claim:cob-exp-of-join}]
    By \pref{lem:group-and-cones} it is enough to show that there is a cone whose diameter in \(O(diam(A_1))\). We construct the cone as follows. Fix \(v_1^{*} \in A_1\) and \(v_2^{*} \in A_2\). Our base of the cone is \(v_1^{*}\). For every \(u \in A_2\) we take the path \(P_u=(v_1,u)\), for every \(u \in A_1\) we take the path \(P_u=(v_1,v_2,u)\).

    Now we consider an edge \(u_1 u_2 \in Z\). If \(u_1 u_2 \in A_2\) then the cycle \(C_0=C_{u_1 u_2} = P_{u_1} \circ (u_1,u_2) \circ P_{u_2}^{-1}=(v_1,u_1,u_2,v_1)\) is a triangle in \(Z\) so we can contract it in one step to \(C_1=(v_1,u_1,v_1)\) which contracts to the trivial loop using only backtrack relations. Similarly, if \(u_1 u_2 \in A_1\) then the loop we need to contract is \(C_0=C_{u_1 u_2} = (v_1,v_2,u_1,u_2,v_2,v_1)\) which can also be contracted using a single triangle \(v_2 u_1 u_2 \in Z\) only.

    \begin{figure}
        \centering
        \includegraphics[scale=0.3]{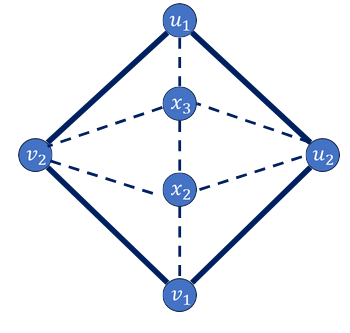}
        \caption{Contraction of the interesting case}
        \label{fig:contraction-in-join}
    \end{figure}

    The interesting case is if (say) \(u_1 \in A_1\) and \(u_2 \in A_2\) so \(C_0 = (v_1,u_2,u_1,v_2,v_1)\). In this case we take some shortest path in \(A_1\) from \(v_1\) to \(u_1\), which we denote \(Q = (v_1=x_1,x_2,x_3,\dots,x_m=u_1)\), where \(m \leq diam(A_1)\). From here we recommend to view \pref{fig:contraction-in-join} that illustrates the contraction. We define for \(i=1,2,\dots m-1\) \(C_{i}=(v_1,x_2,x_3,\dots,x_i,u_2,u_1,v_2,v_1)\), where we go from \(C_i\) to \(C_{i+1}\) via the triangle \(x_i x_{i+1} u_2 \in Z(2)\). And similarly we define \(C_m = (v_1,x_2,\dots,x_m,v_2,v_1)\) (where again we use \(x_{m-1} u_2 x_m \in Z(2)\)). Similarly, we define \(C_{m+1},C_{m+2},\dots,C_{2m}\) where 
    \(C_{m+i} = (v_1,x_2,x_3,\dots,x_{m-i},v_2,v_1)\) where we use the triangle \(x_{m-i+1} v_2 x_{m_i}\) to go from \(C_{m+i}\) to \(C_{m+i+1}\). We end up with \(C_{2m}=(v_1,v_2,v_1)\) which contracts to the trivial loop with a backtracking relation. Finally, note that we use \(2m \leq 2diam(A_1)\) so the claim follows.
\end{proof}

\subsection{Additional coboundary expansion machinery}
In this subsection we survey some of the machinery we need on coboundary expansion from previous works that we need in order to prove \pref{thm:coboundary-expansion}.
\subsubsection{Local to global coboundary expansion}
The following lemma is a local-to-global theorem that deduces cocycle expansion of a complex from coboundary expansion of the links. The first to show such a theorem were \cite{KaufmanKL2014} (for \(1\)-cochains) and \cite{EvraK2016} (for arbitrary \(i\)-cochains). We use a quantitatively stronger version, due to \cite{DiksteinD2023swap}.

\begin{lemma}[{\cite{DiksteinD2023swap}}]\label{lem:trick}
    Let \(X\) be a \(\exp(-O(i+1))\)-high dimensional expander so that every non-empty link in \(X\) is simply connected. Then
        \[h^1(X^J) \geq exp(-O(i))\cdot \min_{s\in X(i)} h^1(X_s^J). \]
\end{lemma}
\subsubsection{Colored complexes}
We also use the following color restriction lemma, that deduces coboundary expansion of a partite complex, from coboundary expansion of its color restrictions.
\begin{theorem}[{\cite{DiksteinD2023cbdry}}] \label{thm:coboundary-expansion-from-colors}
    Let \(\ell, d\) be integers so that \(3\leq \ell \leq d\) and let \(\beta,p, \lambda \in (0,1]\). Let \(\Gamma\) be some group. Let \(X\) be a \(d\)-partite simplicial complex so that
\[\Prob[F \in \binom{[d]}{\ell}]{X^F \text{ is a \(\beta\)-coboundary expander} \ve \forall s \in X(0) \; X^F_s \text{ is a \(\lambda\)-spectral expander}} \geq p.\]
Then \(X\) is a coboundary expander with \(h^{1}(X) \geq \frac{p (1-\lambda) \beta}{6e}\). Here \(e \approx 2.71\) is Euler's number.
\end{theorem}
We note that \pref{thm:coboundary-expansion-from-colors} is proven in \cite{DiksteinD2023cbdry} assuming that the spectral expansion of the graph is \(1-\beta\). This assumption is not needed in the proof; following the same steps with a separate parameter \(\lambda\) gives us a bound of \(h^{1}(X) \geq \frac{p (1-\lambda) \beta}{6e}\).
\begin{claim}[Color Swap] \label{claim:color-swap}
    For every \(\ell \geq 4\) there is a universal constant \(c_\ell > 0\) so that the following holds. Let \(I,I' \in \binom{[n]}{\ell}\) be two sets of size \(\ell\) such that their symmetric difference \(I \Delta I' = \set{i,i'}\) where \(i \in I, i' \in I'\). Let \(X\) be a \(n\)-partite \(\lambda\)-local spectral expander for \(\lambda < \frac{1}{100}\).
    \(h^1(X^I) \geq c_{\ell} h^1(X^{I'}) \min_{v \in {X[i']}}h^1(X_v^I).\)
\end{claim}

We denote by \(K_{n_1,n_2,...,n_m}\) the complete partite complex with \(n_i\) vertices on every side.

\begin{claim} \label{claim:triangle-complex}
    Let \(m \geq 5\). Let \(X\) be a \(m\)-partite simplicial complex, such that \(h^1(X) \geq \beta\). Assume that the colored swap walk between vertices to triangles is an \(\eta\)-spectral expander. Then \(Y = X \otimes K_{n_1,n_2,...,n_m}\) is a coboundary expander and \(h^1(Y) \geq (1-O(\eta)) \exp(-O(\ell)) \beta\) where \(\ell = \Abs{\sett{i \in [m]}{n_i > 1}}\).
\end{claim}

\subsubsection*{The faces complex}
\begin{definition} \label{def:face-complex}
    Let \(X\) be a \(d\)-dimensional simplicial complex. Let \(r \leq d\). We denote by \( \FX \) the simplicial complex whose vertices are \(\FX(0)=X(r)\) and whose faces are all \(\sett{\set{s_0,s_1,...,s_j}}{s_0\dunion s_1 \dunion \dots \dunion s_j \in X((j+1)(r+1)-1)}\). 
\end{definition}
It is easy to verify that this complex is \(\left ( \lfloor \frac{d+1}{r+1} \rfloor - 1 \right )\)-dimensional and that if \(X\) is a clique complex then so is \(\FX\). 
The distribution on the top-level faces of \(\FX\) is given by the following. Let \(m = \left ( \lfloor \frac{d+1}{r+1} \rfloor - 1 \right )\)
\begin{enumerate}
    \item Sample a \(d\)-face \(t=\set{v_0,v_1,\dots,v_d} \in X(d)\).
    \item Sample \(s_0,s_1,\dots,s_m \subseteq t\) such that $|s_i|=r+1$, \(s_i \cap s_j = \emptyset\) and output \(\set{s_0,s_1,\dots,s_m}\).
\end{enumerate}
\begin{definition}[Swap coboundary expansion]\label{def:swap}
    A simplicial complex $X$ is said to have $(\beta,r)$-{\em swap coboundary (cocycle) expansion} if $\FX[r]$ is a $\beta$ coboundary (cocycle) expander for $1$-cochains.
\end{definition}

It is convenient to view the faces complex as a subcomplex of the following complex.
\begin{definition}[Generalized faces complex]
    Let $X$ be a simplicial complex. The generalized faces complex, denoted $FX$, has a vertex for every $w\in X$, and a face $s=\set{w_0,\ldots,w_i}\in FX$ iff $w_i\cap w_j=\phi$ and $\dunion s := w_0\dunion w_1\dunion\cdots\dunion w_i\in X$.
\end{definition}
This complex is not pure so we do not define a measure over it.
One can readily verify that links of the faces complex correspond to faces complexes of links in the original complex. That is,
\begin{claim} \label{claim:link-of-a-faces-complex}
    Let $s\in FX$. Then $FX_s = F(X_{\cup s})$ where \(\cup s = \bigcup_{t \in s}t\). The same holds for \(\FX[r]_s=F^r(X_{\cup s})\). \(\qed\)
\end{claim}
We are therefore justified to look at generalized links of the form $FX_{\cup s}$,
\begin{definition}[Generalized Links]
    Let $w\in X$. We denote by $FX_w  = F(X_{ w})$. We also denote by $\FX[r]_w = \FX[r] \cap FX_w$. Note that this is not necessarily a proper link of $\FX$.
\end{definition}

\subsubsection*{Colors of a faces complex}    

\begin{definition}[Simplicial homomorphism]
Let $X,Y$ be two simplicial complexes. 
A map $\varphi:X\to Y$ is called a simplicial homomorphism if $\varphi:X(0)\to Y(0)$ is onto and for every $s=\set{v_0,\ldots,v_i}\in X(i)$, $\varphi(s) = \set{\varphi(v_0),\ldots,\varphi(v_i)}\in Y(i)$. 
\end{definition}

\begin{claim}
Let $\varphi:X\to Y$ be a simplicial homomorphism. Then there is a natural homomorphism $\varphi:FX\to FY$ given by 
$\varphi(\set{s_0,\ldots,s_i})=\set{\varphi(s_0),\ldots,\varphi(s_i)}$.  
\end{claim}
\begin{proof}
    Suppose $s=\set{s_0,\ldots,s_i}\in FX(i)$. By definition this means that $\dunion s \in X$ so $\varphi(\dunion s) \in Y$. But $\varphi(\dunion s) = \varphi(s_0\dunion \cdots\dunion s_i)=\varphi(s_0)\dunion \cdots\dunion \varphi(s_i)$ (because for a simplicial homomorphism $\varphi:X\to Y$ whenever $a\dunion b\in X$, $\varphi(a\dunion b) = \varphi(a)
    \dunion\varphi(b)\in Y$). Thus $\set{\varphi(s_0),\ldots,\varphi(s_i)}\in Y$.
\end{proof}

Let $Y = \Delta_n$ be the complete complex on $n$ vertices. Observe that $X$ is $n$-partite if and only if there is a homomorphism $col:X\to \Delta_n$. 
We say that a complex is $n$ colorable if its underlying graph is $n$ colorable, namely one can partition the vertices into $n$ color sets such that every edge crosses between colors.
\begin{claim}
    Let $X$ be an $n$-colorable complex. Then $\FX$ is $\binom{n}{r+1}$-colorable.
\end{claim}
We denote the set of colors of $\FX$ by $\C= \FD(0)$ (supressing $n$ from the notation). This is the set of all subsets of $[n]$ of size $r+1$.

Fix a set $J\in F\Delta$, namely $J=\set{c_1,\ldots,c_m}$ and $c_j\subset [n]$ are pairwise disjoint. Let $\FX[J] = \sett{s\in FX}{col(s)\subseteq J}$ be the sub-complex of $FX$ whose vertex colors are in $J$, so $\FX[J](0) = \bigcup_{j=1}^m X[c_j]$. 
We abuse notation in this section allowing multiple \(c_j\)'s to be empty sets. In this case \(X[c_j]\) are copies of \(\set{\emptyset}\), and every such copy of $\emptyset$ participates in every top level face of \(\FX[J]\).
The measure induced on the top level faces of \(\FX[J]\) is the one obtained by sampling \(t \in X[\cup J]\) and partitioning it to \(t= s_1 \dunion s_2 \dunion \dots \dunion s_m\) such that \(s_i \in X[c_i]\).

Finally, throughout the paper we use the following notation. Let \(J',J \subseteq \FD[]\) We write $J'\leq J$, if \(J = \set{c_1,c_2,\dots,c_m}\) and $J' = \set{c'_1,\ldots,c'_m}$ where $c'_j\subseteq c_j$.

\subsubsection*{Coboundary expansion of the faces complex}
The following two statements will be used in \pref{sec:proof-of-faces-complex-lower-bound} to prove lower bounds on swap coboundary expansion of the quotients of \(\tilde{C}_g\) defined below.
\begin{lemma}[Color restriction for faces complex] \label{lem:colorest}
    Let \(X\) be an \(n\)-partite complex for \(n \geq \d_1^5\). Let \(m \in [3,n^{0.5}/\d_1]\) and let \(\mathcal{J} \subseteq \FD{\d_1}(m)\) be set of relative size \(p\). Assume that for every \(J \in \mathcal{J}\), \(h^1(F^{\d_1}X^J) \geq \beta\) then \(h^1(\FX[\d_1]) \geq \Omega(\beta p^2)\).
\end{lemma}
The following proposition bounds the (color-restricted) faces complex using bounds on (color-restrictions of) its links. 
\begin{proposition} \label{prop:colored-exponential-decay-bound}
    Let \(X\) be a \(n\)-partite complex that is a \(\lambda\)-local spectral expander for \(\lambda \leq \frac{1}{2r^2}\). Let \(\ell \geq 5\) and let  \(J=\set{c_1,c_2, \dots ,c_\ell}\) be a set of mutually disjoint colors \(c_j \subseteq [n]\), $\card {c_j}\leq r$. Denote by \(R=\sum_{j=1}^\ell |c_j|\). Let \(\beta > 0\) and assume that for every \(I=\set{i_1,i_2,\dots,i_\ell}\) such that \(i_j \in c_j\) and every \(w \in X^{\cup J \setminus I}\), \(h^1(X_w^I) \geq \beta\). Then \(h^1(\FX[J]) \geq \beta_1^{R}\) for \(\beta_1 = \Omega_{\ell}(\beta)\). 
\end{proposition}

\subsection{Covers and quotients of simplicial complexes} \label{sec:building-background}

 \begin{definition}[Covering map]
     Let \(Y,X\) be simplicial complexes. We say that a map \(\rho:Y(0) \to X(0)\) is a covering map if the following holds.
     \begin{enumerate}
         \item \(\rho\) is a surjective homomorphism.
         \item For every \(v \in X(0)\), and \((v,i) \in \rho^{-1}(\set{v})\) it holds that \(\rho|_{Y_{(v,i)}}:Y_{(v,i)}(0) \to X_v(0)\) is an isomorphism.
     \end{enumerate}
     We often denote \(\rho:Y\to X\). We say that \(\rho\) is an \(\ell\)-cover if for every \(v \in X(0)\) it holds that \(\Abs{\rho^{-1}(\set{v})}=\ell\). If there exists such a covering map \(\rho:Y\to X\) we say that \(Y\) covers \(X\).
\end{definition}

Covers are intimately connected to the fundamental group. For a thorough definition and discussion see \cite{Surowski1984}.

In this subsection we give a description (and proof sketch) of a general technique for constructing topological spaces with a given fundamental group and other local properties. For this we use deck transformations and quotient maps. For a general and more formal setup the reader is referred to \cite[Section 1.3]{Hatcher2002}.

Recall that for an action of a group \(\Gamma\) on a set \(B\) we denote by \(\Gamma \setminus B\) the set of orbits of the action. We denote by \([v]\) the orbit of a vertex \(v \in B(0)\). 

Let \(B\) be a locally finite, connected and simply connected simplicial complex. Let \(\Gamma\) be a group that acts simplicially on \(B\). We say that the action is \emph{proper} if for every \(v \in B(0)\) and \(\gamma \in \Gamma \setminus \set{1}\), \(\dist(v,\gamma . v) \geq 4\).\footnote{This requirement implies that $v$ and $\gamma.v$ have disjoint neighborhoods. One can weaken this requirement, but we use this stricter definition to make some of the arguments simpler. For more  details, see \cite{Hatcher2002}.}

The quotient of \(B\) by \(\Gamma\) is the following simplicial complex \(X = \Gamma \setminus B\).
\[X = \sett{\set{[v_0],[v_1],\dots,[v_i]}}{\set{v_0,v_1,\dots,v_i} \in B}.\]
We denote the quotient map \(\rho:B \to X\) by \(\rho(\set{v_0,v_1,\dots,v_i}) = \set{[v_0],[v_1],\dots,[v_i]}\). By definition every face \(\tilde{s} \in B\) maps to a face \(s \in X\). 

The properties of the action promise that dimension is maintained, that is, that every \(\tilde{s} \in B(i)\) maps to a face \(s \in X(i)\). Indeed, this follows because every \(v,u \in s\) have \(\dist(u,v) = 1\) so they must be in different orbits. 

\begin{claim} \label{claim:quotient-is-cover}
    Let \(B\) be as above. Then \(\rho:B(0) \to X(0)\) is a covering map.
\end{claim}

\begin{proof}
    Fix \([v_0] \in B(0)\). We need to show that for every \(v_0 \in [v_0]\), the restriction of \(\rho\) to the link of \(v_0\) is a simplicial isomorphism between \(B_{v_0}\) and \(X_{[v_0]}\). Fix \(v_0 \in [v_0]\) as well. First, we note that indeed \(\rho(B_{v_0}) \subseteq X_{[v_0]}\): for every \(v_1 \in B_{v_0}(0)\), \(\set{v_0,v_1} \in B(1)\) so \(\set{[v_0],[v_1]} \in X(1)\) or equivalently \([v_1] \in X_{[v_0]}(0)\).
    
    Next we show that this is a bijection. Surjectivity is because if \([v_1] \in X_{[v_0]}(0)\) then \(v_1\) is a neighbor of some \(\gamma. v_0\), and in particular, \(\gamma^{-1} .v_1 \in [v_1] \cap B_{v_0}(0)\). Injectivity is due to the distance assumption: two neighbors \(v_1,v_2\) of \(v_0\) have \(\dist(v_0,v_1) \leq 2\) and therefore they must belong to different orbits.

    This is also an isomorphism between the links as complexes: It is clear that every \(\set{v_1,\dots,v_i} \in B_{v_0}\) maps to \(\set{[v_1],\dots,[v_i]} \in X_{[v_0]}\). Let us show for every \(\set{[v_1],\dots,[v_i]} \in X_{[v_0]}\) there exists a set \(\set{v_1,\dots,v_i} \in B_{v_0}\) such that \(\rho(\set{v_1,\dots,v_i}) = \set{[v_1],\dots,[v_i]}\). Indeed, if \(\set{[v_1],\dots,[v_i]} \in X_{[v_0]}\) then there is some \(v_0' \in [v_0]\) and face \(\set{v_0',v_1',\dots,v_i'} \in B(i)\). There is also an element \(\gamma\) sending \(v_0'\) to \(v_0\). Thus by setting \(v_j = \gamma v_j'\) we have that \(\set{v_1,\dots,v_i} \in B_{v_0}\) has that \(\rho(\set{v_1,\dots,v_i}) = \set{[v_1],\dots,[v_i]}\).
\end{proof}

Our agreement theorem requires as a technical property, that the have clique complexes. Let us show that if the action is \emph{proper} then the resulting complexes are clique complexes.
\begin{claim} \label{claim:quotients-are-clique-complexes}
    Let \(B\) be an infinite clique complex and let \(\Gamma\) be a group acting on \(B\). Then \(X = \Gamma \setminus B\) is a clique complex.
\end{claim}

\begin{proof}[Proof of \pref{claim:quotients-are-clique-complexes}]
    Let \(\set{[v_0],[v_1],\dots,[v_i]} \subseteq X(0)\) be a clique. Let \(v_0 \in [v_0]\) be a representative and let \(v_1,v_2,\dots,v_i \in X_{v_0}(0)\) be the representatives of \([v_1],[v_2],\dots,[v_i]\) respectively. We claim that properness implies that \(\set{v_0,v_1,\dots,v_i}\) are a clique in \(B\). If this holds then \(\set{v_0,v_1,\dots,v_i} \in B(i)\) and therefore the set \(\set{[v_0],[v_1],\dots,[v_i]} \in X(i)\). Assume towards contradiction that, without loss of generality, \(\set{v_1,v_2} \notin B(1)\). We note that \(\set{[v_1],[v_2]} \in X(1)\) so there exists \(v_1' \ne v_1\) such that \(\set{v_1',v_2} \in B(1)\) and \(v_1' \in [v_1]\). In this case we observe that \((v_1',v_2,v_0,v_1)\) is a path in \(B\) and therefore there are two distinct vertices of distance \(\dist(v_1,v_1') < 4\) in the same orbit, contradicting properness.
\end{proof}

Recall that covers of a complex are connected to its fundamental group.
\begin{fact} \label{fact:connected-covers-and-subgroups} \label{fact:no-small-index-implies-no-small-covers}
    Let \(X\) be a connected simplicial complex and locally finite simplicial complex. Let \(\pi_1(X,v_0)\) be the fundamental group of \(X\) (with \(v_0 \in X(0)\) an arbitrary vertex). Then \(X\) has a connected \(\ell\)-cover if and only if \(\pi_1(X,v_0)\) has a subgroup of index \(\ell\). 
\end{fact}
The main fact we use is the following.

\begin{theorem}[{\cite[Section 1.3]{Hatcher2002}}] \label{thm:fundamental-group-of-proper-action-quotient}
    Let \(B\) be a simply connected simplicial complex. Let \(\Gamma\) be a group acting simplicially and properly on \(B\). Then the fundamental group of \(\Gamma \setminus B\) is \(\Gamma\).
\end{theorem}

To conclude this subsection, let us the following claim that reduces finding groups that act properly on \(B\), to finding groups that act freely on \(B\).
\begin{claim} \label{claim:eventually-no-close-vertices-in-orbit}
    Let \(B\) be a locally finite simplicial complex. Let \(\Gamma\) be a group acting simplicially and freely on \(B\) such that \(\Gamma \setminus B\) is finite. Let \(\Gamma = \Gamma_1 \geq \Gamma_2 \geq \Gamma_3 \dots\) be a chain of subgroups such that \(\bigcap_{i=1}^\infty \Gamma_i = \set{1}\). Then for every \(r>0\) there exists \(i_0\) such that for every \(i>i_0\), every \(\gamma \in \Gamma_i \setminus \set{1}\) and every \(v \in B(0)\) it holds that \(\dist(v, \gamma . v) >r\).
\end{claim}
For \(r=4\) this shows that we can find a subsequence of groups that act properly on \(B\).

\begin{proof}[Proof of \pref{claim:eventually-no-close-vertices-in-orbit}]
    Fix \(r > 1\). An element \(\gamma \in \Gamma \setminus \set{1}\) is `bad' if there exists some \(v \in B(0)\) such that \(\dist(v,\gamma.v) \leq r\). We will show that there is only a finite number of `bad' elements. If this holds then by the fact that the intersection is trivial, for every `bad' every element \(\gamma\) there is a maximal \(i_\gamma\) such that \(\gamma \in \Gamma_{i_\gamma}\), since it is not in  \(\bigcap_{i=1}^\infty \Gamma_i\). Taking \(i_0 = \max \sett{i_\gamma}{\gamma \text{ is bad}}\) will give the result.
    
    Let \(F \subseteq B(0)\) be a fundamental domain of \(\Gamma\) (namely, a set of representatives of the orbits of \(B(0)\) under $\Gamma$).  If for every \(v \in F\), \(\dist(v,\gamma.v) > r\), then the same holds for every \(v \in B(0)\). Indeed, assume towards contradiction that there exists some \(v \in B(0)\) and \(\gamma \in \Gamma\) such that \(\dist(v, \gamma . v) \leq r\). Let \(\gamma' \in \Gamma\) such that \(\gamma' .v \in F\). As \(\gamma'\) preserves distances, \(\dist(\gamma' . v, \gamma'. \gamma . v) \leq r\). We observe that \(\gamma'. \gamma = \gamma'. \gamma \gamma'^{-1} \gamma'\) so with \(u=\gamma'.v\) we have that \(u \in F\) and for \(\gamma'' = \gamma'. \gamma \gamma'^{-1} \ne 1\), \(\dist(u,\gamma''.u) \leq r\).

    Thus we need to show that there is a finite number of elements \(\gamma \in \Gamma \setminus \set{1}\) such that \(\dist(v,\gamma .v)\leq r\) for some \(v \in F\). By assumption that \(|F|=|\Gamma \setminus B(0)| < \infty\). Denote by \(\tilde{F}\) the union of \(r\)-balls around the vertices in \(F\). It suffices to show that the set of elements of $\Gamma$ that send some \(v \in \tilde{F}\) to \(v' \in \tilde{F}\) is finite. By local finiteness of \(B\), \(\tilde{F}\) is also finite. By freeness of the action, there are at most \(|\tilde{F}|^2\) elements in \(\Gamma\) sending some \(v \in \tilde{F}\) to \(v' \in \tilde{F}\). Otherwise, by the pigeonhole principle there is a pair \((v,v')\) and two distinct \(\gamma_1,\gamma_2\) that send \(\gamma_i . v = v'\) contradicting freeness. 
\end{proof}

\section{Spherical and affine buildings} \label{sec:buildings}
In this section we give a brief survey of the buildings we deal with in this paper. Most of the material in this section is well known, and we bring it here just to introduce the notation and to stay self contained.

\subsection{The SL case (type A)}

\begin{definition}[Spherical building - type A]\label{def:typeA}
We denote by \(A_n=A_n(\mathbb{F}_p)\) the $n$ partite spherical building associated with \(SL_{n+1}(\mathbb{F}_p)\). This is an $(n-1)$-dimensional complex whose vertices are all non-trivial linear subspaces of \(\mathbb{F}_p^{n+1}\). A \(t\)-face in this complex is a set of subspaces \(\set{v_1,\ldots,v_{t}}\) so that \(v_1 \subset v_2 \subset \dots \subset v_{t}\).
\end{definition}

Previous work has shown that this complex is both a local spectral expander, and a (swap) coboundary expander.

\begin{claim}[\cite{EvraK2016}, \cite{DiksteinD2019} for the color restriction] \label{claim:spherical-building-hdxness}
    The spherical building \(A_n\) is a \(O(\frac{1}{\sqrt{p}})\)-one sided local spectral expander. Moreover, \(A_n^{\leq k}\) is a \(\max \set{O(\frac{1}{\sqrt{p}}), \frac{1}{d-k}}\)-two sided local spectral expander. The same holds for \(A_n^J\) for all subsets \(J \subseteq [d]\).
\end{claim}

The following claim is elementary, it is proved in full in \cite{DiksteinD2023swap}.
\begin{claim}[\cite{DiksteinD2023swap}] \label{claim:diam-of-spherical-building}
    Let \(A=A_n\) be an \(SL\)-spherical building. Let \(I \subseteq [n], |I| \geq 2\). Then for every flag \(w\) including the empty flag, \(diam(A_w^I) = O(\frac{\max I}{\max I - \min I})\). In particular, for every \(i_1 > i_0\) in \(I\), 
\(diam(A_w^I) = O(\frac{i_1}{i_1-i_0})\).
\end{claim}
The `in particular' part follows from standard calculus.

We will also need the following bound on color restrictions of the spherical building of type \(A\).
\begin{lemma}[\cite{DiksteinD2023swap}] \label{lem:general-case-subspace-complex-sl}
    Let \(I=\set{i_0<i_1<i_2<i_3}\) such that \(i_3 > 21\) and such that \(i_j - i_{j-1} \geq 3\). Let \(A\) be an \(SL\)-spherical building. Then for every flag \(w\) including the empty flag,
    \[h^1(A_w^I) \geq \exp \left (-O\left ( \log \left (\frac{i_3}{i_1-i_0} \right ) \cdot \log\left (\frac{i_3}{i_1}\right ) \right ) \right ).\]
\end{lemma}

\subsection{The symplectic case (type C)}
In this subsection we use definitions from symplectic geometry, see e.g. \cite[Chapter III]{Artin1957} for an introduction to the subject.
Let \(g \in \N\), and let \(V = \mathbb{F}^{2g}\). For \(x,y \in (\mathbb{F}_p)^g\) we denote \((x,y) \in V\) the vector whose first \(g\)-coordinates are \(x\) and the last \(g\) coordinates are \(y\). Let \(\iprod{\cdot,\cdot}: V \to \mathbb{F}_p\) be the following skew-symmetric bilnear form.
\begin{equation} \label{eq:symplectic-form}
    \iprod{(x,y),(z,w)} = x \cdot w - y \cdot z
\end{equation}
where \(a \cdot b = \sum_{i=1}^g a_i b_i\) is the usual inner product over \(\mathbb{F}_p^g\).

A subspace \(v \subseteq V\) is called \emph{isotropic} if for every \(x_1,x_2 \in v\), \(\iprod{x_1,x_2} = 0\). By Witt's theorem (cf. \cite[Theorem 3.10]{Artin1957})  all maximal isotropic spaces have the same dimension. For this form, the maximal dimension is \(g\), and a maximal isotropic subspace is \(Span (\set{(e_i,0): i=1,2,\dots g})\) where \(e_i \in \mathbb{F}_p^g\) are the standard basis vectors.

\begin{definition}[Spherical building - type C]\label{def:typeC}
    For $g\in\N$, the symplectic spherical building over \(\mathbb{F}_p\) denoted \(C_g\) is the $g$-partite $(g-1)$-dimensional simplicial complex whose  vertices are all non-trivial isotropic subspaces. Its faces are all \emph{flags} of isotropic subspaces. That is \(\set{v_0,v_1,\dots,v_k} \in C(k)\) if \(v_0 \subseteq v_1 \subseteq \dots \subseteq v_k\) and every \(v_i\) is isotropic. By the fact that all maximal isotropic subspaces have the same dimension, it follows that \(C\) is $g$-partite and we denote \(C[i] = \sett{v \in C}{dim(v)=i}\).
\end{definition}
We comment that \(C_1 = A_1\) but the buildings differ for larger \(g >1\).
Let us study the structure of links in the spherical building of type $C$. 
We note that the function \(x \mapsto \iprod{x,\cdot}\) gives an isomorphism between \(V\) and \(V^*\), where $V^*$ is the space of linear forms on \(V\). A bilinear form with this property is called non-degenerate. 
\begin{claim}[{\cite[Theorem 3.5]{Artin1957}}] \label{claim:orthogonal-spaces-in-symplectic-geometry}
    Let \(\iprod{\cdot,\cdot}\) be a non-degenerate bilinear form. Let $v\subset V$ be a linear subspace. \(v^{\bot} = \sett{x \in V}{\forall y \in v, \iprod{x,y}=0}\). Then \(dim(v^{\bot})=dim(V) - dim(v)\) and \((v^{\bot})^{\bot} = v\). \(\qed\)
\end{claim}

Let \(v\) be an isotropic subspace. The fact that \(v\) is isotropic is the same as saying that \(v \subseteq v^{\bot}\). Let \(v' = v^{\bot}/v\) be the quotient space. Let us define the following skew-symmetric bilinear form \(\iprod{\cdot,\cdot}_{v'}:v' \to \mathbb{F}_p\) by
\[\iprod{[x],[x']}_{v'}=\iprod{x,x'}\]
for any two \([x],[x'] \in v'\).

\begin{claim} \label{claim:local-form}
    The form \(\iprod{\cdot,\cdot}_{v'}\) is a well-defined skew-symmetric bilinear form. Moreover, it is non-degenerate.
\end{claim}

\begin{proof}[Proof of \pref{claim:local-form}]
    We need to show that the definition does not depend on choice of representatives. Namely, for every \(x_1,x_2 \in [u_1]\) and \(x'\in [u']\) it holds that \(\iprod{x_1,x'} = \iprod{x_2,x'}\) (we need to show it also for representatives on the right but this just holds from skew-symmetry). Indeed, note that \(x_1-x_2 \in v\) hence \(\iprod{x_1-x_2, x'} = 0\) as \(x' \in v^{\bot}\).

    Bilinearity just follows from the fact that the quotient map is linear. Let us show that this form is non-degenerate. Let \([x] \in v'\) be so that \(\forall [x'] \in v'\) it holds that \(\iprod{[x],[x']}_{v'}=0\). By definition this implies that \(\forall x' \in v^{\bot}\), \(\iprod{x,x'}=0\). Thus \(x \in (v^{\bot})^{\bot}\). By \pref{claim:orthogonal-spaces-in-symplectic-geometry} \((v^{\bot})^{\bot} = v\) so \(x \in v\) or equivalently \([x] = 0\).
\end{proof}

With \pref{claim:local-form} we can understand the structure of the subset of isotropic subspaces that contains a fixed subspace \(v \in C(0)\). 
\begin{proposition} \label{prop:isomorphism-of-symplectic-link}
    Let \(t \leq g-1\). Let \(v \in C_g[t]\). Let \(\rho: v^{\bot} \to v'\) be the quotient map taking \(x \in V\) to \([x]=x+v \in v'\). 
    Then \(\rho|_{v^\bot}\) induces an isomorphism between isotropic subspaces that contain \(v\) with respect to \(\iprod{\cdot,\cdot}\), to isotropic subspaces in \(v'\) with respect to \(\iprod{\cdot,\cdot}_{v'}\). This isomorphism takes subspaces of dimension \(t+i\) to subspaces of dimension \(i\).
\end{proposition}

As a corollary we get a concrete description of the link of \(v\).
\begin{corollary} \label{cor:link-of-symplectic-link}
\begin{enumerate}
    \item Let \(v \in C_g[t]\). Then the link of \(v\) is isomorphic to \(A_{t-1} \vee C_{g-t}\).
    \item In particular, for every set \(I \subseteq \set{t+1,t+2,\dots,g}\), the complex \((C_g)_v^I\) is isomorphic to \(C_{g-t}^I\).
    \item Let \(w =\set{v_0 \leq v_1 \leq \dots \leq v_i} \in C(i)\). Then the link of \(w\) is isomorphic to \(A_{j_0} \vee A_{j_1} \vee \dots \vee A_{j_t} \vee C_{j_{t+1}}\). Here \(j_0 = dim(v_0) - 1\), \(j_{t+1} = g-dim(v_i)\) and for all \(i=1,2,\dots,t\), \(j_i = dim(v_i) - dim(v_{i-1}) - 1\). 
\end{enumerate}
\end{corollary}

\begin{proof}[Proof of \pref{prop:isomorphism-of-symplectic-link}]
    Any isotropic subspace \(u \supseteq v\) is in \(v^\bot\). 
    It is well known by the correspondence theorem that the poset of subspaces of \(v^\bot\) containing \(v\) and the subspaces in \(v'=v^\bot/v\) are isomorphic and that this isomorphism sends subspaces of dimension \(t+i\) to subspaces of dimension \(i\). Thus we need to show that a subspace \(v \subseteq u \subseteq v^\bot\) is isotropic if and only if \(\rho(u)\) is isotropic (with respect to the respective inner products).
    
    Indeed \(u \supseteq v\) is isotropic if and only if for every \(x,y \in u\), \(\iprod{x,y} = 0\). By definition \(\rho(u) = \sett{[x] \in v'}{x \in u}\) and \(\iprod{[x],[y]}_{v'} = \iprod{x,y}\) so \(u\) is isotropic if and only if \(\rho(u)\) is.
\end{proof}

\begin{proof}[Proof of \pref{cor:link-of-symplectic-link}]   
    Let \(I_1 = \set{0,1,\dots,t-1}, I_2 = \set{t+1,t+2,\dots,g}\). We first note that \((C_g)_v = (C_g)_v^{I_1} \vee (C_g)_v^{I_2}\) since choosing a top level face in the link corresponds to choosing a flag contained in \(v\) and (independently) a flag that contains \(v\) and taking the union of the two flags. Clearly any subspace contained in \(v\) is itself isotropic so clearly \((C_g)_v^{I_1} \cong A_{t-1}\). Moreover, by \pref{prop:isomorphism-of-symplectic-link}, \((C_g)_v^{I_2} \cong C_{g-t-1}\). This proves the first and second items.

    The third item follows from the inducting over the first item, using the fact that \((C_g)_w = ((\dots (C_g)_{v_0})_{v_1})\dots)_{v_i}\)
\end{proof}

In \pref{sec:local-spectral-expansion-C} we prove that the building is a local-spectral expander.
\begin{lemma}\label{lem:symplectic-local-spectral-expansion}
    There exists \(c > 1\) such that the following holds for every  \(g\in\N\) and prime power \(p\) such that for every \(w \in C_g\) and \(i,j \notin col(w)\), the bipartite graph \(((C_g)_w[i],(C_g)_w[j])\) is an \(\frac{c}{\sqrt{p}}\)-one sided spectral expander. Moreover, \(C_g\) is a \(\frac{c}{\sqrt{p}}\)-one sided local spectral expander.
\end{lemma}

\subsection{The affine symplectic building}
In this section we define and describe the infinite simplicial complex known as the affine symplectic building \(\tilde{C}_g=\tilde{C}_g(\mathbb{Q}_p)\). This is complex is a close relative of the complex of type \(\tilde{A}_g\) which was used in \cite{LubotzkySV2005b} to construct the well known Ramanujan complexes. It is a well-studied complex, so we only give a brief introduction here. For proofs and a more in depth discussion see \cite{AbramenkoB2008} or \cite{Weiss2008}.

In particular, we will see that this is a \((g+1)\)-partite complex and that \(Sp(2g,\mathbb{Q}_p)\) acts transitively on the top level faces in a color preserving manner. We will also give the following description of its links.
\begin{proposition} \label{prop:links-of-affine-symplectic-building}
    Let \(w \in \tilde{C}_g(i)\). Then the link of \(w\) is a join of at most \(i+2\) spherical buildings of type \(A\) or \(C\).
\end{proposition}
We call complexes with such a property \emph{symplectic-like complexes}.

Let \(p\) be any prime and let \(\mathbb{Z}_p = \sett{\sum_{j=0}^\infty a_j p^j}{a_j \in \set{0,1,\dots,p-1}}\) and \(\mathbb{Q}_p = \mathbb{Z}_p[\frac{1}{p}]\) be the \(p\)-adic integers and \(p\)-adic numbers, respectively.

Let \(V = \mathbb{Q}_p^{2g} = \sp(e_1,e_2,\dots,e_g,f_1,f_2,\dots,f_g)\), and let \(\iprod{\cdot,\cdot}\) be the skew symmetric non-degenerate bilinear form defined by the following relations:
\begin{equation}\label{eq:asymmetric-bilinear-form}
    \begin{cases}
        \iprod{e_i,f_j}=-\iprod{f_j,e_i} = \begin{cases}
    1 & i=j \\ 0 & i \ne j
\end{cases}, \\
        \iprod{e_i,e_j}=\iprod{f_i,f_j}=0
    \end{cases}.
\end{equation} 
The group \(Sp(2g,\mathbb{Q}_p)\) is defined with respect to this bilinear form to be all linear operators preserving the bilinear form.

A \(\mathbb{Z}_p\)-lattice is a set of the form \(L = \sett{\sum_{j=1}^{2g} \alpha_j b_j}{\alpha_j \in \mathbb{Z}_p} = \sp_{\mathbb{Z}_p}(b_1,b_2,\dots,b_{2g})\) where \(\set{b_1,b_2,\dots,b_{2g}}\) is a basis for \(V\) (i.e.\ \(\sp_{\mathbb{Q}_p}(b_1,b_2,\dots,b_{2g}) = V\)). The standard lattice is \(L_{std} = \sp_{\mathbb{Z}_p}(e_1,e_2,\dots,e_{g},f_1,f_2,\dots,f_g)\). A lattice is \emph{primitive} if there exists some \(A \in Sp(2g,\mathbb{Q}_p)\) such that \(AL = L_{std}\). An alternative and equivalent definition for a primitive lattice, is that \(L= \sp_{\mathbb{Z}_p}(e_1',e_2',\dots,e_n',f_1',f_2',\dots,f_n')\) such that \eqref{eq:asymmetric-bilinear-form} holds (for the \(\set{e_i'}\cup \set{f_j'}\) instead of \(\set{e_i}\cup \set{f_j}\) respectively).

Recall that \(\mathbb{Z}_p/p\mathbb{Z}_p \cong \mathbb{F}_p\) by \(\sum_{j=0}^\infty a_j p^j \mapsto a_0\), which we also write as \(z \mapsto z \text{ (mod p)}\). This induces an isomorphism of \(L/pL \cong \mathbb{F}_p^{2g}\) in the natural way. We can endow \(L/pL\) with a bilinear form given by
\begin{equation} \label{eq:primitive-lattice-bilinear-form}
    \iprod{u_1 + pL,u_2 + pL} := \iprod{u_1,u_2} \text{ (mod p)}.    
\end{equation}

One may verify that this does not depend on the choice of representatives and that when \(L\) is primitive then this is a a skew symmetric non-degenerate bilinear form.

We also define an equivalence class over the set of lattices, where \(L \sim L'\) if \(L=p^j L'\) for some \(j \in \mathbb{Z}\). We denote an equivalence class of lattices by \([L]\). We say a lattice class \([L]\) is primitive if there exists a primitive representative \(L \in [L]\).

Obviously the group \(Sp(2g,\mathbb{Q}_p)\) acts on lattices in the natural way. This action respects the lattice class, i.e. \(L \sim L'\) if and only if for every \(A \in Sp(2g,\mathbb{Q}_p)\), \(AL \sim AL'\). Therefore the action of the group is well defined on lattice classes.

We are ready to define \(\tilde{C}_g\). 
\begin{definition}
    The \(g\)-dimensional affine building over \(\mathbb{Q}_p\) is the following simplicial complex. The vertices of \(\tilde{C}_g(0)\) are all equivalence classes of lattices \([L]\) so that:
    \begin{enumerate}
        \item There exists \(L \in [L]\) and some primitive lattice \(L_0\) such that \(pL_0 \subseteq L \subseteq L_0\) and
        \item The subspace \(L/pL_0\) is isotropic inside \(L_0/pL_0\), with the bilinear form defined in \eqref{eq:primitive-lattice-bilinear-form}.
    \end{enumerate}

    The top-level faces are all \(\set{[L_0],[L_1],\dots,[L_g]} \in \tilde{C}_g(g)\) such that there exist representatives \(L_0 \in [L_0],L_1 \in [L_1], \dots, L_g \in [L_g]\) such that \(L_0\) is primitive and 
        \begin{equation} \label{eq:isotropic-flag-Qp}
            pL_0 \subsetneq L_1 \subsetneq L_2 \subsetneq \dots \subsetneq L_g \subsetneq L_0
        \end{equation}
    and all the \(L_i/pL_0\) are isotropic.
\end{definition}

We shall not prove all facts we need on the symplectic building, but in \pref{app:building-proofs} we show a few properties of \(\tilde{C}_g\) to gain some intuition on it. In particular, in \pref{app:building-proofs} we prove the following claim.

\begin{claim} \torestate{\label{claim:basic-properties-of-affine-symplectic-building} ~
    \begin{enumerate}
        \item The group \(Sp(2g,\mathbb{Q}_p)\) acts simplicially on \(\tilde{C}_g\).
        \item The action is transitive on \(\tilde{C}_g(g)\).
        \item The complex \(\tilde{C}_g\) is \((g+1)\)-partite. The color of a lattice class \([L]\) is \(i\) such that \([L]\) is in \(i\)-th place in \eqref{eq:isotropic-flag-Qp}.\footnote{In particular, every vertex is contained in some top-level face and the place of \([L]\) does not depend on a choice of face or representative.}
        \item Let \(v=[L_0]\) be primitive. Then \((\tilde{C}_g)_v \cong C_g\), the \(g\)-partite symplectic spherical building.
    \end{enumerate}}
\end{claim}

With \pref{claim:basic-properties-of-affine-symplectic-building} in hand, we comment that a posteriori we can also define \(\tilde{C}_g[i]\) differently. Let
\[L_i^* = \sp((e_1,e_2,\dots,e_i) \cup (pe_{i+1},pe_{i+2},\dots,pe_g) \cup (pf_1,pf_2,\dots,pf_g)).\]
The group \(Sp(2g,\mathbb{Q}_p)\) acts transitively on \(\tilde{C}_g(g)\), so it also acts transitively on \(\tilde{C}_g[i]\) (see \pref{app:building-proofs}). 

Therefore, \(\tilde{C}_g[i] = Orbit([L_i^*])\). Equivalently, these are all lattice classes \([L_i]\) that have a representative 
\[L_i = \sp(e_1',e_2',\dots,e_i',h_{i+1},h_{i+2},\dots, h_{g},f_1',f_2',\dots,f_g')\]
such that the skew symmetric bilinear form acts the same as in \(L_i^*\) with respect to the above basis. That is,
\[\iprod{e_i,h_j}=\iprod{e_i,e_j}=\iprod{h_i,h_j}=\iprod{f_i,f_j} = 0,\]
\[\iprod{e_i,f_j} = \begin{cases}
    p & i=j \\
    0 & i \ne j
\end{cases},\]
and
\[\iprod{h_i,f_j} = \begin{cases}
    p^2 & i=j \\
    0 & i \ne j
\end{cases}.\]

The link structure of every \([L_i] \in \tilde{C}_g[i]\) is the same because of the transitive action of \(Sp(2g,\mathbb{Q}_p)\). The following fact gives a complete characterization of the links of vertices in this complex.
\begin{fact} \label{fact:links-of-affine-spherical-building}
    Every link of \(v \in \tilde{C}[i]\) is isomorphic to \(C_{i-1}(\mathbb{F}_p) \vee C_{g-i}(\mathbb{F}_p)\).\footnote{For \(i=0\) or \(i=g\) one of the complexes is empty, i.e. this is a single symplectic spherical building of dimension \(g\).}
\end{fact}

This fact together with \pref{cor:link-of-symplectic-link} that shows how links of faces in \(C_{g-i}\) decompose into joins of smaller spherical buildings, immediately imply \pref{prop:links-of-affine-symplectic-building}.

Using \pref{prop:links-of-affine-symplectic-building} we also prove the following. 
\begin{theorem} \label{thm:expansion-of-building-quotients}
    For every \(g \in \N\), a prime \(p\) and a quotient \(X\) of \(\tilde{C}_g(\mathbb{Q}_p)\), the complex \(X\) is a \(O(\frac{1}{\sqrt{p}} )\)-one sided local spectral expander. The skeleton \(X^{\leq k}\) is a \(\max\set{O(\frac{1}{\sqrt{p}}),\frac{1}{d-k+1}}\)-two sided local spectral expander. The same holds true for any symplectic-like complex.
\end{theorem}

We mention that the phenomenon observed in \pref{prop:links-of-affine-symplectic-building} is more general. In fact, one can read the structure of the links of an affine building via its associated Coxeter-Dynkin diagram. Every building has an associated Coxeter-Dynkin diagram (the type of the building is named after it). A link of a vertex of color \(i\) is a building whose diagram is obtained by removing the \(i\)-th vertex from the original Coxeter-Dynkin diagram. If the removal of a vertex disconnects the diagram, this corresponds to the link being a join of complexes whose respective diagrams are the different connected components. For more details see \cite{AbramenkoB2008} or \cite{Weiss2008}.

Towards taking quotients of \(\tilde{C}_g\) we mention the following fact.
\begin{fact}\label{fact:basic-symplectic-building}
    For every \(g \geq 1\), the complex \(\tilde{C}_g\) is a connected clique complex. Moreover, it is simply connected and contractible.
\end{fact}

Finally, we need a criterion for when subgroups of \(Sp(2g,\mathbb{Q}_p)\) act freely on \(\tilde{C}_g\). The group \(Sp(2g,\mathbb{Q}_p)\) comes with a topology induced by the topology of \(\mathbb{Q}_p\). For this we need the following fact. 
\begin{fact} \label{fact:compact-stabilizers}
    Let \(v \in \tilde{C}_g(0)\). Then its stabilizer \(Stab(v) \subseteq Sp(2g,\mathbb{Q}_p)\) is an open compact subgroup.
\end{fact}

We 
note that if \([L] \in \tilde{C}_g[0]\) then the stabilizer is conjugate to \(Stab([L_{std}]) = Sp(2g,\mathbb{Z}_p)\), so in this case the fact is clear. For other colors the stabilizer is commensurable to \(Sp(2g,\mathbb{Z}_p)\), hence also open compact.
\begin{corollary} \label{cor:discrete-torsion-free-acts-freely}
    Let \(\Gamma \leq Sp(2g,\mathbb{Q}_p)\) be a subgroup.
    \begin{enumerate}
        \item If \(\Gamma\) is discrete and cocompact then \(\Gamma \setminus \tilde{C}_g(0)\) is finite.    
        \item If \(\Gamma\) is discrete and torsion free then the action of \(\Gamma\) is free on the vertices of  \(\tilde{C}_g\).    
    \end{enumerate}
\end{corollary}
\begin{proof}[Proof of \pref{cor:discrete-torsion-free-acts-freely}]
    We start with the first item. It is enough to show that \(\Gamma \setminus \tilde{C}_g[i]\) is finite for every color \(i\). By \pref{claim:basic-properties-of-affine-symplectic-building} the action of \(Sp(2g,\mathbb{Q}_p)\) is transitive on top-level faces. This shows it also acts transitively on every \(\tilde{C}_g[i]\): given \(v_1,v_2 \in \tilde{C}_g[i]\), find top level faces \(s_1\ni v_1,s_2 \ni v_2\). There is an element \(A \in Sp(2g,\mathbb{Q}_p)\) so that \(A s_1 = s_2\). The action preserves colors hence \(A v_1 = v_2\). Thus the action of \(\Gamma\) on \(\tilde{C}_g[i]\) is the same as the left multiplication action of \(\Gamma\) on the right cosets \(Sp(2g,\mathbb{Q}_p)/Stab(v)\) for any \(v \in \tilde{C}_g[i]\). By \pref{fact:compact-stabilizers} the stabilizer is open, and therefore \(Sp(2g,\mathbb{Q}_p)/Stab(v)\) is discrete. On the other hand, the orbits of this action are \(\Gamma \setminus Sp(2g,\mathbb{Q}_p)/Stab(v)\). The space \(\Gamma \setminus Sp(2g,\mathbb{Q}_p)\) is compact, so this quotient \(\Gamma \setminus Sp(2g,\mathbb{Q}_p)/Stab(v)\) is compact and discrete, and therefore finite.

    We now show the second item. Let \(\gamma \in \Gamma\) and assume that there exists some \(v \in \tilde{C}_g(0)\) so that \(\gamma \in Stab(v)\); let us show that \(\gamma = 1\). Obviously for every \(m \in \N\), \(\gamma^m \in Stab(v)\). By \pref{fact:compact-stabilizers}, \(Stab(v)\) is compact, therefore the sequence \(\gamma^m\) has a converging subsequence. By assumption \(\Gamma\) is discrete, thus any converging sequence must be constant after some finite point. In particular, there exists \(m_2 > m_1\) such that \(\gamma^{m_2}=\gamma^{m_1}\), or equivalently \(\gamma^{m_2-m_1} = 1\). By torsion freeness of \(\Gamma\) this implies that \(\gamma = 1\).
\end{proof}

%% file: agreement.tex
\section{Proof of the Agreement Theorem} \label{sec:mainproof}
In this section we prove \pref{thm:main}, in a more general setup. In fact we prove a more general version of this theorem, \pref{thm:main-technical}. We first repeat some definitions from \cite{DiksteinD2023agr}. We begin with the definitions required for our agreement tests, then we define ``suitable complexes'' which are needed for stating the main theorem of \cite{DiksteinD2023agr}. Finally we prove that quotients of \(\tilde{C}_g\) satisfy these requirements and use \pref{thm:agrcover} to conclude our main theorem.

\subsection{Agreement Tests Definitions} \label{sec:agreement-tests}
Let \(k<d\) and let \(X\) be a \(d\)-dimensional simplicial complex. Let \(\Sigma\) be some fixed alphabet and suppose we have an ensemble of functions \(\mathcal{F} = \sett{f_s:s\to \Sigma}{s \in X(k)}\).

A two-query agreement test is a distribution \(\mathcal{D}\) over pairs \(s_1,s_2 \in X(k)\). The agreement of an ensemble is
\begin{equation}
    \agr_{\mathcal{D}}(\mathcal{F}) = \Prob[s_1,s_2\sim \mathcal{D}]{f_{s_1} = f_{s_2}}.
\end{equation}
When writing \(f_{s_1} = f_{s_2}\) we mean that \(f_{s_1}(v)=f_{s_2}(v)\) for every \(v \in s_1 \cap s_2\).

More generally, a \(q\)-ary agreement test is a distribution of \(s_1,s_2,...,s_q \in X(k)\), where the agreement of the ensemble
\begin{equation}
    \agr_{\mathcal{D}}(\mathcal{F}) = \Prob[s_1,s_2,\dots s_q \sim \mathcal{D}]{\forall i,j \; f_{s_i} = f_{s_j}}.
\end{equation}
Let \(\Delta_d(k)\) be the \(k\)-dimensional complete complex over \(d\) vertices. 
\begin{definition}[Extension]
    Let \(\mathcal{D}\) be a symmetric\footnote{i.e. for every permutation \(\pi:[d] \to [d]\),  \(\prob{s_1,s_2,...,s_q} = \prob{\pi(s_1),\pi(s_2),...,\pi(s_q)}.\)} \(q\)-ary agreement test on \(\Delta_d(k)\). Let \(X\) be any \(d\)-dimensional simplicial complex. We define the \emph{extension} \(\mathcal{D}_X\) of \(\mathcal{D}\) to an agreement test on \(X\), as follows:
\begin{enumerate}
    \item Sample \(t \in X(d)\).
    \item Query \(s_1,s_2,...,s_q \subseteq t\) according to \(\Delta_d(k)\).
\end{enumerate}
\end{definition}
We note that by the symmetry of \(\mathcal{D}\) the second step doesn't depend on the way we identify the vertices of \(t\) with the vertices of \(\Delta_d(k)\). Here are some classical examples.
\begin{definition}[Two-query \(V\)-test]
    Let \(d = 2k-\sqrt{k+1}+1\).
    \begin{enumerate}
        \item Sample some \(t \in X(d)\).
        \item Sample uniformly \(s_1,s_2 \in X(k)\) such that \(s_1, s_2 \subseteq t\), conditioned on \(\abs{s_1 \cap s_2} = \sqrt{k+1}\).
    \end{enumerate}
\end{definition}

\begin{definition}[Three-query \(Z\)-test]
    Let \(d=3k-2\sqrt{k+1}+2)\).
    \begin{enumerate}
        \item Sample some \(t \in X(d)\).
        \item Sample three \(s_1,s_2,s_3 \in X(k)\) such that \(s_1, s_2, s_3 \subseteq t\), conditioned on \(\abs{s_1 \cap s_2}, \abs{s_2 \cap s_3} = \sqrt{k+1}\) and \(s_1 \cap s_3 = \emptyset\).
    \end{enumerate}
\end{definition}

A sound distribution is a distribution that supports an agreement theorem.
\begin{definition} \label{def:distribution-soundness}
    Let \(X\) be a simplicial complex and let \(\mathcal{D}\) be an agreement distribution on \(X\). Let \(\eta,\varepsilon_0,e > 0\) be constants. We say that \(\mathcal{D}\) is \emph{\((\eta,\varepsilon_0,e)\)-sound} if for every ensemble of functions \(\mathcal{F}\) such that \(\agr_{\mathcal{D}}(\mathcal{F}) = \varepsilon \geq \varepsilon_0\), there exists a function \(G:X(0)\to \Sigma\), such that
    \begin{equation} \label{eq:dist-soundness}
        \Prob[r_1,r_2,\dots,r_q \sim \mathcal{D}]{\forall j \; G|_{r_j} \overset{1-\eta}{\approx} f_{r_j} \ve \forall i,j \; f_{r_i} = f_{r_j}} \geq \frac{1}{2}\varepsilon^e.    
    \end{equation}
\end{definition}
Here \(f \overset{1-\eta}{\approx} g\) means that \(f,g\) differ on at most a \(\eta\)-fraction of their coordinates, or stated differently \(\dist(f,g) \leq \eta\).

Examples of such distributions include:
\begin{example} \label{ex:sound-distributions}
    \begin{enumerate}
    \item The \(V\)-test extended to \(X=\Delta_n(k)\) is \((1/\sqrt{k},k^{-c},1)\)-sound for \(d \geq k^3\)  and \(c > 0\) \cite{DinurG2008} (see also \cite{ImpagliazzoKW2012}).
    \item The \(Z\)-test extended to \(X=\Delta_n(k)\) in \cite{ImpagliazzoKW2012} show that the \(Z\)-test is \((k^{-0.2},exp(-\Omega(k^{1/2}),O(1))\)-sound. For constant \(\lambda > 0\), this was improved to \((\lambda,exp(-\Omega(k)),O(1))\)-soundness by \cite{DinurL2017}.
\end{enumerate}
\end{example}
The works mentioned in the second item prove a weaker soundness guarantee than that of \pref{def:distribution-soundness} but \cite[Theorem 5.1]{DinurG2008} give a general technique to lift this weaker guarantee into the soundness in \pref{def:distribution-soundness}.

We also need the definition of \emph{cover soundness}.
\begin{definition}[Cover sound] \label{def:distribution-cover-soundness}
    Let \(X\) be a simplicial complex and let \(\mathcal{D}\) be an agreement distribution on \(X\). Let \(\eta,\varepsilon_0,e > 0\) be constants. We say that \(\mathcal{D}\) is \emph{\((\eta,\varepsilon_0,e)\)-cover-sound} if for every ensemble of functions \(\mathcal{F} = \sett{f_s:s \to \Sigma}{r \in X(k)}\) such that \(\agr_{\mathcal{D}}(\mathcal{F}) = \varepsilon \geq \varepsilon_0\),
    there exists a simplicial $\frac{1}{\varepsilon^e}$-cover \(\rho:Y \to X\) and a global function \(G:Y(0)\to \Sigma\) such that 
    \begin{equation} \label{eq:dist-cover-soundness}
        \Prob[\tilde s \in Y(k)]{f_{\rho(\tilde s)}\circ \rho \overset{1-\eta}{\approx} G|_{\tilde s}} \geq \frac{1}{2}\varepsilon^e.
    \end{equation}
\end{definition}
The probability in \eqref{eq:dist-cover-soundness} is equivalent to choosing \(s \in X(k)\) and then a random preimage \(\tilde{s} \in Y(k)\). Therefore this inequality implies that for at least \(\frac{1}{2}\varepsilon^e\)-fraction of the \(f_s\)'s, $f_s$ agrees with \(G|_{\tilde{s}}\) for one of the preimages \(\tilde{s} \in \rho^{-1}(s)\) (i.e. using \(\rho\) to send vertices from \(\tilde{s}\) to \(s\)).

\subsection{Suitable Complexes}
The main theorem of \cite{DiksteinD2023agr} requires certain conditions from the complexes to get a sound agreement test. 
\begin{definition}[Well connected] \label{def:well-connected-complex}
    Let \(X\) be a \(d\)-dimensional simplicial complex. Let \(d_1 \leq d-2\). We say that \(X\) is \emph{\(d_1\)-well-connected} if for every \(r \in X^{\leq d_1}\) it holds that \(\FA[d_1](X_r)\) is connected. Moreover, if \(r \in X(0)\) then we require that \(\FA[d_1](X_r)\) is \emph{simply connected}. When \(d_1\) is clear from context we omit it and say that \(X\) is \emph{well connected}.
\end{definition}

\begin{definition}[Suitable complex] \label{def:suitable-complexes}
    Let \(d>k>0\) be integers, and let \(\alpha > 0\). Let \(X\) be a \(d\)-dimensional simplicial complex. We say that \(X\) is \((d,k,\alpha)\)-\emph{suitable} if it has the following properties: 
\begin{enumerate}
    \item There exists some integer \(d_1\leq d-2\) with the following properties:
    \begin{enumerate}
        \item \(k^3 \leq d_1 \leq d \exp(-\alpha\frac{d_1}{k})\).
        \item \label{item:cob-exp} $X$ has \((\exp(-\alpha \frac{d_1}{k}),d_1)\)-swap-cocycle expansion. 
        \item \label{item:well-connected} \(X\) is \(d_1\)-well connected as in \pref{def:well-connected-complex}.
    \end{enumerate}
    \item \(X\) is a \(\frac{1}{d^2}\)-two-sided local spectral expander.
    \item \(X\) is a clique complex. 
    \end{enumerate}
\end{definition}

We remark that if one assumes swap coboundary expansion (not only swap cocycle expansion), then one can relax the requirements for well connectivity and clique complexes. These are required so to construct a cover, which in the swap coboundary case is immediate, since the cover is just many disjoint copies of the original complex. However, to stay as close to the statement in \cite{DiksteinD2023agr} as possible, we verify the full requirements.

\begin{claim} \label{claim:complexes-are-suitable}
     For every \(\alpha > 0\) and \(k\) and large enough $d$, there exists  \(p_0(d)\) and \(g_0(d)\) such that for \(g \geq g_0\) and prime \(p \geq p_0\), the \(d\)-skeleton of any finite quotient of \(\tilde{C}_g(\mathbb{Q}_p)\) is \((d,k,\alpha)\)-suitable. 
\end{claim}

\begin{proof}[Proof of \pref{claim:complexes-are-suitable}]
We first find \(d_1 \geq k^3\) such that \pref{thm:coboundary-expansion} shall imply our quotients are \((\exp(-\alpha \frac{d_1}{k}),d_1)\)-cocycle expanders.

Let \(C > 0\) be the constant such that any finite quotient \(X\) of \(\tilde{C}_g(\mathbb{Q}_p)\) is a \((\exp(-C\sqrt{d_1}),d_1)\)-swap cocycle expander for all \(d_1\) provided that \(g>d_1^5\) and \(p\) is large enough. Such a constant exists by \pref{thm:coboundary-expansion}. Let \(d_1\) be the smallest integer such that \(\exp(-\alpha \frac{d_1}{k}) < \exp(-C\sqrt{d_1})\) and such that \(d_1 \geq k^3\) (this is \(d_1 = \max \set{k^3, \lceil \frac{C^2}{\alpha^2}k^2 \rceil}\)).

Let \(d \geq \exp(\alpha \frac{d_1}{k}) d_1\) (this is required for item \(1(a)\) in the definition of suitable).
Recall that by \pref{lem:symplectic-local-spectral-expansion}, two-sided spectral expansion of the \(d\)-skeleton \(X\) goes to \(0\) as \(p,g\) go to infinity, so let \(p,g\) be large enough so that \pref{thm:coboundary-expansion} holds and so that \(X\) is \(\frac{1}{d^2}\) two-sided local spectral expander. We claim that for quotients of \(\tilde{C}_g\), the \(d\)-skeleton is \((d,k,\alpha)\)-suitable.
    \begin{enumerate}
        \item Swap cocycle expansion follows from \pref{thm:coboundary-expansion}.
        \item Connectivity of the links of the faces complex follows from their spectral expansion. By \pref{lem:symplectic-local-spectral-expansion} these complexes are local spectral expanders. The quotients are partite, so the swap walks between different colors are local spectral expanders by \pref{claim:color-swap} (for large enough \(p\)). This implies connectivity of the faces complex itself when \(g\) is large enough, because given two faces \(s_1,s_2\) in the face complex, we can find a face \(s_3\) whose colors are disjoint from the colors of \(s_1\) and \(s_2\). Then we can use the color-swap walk to go from \(s_1\) to \(s_3\) and from \(s_3\) to \(s_2\). 
        
        Coboundary expansion of the faces complexes of \(X_v\) for vertices \(v \in X(0)\), implies simple connectivity. The links, which are isomorphic to some links in \(\tilde{C}_g\) itself, are \emph{symplectic like} as in \pref{prop:links-of-affine-symplectic-building}. Therefore by \pref{thm:coboundary-expansion}, they are coboundary expanders for \emph{some positive constant} and therefore simply connected.
        \item As explained above, local spectral expansion is by \pref{lem:symplectic-local-spectral-expansion}.
        \item The building \(\tilde{C}_g\) is a clique complex by \pref{fact:basic-symplectic-building} and hence its quotient is a clique complex by \pref{claim:quotients-are-clique-complexes}.
    \end{enumerate}
\end{proof}

\subsection{Proof of the main theorem}
We can now state the formal version of \pref{thm:agrcover}. This theorem is essentially a reduction from a sound agreement test on $\Delta_d(k)$ to a sound agreement test on a complex $X$, assuming that $X$ is suitable. 
\begin{theorem}[{\cite[Theorem 3.1]{DiksteinD2023agr}}] \label{thm:agrcover-technical}
    For every $k > 0$, $\epsilon_0>\Omega (1/\log k)$, and \(C > 1\) there exists \(\alpha = \poly(\varepsilon_0)\) and \(d_0\in\mathbb{N}\) such that the following holds for any \(d \geq d_0\). Let \(\eta \leq \exp(-\poly(1/\varepsilon_0))\) be sufficiently small. Let \(\mathcal{D}\) be an agreement distribution on \(\Delta_{Ck}(k)\) such that its extension to \(\Delta_m(k)\) is \((\eta,\poly(\frac{1}{k}),O(1))\)-sound for every \(m \geq k^3\). Let \(X\) be a \((d,k,\alpha)\)-suitable complex and let \(\mathcal{D}_X\) be the extension of \(\mathcal{D}\) to $X$. Then \(\mathcal{D}_X\) is \((\varepsilon_0,\gamma,e)\)-cover sound where \(\gamma = \exp(\poly(1/\varepsilon_0))\eta\) and \(e = O(1)\) depends only on \(\varepsilon_0\).
\end{theorem}

We use this to prove our main theorem, the formal version of \pref{thm:main}.
\begin{theorem}[Main]\label{thm:main-technical}
    For every $k > 0$, $\epsilon_0>\Omega (1/\log k)$ \(\eta \leq \exp(-\poly(1/\varepsilon_0))\), and \(C > 1\) there exists a family of \(Ck\)-dimensional simplicial simplicial complexes \(X\) so that \(\mathcal{D}_X\) is \((\varepsilon_0,\gamma,e)\)-sound where \(\gamma = \exp(\poly(1/\varepsilon_0))\eta\) and \(e = O(1)\) depends only on \(\varepsilon_0\).
    These complexes are skeletons of quotients of \(\tilde{C}_g\).
\end{theorem}

A more concrete corollary is this.
\begin{corollary}[\pref{thm:main}]\label{cor:v-z-soundness}
        For every \(k > 0\), \(\varepsilon > \Omega(1/\log k)\) and \(\gamma > 0\) there exists a family of complexes \(X\) such that the extensions of the \(V\)-test and \(Z\)-test on \(k\)-sets are \((\varepsilon_0,\gamma,O(1))\) sound.
        These complexes are skeletons of quotients of \(\tilde{C}_g\).
\end{corollary}

\begin{proof}[Proof of \pref{thm:main-technical}]
        The family of complexes we construct are \(Ck\)-skeletons of the complexes in \pref{thm:CLhighdim}.\footnote{Technically, the \emph{analysis} of these complexes require higher dimensional sets, but the test itself only requires \(Ck\)-sets.} 
        Let \(m=\frac{1}{\varepsilon^e}\). We find \(g,p\) sufficiently large so that the family of \(k\)-skeletons in \pref{thm:CLhighdim} (which are quotients of \(\tilde{C}_g(\mathbb{Q}_p)\)) have the following two properties.
    \begin{enumerate}
        \item They are \((d,k,\alpha)\)-suitable for a sufficiently small \(\alpha\) and large \(d\) required in \pref{thm:agrcover-technical} (such \(g,p\) exists by \pref{claim:complexes-are-suitable}).
        \item They do not have connected \(m'\)-covers for any \(1< m' \leq m\).
    \end{enumerate}

    By cover soundness,
    \begin{equation*}
        \agr (\set{f_s}) > \eps  \quad \Longrightarrow \quad \exists Y\xrightarrowdbl{\rho} X, \exists G:Y(0)\to\Sigma,\quad \Pr_s[f_{\rho(s)}\circ \rho \overset{1-\gamma}{\approx} G|_s]\geq \frac{1}{2}\eps^e.
    \end{equation*} 
    where \(\rho\) is an \(m'\)-cover for some \(m' \leq m\). By assumption any \(m'\)-cover is trivial, i.e. just disjoint copies of \(X\). Therefore there exists some connected component in \(Y'\subseteq Y\) isomorphic to \(X\) such that 
    \[\Pr_{s \in X, \tilde{s}\in Y'\cap \rho^{-1}(s)}[f_{\rho(s)}\circ \rho \overset{1-\gamma}{\approx} G|_s]\geq\frac{1}{2} \eps^e.\]
    Thus \(G|_{Y'(0)}\) is the global function showing the distribution is \((\varepsilon_0,\gamma,e)\)-sound.
\end{proof}

\begin{proof}[Proof of \pref{cor:v-z-soundness}]
    Apply \pref{thm:main-technical} together with the tests in \pref{ex:sound-distributions}.
\end{proof}

%% file: nosmallcovers.tex
\section{Complexes with no small covers} \label{sec:no-small-covers}
In this section we will prove \pref{thm:CLhighdim}.

It is well known that every \(m\)-cover of \(X\) corresponds to an \(m\)-index subgroup in the fundamental group of \(X\) \cite{Surowski1984}. Thus \pref{thm:CLhighdim} follows directly from this proposition.

\begin{proposition}[Main] \label{prop:no-small-subgroups}
    Let \(m\geq 2\) be an integer and let \(g \geq 100\sqrt{m \log m}\). Then for every prime \(p\) the \(p\)-adic group \(G = Sp(2g,\mathbb{Q}_p)\) has infinitely many cocompact and torsion free lattices \(\set{G \geq \Gamma_1 \geq \Gamma_2 \geq \dots}\) that intersect trivially, satisfying that for every \(i\), \(\Gamma_i\) has no proper subgroup of index \(\leq m\).
\end{proposition}

We spell out the proof of \pref{thm:CLhighdim} for completeness.
\begin{proof}[Proof of \pref{thm:CLhighdim}, assuming \pref{prop:no-small-subgroups}]
    Let \(\set{\Gamma_i}_{i=1}^\infty\) be as in \pref{prop:no-small-subgroups}. By \pref{cor:discrete-torsion-free-acts-freely} their action on \(\tilde{C}_g\) is free. By \pref{claim:eventually-no-close-vertices-in-orbit} there is a subsequence \(\set{\Gamma_i}_{i=i_0}^\infty\) so that for every \(v \in \tilde{C}_g(0)\), \(i \geq i_0\) and non-zero element \(\gamma \in \Gamma_i\), \(\dist(v, \gamma. v) \geq 4\). Without loss of generality let us assume that all \(\Gamma_i\) have this property. By cocompactness, the quotients \(X_i = \Gamma_i \setminus \tilde{C}_g\) are finite. By \pref{claim:quotient-is-cover} and \pref{thm:fundamental-group-of-proper-action-quotient}, their universal cover is \(\tilde{C}_g\) and their fundamental groups are \(\Gamma_i\) respectively. By \pref{fact:no-small-index-implies-no-small-covers} and the fact that every \(\Gamma_i\) has no subgroups of index \(\leq m\), we get that all the complexes \(X_i\) have no connected \(m'\)-covers for \(m' \leq m\). Hence the complexes \(\set{X_i}_{i=1}^\infty\) are the desired family.
\end{proof}
See \pref{sec:poly-const} for a modification of this proof that gives a family of polynomially constructible complexes as such.

\subsection{Background}
\subsubsection{Profinite groups}
Before we prove \pref{prop:no-small-subgroups}, let us give some necessary background.

\begin{definition}[Profinite topology]
    Let \(\Gamma\) be a finitely generated group. Its profinite topology is defined as the topology generated by the basis of open sets \(\sett{\gamma H}{\gamma \in \Gamma, H \leq \Gamma, [\Gamma : H] < \infty}\).
\end{definition}
One can verify that in this topology the multiplication and inverse operations are continuous.

\begin{definition}[Profinite completion]
    Let \(\Gamma\) be a finitely generated group. Its profinite completion is the group \(\widehat{\Gamma}\), which is the topological completion of \(\Gamma\) with respect to the profinite topology.
\end{definition}
An equivalent definition is to say that \[\widehat{\Gamma} = \lim_{\leftarrow}\sett{\Gamma /N}{N \trianglelefteq \Gamma, [\Gamma: N] < \infty}.\]
This means that
\[\widehat{\Gamma} \subseteq \prod_{N \trianglelefteq \Gamma, [\Gamma: N] < \infty} \Gamma/N\]
where \((\gamma_{N})_{N} \in \widehat{\Gamma}\) if for all \(N_1 \leq N_2\), \(\pi_{N_1,N_2}(\gamma_{N_1}) = \gamma_{N_2}\), where \(\pi_{N_1,N_2}:\Gamma/N_1 \to \Gamma/N_2\) is the natural projection (see e.g. \cite{RibesZ2000}).

We say that a group \(K\) is profinite if it is an inverse limit of finite groups, an obvious example is that profinite completions are profinite groups. One can equivalently define profinite groups as topological groups where the topology is compact, totally disconnected and Hausdorff \cite{DixonDMS1999}. 

There is a homomorphism \(p:\Gamma \to \widehat{\Gamma}\) where \(p(\gamma) = (\gamma N)_{N}\). This homomorphism is injective exactly when \(\Gamma\) is residually finite (because then for every \(\gamma \ne \gamma'\) there is a normal subgroup \(N\) such that \(\gamma N \ne \gamma' N\)). We will only work with residually finite groups so we hence assume that \(\Gamma \subseteq \widehat{\Gamma}\) (and the inclusion is via this homomorphism).

\begin{proposition}[\cite{LubotzkyD1981}] \label{prop:profinite-correspondence}
    Let \(\Gamma\) be a finitely generated, residually finite group. Then there is a bijection between the finite index subgroups of \(\Gamma\), and the open subgroups in \(\widehat{\Gamma}\). This bijection preserves indexes. It is 
        \[H \leq \Gamma \mapsto \overline{H}\]
    and in the inverse direction by
        \[H' \leq \widehat{\Gamma} \mapsto H' \cap \Gamma.\]

        Moreover, we note that in this case for every finite index subgroup \(H \leq \Gamma\), \(\overline{H} = \widehat{H}\).
\end{proposition}

\subsubsection{Preliminary observations}

The following observations are elementary but we prove them here for concreteness.
\begin{claim} \label{claim:continuous-homomorphism} ~
\begin{enumerate}
    \item A group \(\Gamma\) has no proper subgroup of index \(\leq m\) if and only if the only homomorphism from \(\Gamma\) to \(Sym(m)\) is the trivial one.
    \item A profinite group \(K\) has no proper \emph{open} subgroup of index \(\leq m\) if and only if the only \emph{continuous} homomorphism from \(K\) to \(Sym(m)\) (equipped with the discrete topology) is the trivial one.
\end{enumerate}
\end{claim}

\begin{proof}
    Let us prove the contrapositive. I.e., that there is a proper subgroup of index \(\leq m\) if and only if there is a non-trivial homomorphism \(\phi: \Gamma \to Sym(m)\). For the first item, observe that if \(H \leq \Gamma\) is of index \(\leq m\), then \(\Gamma\) acts transitively on the cosets of \(H\). This action gives rise to a non-trivial homomorphism to \(Sym(\Gamma / H)\) which is (isomorphic to) a subgroup of \(Sym(m)\). In the other direction, suppose there is a non-trivial homomorphism to \(Sym(m)\). Then there is an element \(i \in [m]\) such that \(Orb(i)=\sett{\phi(\gamma).i}{\gamma \in \Gamma} \ne \set{i}\). It is easy to see that the stabilizer of \(i\), i.e. \(H = \sett{\gamma \in \Gamma}{\phi(\gamma) .i = i}\) is indeed a subgroup, and its index is the size of the orbit. In particular, this is a proper subgroup of index \(\leq m\).

    Let us move on to the second item. Note that if \(H \leq K\) is an open subgroup, then the homomorphism from $K$ to $Sym(K/H)$ as above is a continuous one. To show continuity we need to show that for any \(\sigma \in Sym(m)\), \(\phi^{-1}(\sigma)\) is closed. \(\phi^{-1}(\sigma)\) is a coset of the kernel, hence it is equivalent to show that \(ker(\phi)=\phi^{-1}(1)\) is closed. It can be verified that \(ker(\phi) = \bigcap_{g \in K} g^{-1} H g\). Let us explain why this is closed. As this is an intersection, it is enough to show that every \(g^{-1}Hg\) is closed. Multiplication is continuous, if \(H\) is closed then every \(g^{-1} H g\) is also closed, so it suffices to show that \(H\) is closed. But indeed, if \(H\) is open, then its complement is a union of cosets, which are also open - thus \(H\) is closed.
    
    For the other direction, let \(\phi\) be the homomorphism. The \(\leq m\)-index subgroup \(H\) constructed above contains \(ker(\phi)\), which is open (from continuity of \(\phi\)). Thus \(H\) is a union of cosets of an open subgroup, which implies it is open.
\end{proof}

\begin{claim} \label{claim:cont-homom-product}
    Let \(\set{K_i}_{i\in I}\) be profinite groups and let \(K = \prod_{i \in I} K_i\). The group \(K\) has a non-trivial continuous homomorphism to \(Sym(m)\) if and only if there exists \(j \in I\) such that \(K_j\) has a non-trivial continuous homomorphism to \(Sym(m)\).
\end{claim}

\begin{proof}
    For the first direction, observe that we can embed every \(K_j\) in \(K\) where every \(k \in K_j\) corresponds to \(\tilde{k} \in K\) where 
    \[\tilde{k}_i = \begin{cases}
        k & i=j \\
        1 & i \ne j
    \end{cases}.\]
    Let \(T = \iprod{K_j}_{j \in I} \subseteq K\). It is easy to see that \(T\) consists of all elements that are not the identity on a finite number of components. We note that \(T\) is dense inside \(K\). 
    Thus every continuous homomorphism of \(Sym(m)\) that is non-trivial, must also be non-trivial on \(T\). This implies that it must be non-trivial on one of the sets generating \(T\), i.e. that there exists a \(j\) such that \(\phi|_{K_j}\) is non-trivial.
    
    The other direction is simple. If \(\phi:K_j \to Sym(m)\) is a non-trivial homomorphism, then \(\phi \circ p_j:K \to Sym(m)\) is also a continuous and non trivial homomorphism, where \(p_j\) is the projection to the \(j\)-th coordinate. 
\end{proof}

An immediate corollary from the two claims is:
\begin{corollary} \label{cor:normal-subgroup-of-component}
    Let \(\set{K_i}_{i\in I}\) be profinite groups and let \(K = \prod_{i \in I} K_i\). Then \(K\) has a proper open subgroup of index \(\leq m\) if and only if there exists \(K_j\) that has a proper open subgroup of index \(\leq m\). 
\end{corollary}
Finally, we also need the notion of Frattini subgroups.
\begin{definition}[Frattini subgroup]
    Let \(K\) be a profinite group. Its Frattini subgroup \(\Phi(K)\), is the intersection of all maximal open subgroups \(M \leq K\). 
\end{definition}
The subgroup \(\Phi(K)\) is a normal subgroup since every conjugate of a maximal open group is also a maximal open subgroup. It is also closed, since it is an intersection of closed sets (recall that every open subgroup is also closed since its complement is a union of cosets, which are themselves open).

This is the main observation we need about Frattini subgroups.
\begin{observation} \label{obs:frattini-quotient}
    Let \(K\) be a profinite group. Then \(K\) has a proper subgroup of index \(\leq m\) if and only if \(K/\Phi(K)\) has a proper subgroup of index \(\leq m\).
\end{observation}

\begin{proof}
    Let \(L \leq K\) be a proper subgroup of index \(\leq m\). It is contained in a maximal proper subgroup \(M\) with index \(\leq m\). By definition \(\Phi(K) \leq M \leq K\) and by the correspondence theorem \(M/\Phi(K)\) has the same index in \(K/\Phi(K)\). The other direction follows from the same argument, reversed.
\end{proof}

\subsubsection{Quaternion Algebras}
In this subsection we present without proof some classical material from the theory of quaternion algebras and arithmetic groups. For more on this and complete references see \cite{PlatonovRR1993}, \cite{Morris2001} and \cite{MaclachlanR2003}.
\begin{definition}[Quaternion Algebra]
Let \(\mathbb{F}\) be a field of characteristic \(\ne 2\) and let \(a,b \in \mathbb{F}^*\). The Quaternion algebra is the \(\mathbb{F}\)-algebra
\[ \left (\frac{a,b}{\mathbb{F}} \right ) = \left \langle 1,i,j,k \Big | i^2=a, j^2=b, ij=-ji=k \right \rangle.\] 
\end{definition}
When \(a=b=-1\) and \(\mathbb{F}=\mathbb{R}\) this is the Hamilton's Quaternion algebra. If \(\mathbb{F} \subseteq L\) are two fields and \(a,b \in \mathbb{F}^*\) then \(\left (\frac{a,b}{\mathbb{F}} \right ) \leq \left (\frac{a,b}{L} \right )\).

The \(\mathbb{F}\)-algebra can be either a division algebra (e.g. \(\left (\frac{-1,-1}{\mathbb{R}} \right )\))
or it is isomorphic to \(M_{2}(\mathbb{F})\), the \(2\times 2\) matrices over \(\mathbb{F}\) (e.g. \(\left (\frac{-1,-1}{\mathbb{C}} \right )\)). In the first case we say that it ramifies over \(\mathbb{F}\) and in the second case we say that it splits.

The algebra \(\left (\frac{a,b}{\mathbb{F}} \right )\) splits over \(\mathbb{F}\) if and only if the quadratic form \(ax^2+by^2=1\) has a solution in \(\mathbb{F}\) (see \cite[Theorem 2.3.1 p.87]{MaclachlanR2003}).

Let us concentrate at quaternion algebras over \(\mathbb{Q}\). The field \(\mathbb{Q}\) has infinitely many completions - the \(p\)-adic fields \(\mathbb{Q}_p\) (one for every prime \(p\)) and the reals \(\mathbb{R}\), which we sometimes consider as \(p=\infty\) and write \(\mathbb{Q}_\infty = \mathbb{R}\). Given \(\left (\frac{a,b}{\mathbb{Q}} \right )\) we can extend it to every \(\mathbb{Q}_p\) (\(p \leq \infty\)). A given algebra ramifies only at a finite \emph{even} set of completions \cite[Theorem 2.7.3 p.99]{MaclachlanR2003}.

For concreteness look now at the algebra \(H^\ell(\mathbb{F}) = \left (\frac{-1,-\ell}{\mathbb{F}} \right )\) where \(\mathbb{F} \supseteq \mathbb{Q}\) and \(\ell\) is a prime and \(\ell \equiv 3 \text{ (mod 4)}\).

\begin{proposition}
    The algebra \(H^\ell\) ramifies over \(\mathbb{R}\) and \(\mathbb{Q}_\ell\) and splits over \(\mathbb{Q}_p\) for every \(p\ne \ell,\infty\).
\end{proposition}

\begin{proof}
    Assume first that \(p \ne 2,\ell,\infty\). In this case \(-1,-\ell \in \mathbb{Z}_p\), the \(p\)-adic integers, but not in \(p \mathbb{Z}_p\). By \cite[Theorem 2.6.6, item (1), p. 97]{MaclachlanR2003} the algebra \(H^\ell\) splits over \(\mathbb{Q}_p\).

    In the other hand, by the second item of the same theorem there, \(\left (\frac{-1,-\ell}{\mathbb{Q}_\ell} \right )\) does not split as \(-1\) is not a square mod \(\ell \mathbb{Z}_\ell\), since \(\ell \cong 3 \text{ (mod 4)}\).
    
    Now, the quadratic form \(-x^2-\ell y^2 = 1\) has no solution in \(\mathbb{R}\), thus \(\left (\frac{-1,-\ell}{\mathbb{R}} \right )\) ramifies.
    
    Finally, by \cite[Theorem 2.7.3 p.99]{MaclachlanR2003} the algebra \(\left (\frac{-1,-\ell}{\mathbb{Q}_p} \right )\) must split since the number of ramified completions of \(H^\ell\) is even.
\end{proof}

For \(\alpha = w+xi+yj+zk\) we define the involution \(\overline{\alpha} = \overline{w+xi+yj+zk} = w-xi-yj-zk\). For every \(\mathbb{F} \supseteq \mathbb{Q}\) we also denote the hermitian form 
\[\iprod{,}:H^\ell(\mathbb{F})^g \times H^\ell(\mathbb{F})^g \to H^\ell(\mathbb{F}); \;
\iprod{\alpha,\beta} = \sum_{t=1}^g \overline{\alpha_t} \beta_t.\]

With this form in mind we denote by \(G(\mathbb{F}) = SU_g(H^\ell(\mathbb{F}))\), i.e.
\[SU_g(H^\ell(\mathbb{F})) = \sett{A \in M_{g \times g}(H^{\ell}(\mathbb{F}))}{ \forall x,y\in H^{\ell}(\mathbb{F})^g, \; \iprod{Ax,Ay} = \iprod{x,y}}.\]

When \(\mathbb{F} = \mathbb{Q}_p\), then \(G(\mathbb{Q}_p))\), being a closed group of some \(GL_{m}(\mathbb{Q}_p)\) is locally compact, totally disconnected and Hausdorff. This implies that every compact subgroup of it is profinite.

\begin{fact} \label{fact:chain-of-groups}
    Let \(\ell\) be a prime. Then
    \begin{enumerate}
        \item For every prime \(p \ne \ell\), \(SU_g(H^\ell(\mathbb{Q}_p)) \cong Sp(2g,\mathbb{Q}_p)\).
        \item The group \(SU_g(H^\ell(\mathbb{R}))\) is a compact Lie group.
        \item\label{item:infinite-subgroups} The group \(SU_g(H^\ell(\mathbb{Q}_\ell))\) is an \(\ell\)-adic Lie group, and therefore it has a torsion-free pro-\(\ell\) open subgroup \(H_0\). In particular, the index of an open subgroup in \(H_0\) is a power of \(\ell\) \cite{DixonDMS1999}. Hence there is an infinite sequence of open compact subgroups, 
        \[H_0 \trianglerighteq H_1 \trianglerighteq \dots\] such that every \(H_{i+1}\) is a normal subgroup inside \(H_i\), the index \([H_i:H_{i+1}]=\ell\) and the intersection \(\bigcap_{i=0}^\infty H_i = \set{1}\). 
    \end{enumerate}
\end{fact}
Let us fix \(g \geq 1\) and two primes \(p \ne \ell\). Denote the algebraic group \(G( \cdot ) = SU_g(H^\ell(\cdot))\). The ring of Adeles of \(\mathbb{Q}\), denoted \(\mathbb{A} = \mathbb{A}_{\mathbb{Q}}\), is the restricted product \(\prod_{p \leq \infty}^{*} \mathbb{Q}_p\) (including \(\mathbb{Q}_\infty = \mathbb{R}\)). This is the subring of \(\prod_{p \leq \infty} \mathbb{Q}_p\) where a sequence \((x_p)_p\) is in \(\mathbb{A}\) if for all but a finite number of elements, \(x_p \in \mathbb{Z}_p\) (there is no restriction on the real coordinate). The Adeles come with a natural topology (cf.\ \cite{PlatonovRR1993}) in which \(\mathbb{R} \times \prod_{p < \infty}\mathbb{Z}_p\) is an open subring.

Consider now \(G(\mathbb{A})\). This is the restricted product \(\prod_{p \leq \infty}^{*} G(\mathbb{Q}_p)\), that is, this is a subgroup of the \(\prod_{p \leq \infty}G(\mathbb{Q}_p)\) where \((g_p)_{p} \in \prod_{p \leq \infty}^{*} G(\mathbb{Q}_p)\) if for all but a finite number of elements, \(g_p \in G(\mathbb{Z}_p)\). This group comes with a topology induced by \(\mathbb{A}\). The group \(G(\mathbb{Q})\) is diagonally embedded in \(G(\mathbb{A})\) as a discrete subgroup but for every \(p\ne \infty\), the projection of \(G(\mathbb{Q})\) into \(G(\mathbb{A}) / G(\mathbb{Q}_p)\) is dense by the Strong Approximation Theorem \cite{PlatonovRR1993}.

Let \(K_0 \leq G(\mathbb{A})\) be
\[K_0 = G(\mathbb{R}) \times H_0 \times G(\mathbb{Q}_p) \times \prod_{q \ne p,\ell,\infty} G(\mathbb{Z}_q),\]
and let \(\Gamma_0 = K_0 \cap G(\mathbb{Q})\).

Recall that by \pref{item:infinite-subgroups} in \pref{fact:chain-of-groups}, \(G(\mathbb{R})\) is compact. As a subgroup of \(G(\mathbb{Q})\), the group \(\Gamma_0\) is also discrete, and so is its projection to \(G(\mathbb{Q}_p)\) since the kernel of the projection \(K_0 \to G(\mathbb{Q}_p)\) is compact (this is because all other factors in the product are compact). However, the projection of \(\Gamma_0\) to \(H_0 \times \prod_{q \ne p,\ell,\infty} G(\mathbb{Z}_q) \cong H_0 \times \prod_{q \ne p,\ell,\infty} Sp(2g,\mathbb{Z}_p)\) is dense, again by the Strong Approximation Theorem.
Now the affirmative solution to congruence subgroup problem (see \cite{Rapinchuk1989}) says that in such a situation, this profinite group \(H_0 \times \prod_{q \ne p,\ell,\infty} G(\mathbb{Z}_q)\), is isomorphic to the profinite completion of \(\Gamma_0\).
In summary: 
\begin{fact}\label{fact:gamma-0}~
    \begin{enumerate}
        \item The group \(\Gamma_0\) is a discrete cocompact lattice of \(SU_g(H^\ell(\mathbb{Q}_p)) \cong Sp(2g,\mathbb{Q}_p)\).
        \item The profinite completion \(\widehat{\Gamma_0} \cong H_0 \times \prod_{q \ne \ell,p}Sp(2g,\mathbb{Z}_q)\) as above. Here \(\mathbb{Z}_q\) are the \(q\)-adic integers and \(H_1\) is the pro-\(\ell\) group. The group \(\Gamma_0\) is embedded diagonally in \(H_1 \times \prod_{q \ne \ell,p}Sp(2g,\mathbb{Z}_q)\), and the projection of this embedding to every factor is injective.
    \end{enumerate}
\end{fact}

\subsection[Proof of the main lemma]{Proof of \pref{prop:no-small-subgroups}}
\begin{proof}
Fix some two primes \(p,\ell\) such that \(\ell > m\). Let \(\Gamma_0\) be as above. Let \(\tilde{H}_i \cong H_i \times \prod_{q \ne \ell,p}Sp(2g,\mathbb{Z}_q)\) be open subgroups of \(\widehat{\Gamma_0}\), where \(H_i\) are the subgroups from \pref{fact:chain-of-groups}. Let \(\Gamma_i = \Gamma_0 \cap \tilde{H}_i\). By \pref{prop:profinite-correspondence}, \(\Gamma_i \subseteq \Gamma_0\) is a finite index subgroup of \(\Gamma_0\). In particular, by \pref{fact:gamma-0} it is a discrete cocompact lattice of \(Sp(2g,\mathbb{Q}_p)\). In addition, by \pref{prop:profinite-correspondence}, \(\tilde{H}_i \cong \widehat{\Gamma_i}\), so instead of showing that \(\Gamma_i\) has no subgroups of index \(\leq m\), we will show that \(\tilde{H}_i\) has no \emph{open} subgroups of index \(\leq m\).

Let us fix \(i\). As above, \(\tilde{H}_i = \widehat{\Gamma_i}\) is an infinite product of profinite groups. 
By \pref{cor:normal-subgroup-of-component} showing that \(\tilde{H}_i\) has no  open subgroup of index \(\leq m\) is equivalent to showing that none of the groups in the product have open subgroups of index \(\leq m\).

The group \(H_i\) is a pro \(\ell\)-group so all proper open subgroups must have index at least \(\ell > m\). Let us consider \(Sp(2g,\mathbb{Z}_q)\). Fix some \(q \ne p,\ell\). Assume towards contradiction that \(K = Sp(2g,\mathbb{Z}_q)\) has a subgroup of index at most \(m\). By \pref{obs:frattini-quotient}, \(K/\Phi(K)\) also has a subgroup of index at most \(m\). By \cite{Weigel1996}
\[K/\Phi(K) \cong PSp(2g,\mathbb{F}_q).\]

It is well known that \(PSp(2g,\mathbb{F}_q)\) is a simple group. By \pref{claim:continuous-homomorphism}, if it has a non-trivial homomorphism to \(Sym(m)\) the kernel of this homomorphism is a proper \emph{normal} subgroup of index at most \(m!\). But the only proper normal subgroup in \(PSp(2g,\mathbb{F}_q)\) is the trivial subgroup, which has index larger than \(m!\), since the order of \(PSp(2g,\mathbb{F}_q)\) is \( \abs{PSp(2g,\mathbb{F}_q)} > \abs{PSp(2g,\mathbb{F}_2)} \geq 2^{g^2-g-1}>m!\) for \(g \geq 100\sqrt{m\log m}\).

Finally, the groups \(\Gamma_i\) have no torsion, because by \pref{fact:gamma-0}, they are embedded in \(H'\) which has no torsion. This diagonal embedding also promises that the intersection \(\bigcap_{i=1}^\infty \Gamma_i = \set{1}\) since by the definition this is equal to
\[\Gamma_0 \cap \bigcap_{i=1}^\infty \left ( H_i \times \prod_{q \ne \ell,p}Sp(2g,\mathbb{Z}_q) \right ) = \Gamma_0 \cap \left ( \set{1} \times \prod_{q \ne \ell,p}Sp(2g,\mathbb{Z}_q) \right ).\]
By \pref{fact:gamma-0}, \(\Gamma_0\) is \emph{diagonally} embedded inside \(H_1 \times \prod_{q \ne \ell,p}Sp(2g,\mathbb{Z}_q)\), this implies that this intersection is \(\set{1}\).
\end{proof}

\subsection{Constructing a Family of Quotients in Polynomial Time} \label{sec:poly-const}
In this section we prove a polynomially constructible version of \pref{thm:CLhighdim}. Recall that a family \(\set{X_N}_N\) is polynomially constructible if there is an algorithm that is given \(n\) in unary input, and it produces a member in the family \(X_N\) with at least \(n\) vertices in \(\poly(n)\)-time.

\begin{proposition} \label{prop:poly-constructible}
    Let \(m \geq 2\) and let \(g \geq 100 \sqrt{m \log m}\). There exists a family of \emph{polynomially constructible} complexes \(\set{X_N}_N\) that are finite quotients of \(\tilde{C}_g(\mathbb{Q}_p)\) such that every \(X_N\) has no \(m'\)-connected covers for \(1< m' \leq m\).
\end{proposition}

We fix \(m\) as in the proposition and a sufficiently large prime \(\ell\) as in the construction above. Let \(G\) be a finite group acting on \([\ell]\). To understand the algorithm, we must explain how to construct a cover for a clique complex \(X\), given a cocycle \(\phi:\dir{X}(1) \to G\). Let \(X^\phi\) be the following complex.
\begin{enumerate}
    \item The vertices of the cover are \(X^\phi(0) = X(0) \times [\ell]\).
    \item The edges are all \(\set{(v,i),(u,j)}\) such that \(\set{v,u} \in X(1)\) and \(j = \phi(uv). i\).
    \item The higher dimensional faces are all cliques in the graph \((X^\phi(0),X^\phi(1))\).\footnote{This construction could be modified to work over complexes that are not clique complexes, but we omit the details for simplicity.}
\end{enumerate}
This construction indeed results in an \(\ell\)-cover of \(X\) for any cocycle \(\phi\). For a proof see \cite[Proposition 2.1]{Surowski1984}.

Given \(X\) and such a coycle \(\phi:\dir{X}(1) \to G\), the construction of \(X^\phi\) can clearly be done in \(O_\ell (|X|)\) time. We will use this construction with \(G = \mathbb{Z}/\ell \mathbb{Z}\) that is, \(G = \set{0,1,\dots,\ell-1}\) with addition modulo \(\ell\). For this special case we state the following fact.
\begin{fact} \label{fact:covers-from-l-abelian-cohomology}
    Let \(X\) be a connected clique complex and let \(\phi \in Z^1(X,\mathbb{Z}/\ell \mathbb{Z}) \setminus B^1(X,\mathbb{Z}/\ell \mathbb{Z})\). Then \(X^\phi\) is a connected \(\ell\)-cover of \(X\), and moreover the fundamental group of \(X^\phi\) is a normal index \(\ell\) subgroup of the fundamental group of \(X\).
\end{fact}
The first part is by \cite[Theorem 5.3]{Surowski1984}. The ``moreover'' is by \cite[Proposition 1.39, item (a)]{Hatcher2002}. Before proving the proposition, we will also use the following claim in the analysis.
\begin{claim} \label{claim:any-chain-good}
        Let \(\Gamma\) be a subgroup of \(\Gamma_0\) above such that \(\widehat{\Gamma} = H \times \prod_{q \ne \ell,p} Sp(2g,\mathbb{Z}_q)\), where \(H \leq H_1\) is a finite index subgroup of \(H_1\), the pro-\(\ell\) group constructed above. Then
        \begin{enumerate}
            \item \(\Gamma\) has no proper subgroup of index \(< m\).
            \item There exists a normal subgroup \(\Gamma' \trianglelefteq \Gamma\) of index \(\ell\).
            \item Every normal subgroup \(\Gamma' \trianglelefteq \Gamma\) of index \(\ell\) has that \(\widehat{\Gamma'} = H' \times \prod_{q \ne \ell,p}Sp(2g,\mathbb{Z}_q)\) for some \(H' \trianglelefteq H\) of index \(\ell\). 
        \end{enumerate}
\end{claim}

\begin{proof}[Proof of \pref{claim:any-chain-good}]
    The proof of the first item is just to repeat the proof of \pref{prop:no-small-subgroups}. The proof for the second item is by observing that \(H\), as a finite index subgroup of \(H_1\), is also pro-\(\ell\), and therefore a maximal subgroup \(H' \leq H\) is normal in \(H\) and has index \(\ell\). Therefore, the group \(\Gamma' \leq \Gamma\) whose profinite completion is \(\widehat{\Gamma'} = H' \times \prod_{q \ne \ell,p} Sp(2g,\mathbb{Z}_q)\) is an index \(\ell\) subgroup in \(\Gamma\) (because of \pref{prop:profinite-correspondence} and \([\widehat{\Gamma}: \widehat{\Gamma'}] = \ell\)). It is normal because \(H' \trianglelefteq H\) which implies that \(\widehat{\Gamma'} \trianglelefteq \widehat{\Gamma}\), and this in turn shows that \(\widehat{\Gamma'} \cap \Gamma_0 \trianglelefteq \widehat{\Gamma} \cap \Gamma_0\), but these are equal to \(\Gamma', \Gamma\) respectively by \pref{prop:profinite-correspondence}.

    Finally, let us prove the third theorem. Let \(\Gamma' \trianglelefteq \Gamma\) be a normal subgroup of index \(\ell\). Therefore the same holds for \(\widehat{\Gamma'} \trianglelefteq \widehat{\Gamma}\). Thus the quotient map \(\psi: \widehat{\Gamma} \to \widehat{\Gamma} / \widehat{\Gamma'}\) is non-trivial and by \pref{claim:cont-homom-product}, this implies that \(\psi\) is non-trivial when restricted to one of the factors of the product \(\widehat{\Gamma} = H \times \prod_{q \ne \ell,p} Sp(2g,\mathbb{Z}_q)\)\footnote{Formally \pref{claim:cont-homom-product} is about the non-trivial homomorphism that a proper subgroup induces from \(\widehat{\Gamma}\) to the group of symmetries. In the \emph{normal} subgroup case, one can repeat the same proof for the quotient map.}. However, for every factor \(Sp(2g,\mathbb{Q}_q)\) the homomorphism \emph{is trivial}. Indeed, fix \(K = Sp(2g,\mathbb{Q}_q)\), and assume for sake of contradiction that \(\psi|_K\) is not trivial. As \(Im \psi \cong \mathbb{Z}/\mathbb{Z}_\ell\), if it is not trivial, then the kernel of \(\psi|_K\) is an index \(\ell\) subgroup inside \(K\) and because \(\ell\) is prime this is a maximal subgroup, and so it contains the Frattini subgroup of \(K\), which we denoted \(\Phi(K)\). In other words, \(ker(\psi|_K)/\Phi(K)\) is a normal subgroup of index \(\ell\) inside \(K/\Phi(K)\). But by \cite{Weigel1996}, \(K/\Phi(K) \cong PSp(2g,\mathbb{F}_q)\) which is simple and whose order is not a prime when \(g \geq 2\), thus it has no normal subgroup of index \(\ell\) reaching a contradiction.
    

    Thus \(\widehat{\Gamma'} \leq H' \times \prod_{q \ne \ell,p} Sp(2g,\mathbb{Z}_q)\) and because both subgroups have the same finite index \(\ell\) inside \(\widehat{\Gamma}\), we have that \(\widehat{\Gamma'} = H' \times \prod_{q \ne \ell,p} Sp(2g,\mathbb{Z}_q)\).
\end{proof}

\begin{proof}[Proof of \pref{prop:poly-constructible}]
    Let \(\Gamma = \Gamma_j \leq \Gamma_0\) be a subgroup of the chain of groups defined above that acts properly cocompactly on \(\tilde{C}_g(\mathbb{Q}_p)\) (there exists such a subgroup by \pref{claim:eventually-no-close-vertices-in-orbit}). The algorithm we present is the following.
\begin{algorithm} \label{alg:find-complex}
    Input: Number of vertices \(n\).
    \begin{enumerate}
        \item Set \(X_0 := \Gamma \setminus \tilde{C}_g(\mathbb{Q}_p)\) and \(i:=0\).
        \item While \(|X_i(0)| < n\):
        \begin{enumerate}
            \item Find an arbitrary cocycle \(\phi_i: \dir{X_i}(1) \to \mathbb{Z}/\ell \mathbb{Z}\) that is not a coboundary. If there is no such cocycle, output `FAIL'.
            \item Construct \(X_{i+1} := X_i^\phi\) and increment \(i:=i+1\).
        \end{enumerate}
        \item Output \(X_{i}\).
    \end{enumerate}
\end{algorithm}

We will show the following to prove the proposition.
\begin{enumerate}
    \item Assuming the algorithm has not output `FAIL' up to the \(i\)-th iteration, the complex \(X_{i+1}\) is a quotient of \(\tilde{C}_g(\mathbb{Q}_p)\) with no proper covers of index \(<m\).
    \item The algorithm never outputs `FAIL'. That is, for every \(i \geq 0\), there exists a cocycle \(\phi_i: \dir{X_i}(1) \to \mathbb{Z}/\ell \mathbb{Z}\) that is not a coboundary.
    \item The algorithm terminates in \(\poly(n)\) time.
\end{enumerate}

\paragraph{Item 1}We prove by induction on \(i\) that \(X_i = \Gamma_i' \setminus \tilde{C}_g(\mathbb{Q}_p)\) for some \(\Gamma_i' \leq \Gamma_0\) such that its profinite completion \(\widehat{\Gamma_i'} = H_i' \times \prod_{q \ne \ell,p} Sp(2g,\mathbb{Q}_q)\) as in \pref{claim:any-chain-good}. Then we will conclude by \pref{claim:any-chain-good} that \(\Gamma_i'\) has no proper subgroup of index \(<m\), which by \pref{fact:connected-covers-and-subgroups} implies that \(X_i\) has no non-trivial \(m'\)-covers for \(m' < m\).

For \(X_0\) this is simply by the choice of complex. Assuming this is true for \(X_i\), then \(X_{i+1} = X_i^{\phi_i}\) is a connected cover by \pref{fact:covers-from-l-abelian-cohomology}. Thus by this fact its fundamental group is an index \(\ell\) normal subgroup \(\Gamma_{i+1}' \trianglelefteq \Gamma_i'\). As a connected cover of a quotient of \(\tilde{C}_g\), this complex is also a quotient of \(\tilde{C}_g(\mathbb{Q}_p)\), that is, \(X_{i+1} = \Gamma_{i+1}' \setminus \tilde{C}_g(\mathbb{Q}_p)\). The profinite completion of  \(\Gamma_{i+1}'\) is \(\widehat{\Gamma_{i+1}'} = H_{i+1}' \times \prod_{q \ne \ell,p} Sp(2g,\mathbb{Q}_q)\) as in \pref{claim:any-chain-good}.

\paragraph{Item 2}Fix \(X_i\) and recall that its fundamental group is some \(\Gamma_i' \leq \Gamma_0\) whose profinite completion \(\widehat{\Gamma_i'} = H_i' \times \prod_{q \ne \ell,p} Sp(2g,\mathbb{Q}_q)\) is as in \pref{claim:any-chain-good}. By \pref{claim:any-chain-good} the group \(\Gamma_i'\) has a normal subgroup \(\Gamma' \trianglelefteq \Gamma_{i}'\) of index \(\ell\). By \cite[Theorem 5.3]{Surowski1984} the normal subgroup \(\Gamma'\) induces a connected \(\ell\)-cover that can be realized via the cocycle construction  described above, using a cocycle \(\phi_i' : \dir{X_i}(1) \to \Gamma_{i}'\) with the action of \(\Gamma_i'\) on \(\Gamma_i' / \Gamma'\) by left multiplication.  Because the subgroup \(\Gamma' \trianglelefteq \Gamma_{i}'\) is normal, this realization is the same as the realization given by \(\phi_i : \dir{X_i}(1) \to \Gamma_i' / \Gamma'\) with \(\phi_i(uv) = \phi_i'(uv) \Gamma'\). The index \(\ell\) is prime, thus \(\Gamma_i' / \Gamma' \cong \mathbb{Z}/\ell \mathbb{Z}\) so \(\phi_i \in Z^1(X_i,\mathbb{Z}/\ell \mathbb{Z})\). It is not a coboundary because for coboundaries the construction \(X^\phi\) is not connected.

\paragraph{Item 3} There are at most \(\left \lceil \frac{\log n}{\log \ell} \right \rceil\) iterations of the loop in \pref{alg:find-complex}, because the size \(|X_{i+1}(0)| = \ell|X_i(0)|\). As we saw above the construction time of \(X_{i+1}\) is \(O_\ell(|X_{i+1}|)\). By the first item all these complexes are quotients of \(\tilde{C}_g\) and are therefore bounded degree complexes of a fixed dimension \(g\), hence this is \(O_{\ell,g}(|X_{i+1}(0)|) = O(n)\).
    
Finding a \(\phi_i: \dir{X_i}(1) \to \mathbb{Z}/\ell \mathbb{Z}\) that is not a coboundary is done via linear algebra. The group \(\mathbb{Z}/\ell \mathbb{Z}\) is actually a field. The space of cocycles is a linear subspace of all functions \(\phi\) on the edges satisfying
    \[ \phi(uv) + \phi(vw) + \phi(wu) = 0\]
for every \(\set{u,v,w} \in X_i(2)\). There are \(O(n)\) edges and triangles so finding the space of \emph{all} cocycles could be done in \(O(n^3)\)-time via Gauss elimination. Moreover, the subspace of coboundaries is also a linear subspace \(B^1(X,\mathbb{Z}/\ell \mathbb{Z}) = Im(\delta_0)\), so finding \(\phi \in Z^1(X,\mathbb{Z}/\ell \mathbb{Z}) \setminus B^1(X,\mathbb{Z}/\ell \mathbb{Z})\) can be done in \(O(n^3)\) by linear algebraic means.
\end{proof}

We end this subsection with an open question. Let \(\set{X_N}\) be a family of simplicial complexes and without loss of generality identify the vertices of every complex \(X_N\) with \(\set{1,2,\dots,|X_N(0)|}\). We say that the family \(\set{X_N}\) is \emph{strongly explicit}, if there is an algorithm that takes as input \(N\), a vertex \(i \in X_N(0)\) and a number \(j\). The algorithm runs in \(\poly \log |X_N(0)|\) and outputs the \(j\)-th neighbor of \(i\) in the underlying graph of \(X_N\), where we order the neighbors with the natural number ordering \(<\). 

\pref{alg:find-complex} \emph{does not} suffice for this definition because in order to find \(X_{i+1}\) we need a full description of \(X_i\). In comparison, the quotients of \(\tilde{A}_n\) in \cite{LubotzkySV2005b} are proven there to be a strongly explicit family by using the fact that they are also Cayley graphs. 

Is there a strongly explicit construction of the complexes we present as well?

%% file: coboundary.tex
\subsection{Coboundary expansion of the symplectic spherical building}
In this subsection we show that the spherical symplectic building is a \(\Omega(1)\)-coboundary expander for \(1\)-cochains. 
\begin{theorem} \label{thm:symplectic-building-coboundary-expansion}
There exists an absolute constant $\beta>0$ such that for all $g\in\N$ and all finite fields \(h^1(C_g) \geq\beta\).
\end{theorem}
We note that \cite{LubotzkyMM2016} already gave a bound that depends on the rank \(g\) (for coefficients in \(\mathbb{F}_2\), but the same technique applies for all coefficients). So without loss of generality we may assume that \(g\) is sufficiently large.

\begin{proof}[Proof of \pref{thm:symplectic-building-coboundary-expansion}]
    By \pref{thm:coboundary-expansion-from-colors} it suffices to give a lower bound to a noticeable fraction of color restrictions. The colors we use are
    \[ \Col = \sett{\set{i_0, i_1, i_2} \in \binom{[2g]}{3}}{i_1 \geq 2i_0, i_2 \geq 3i_1,  17i_1 \leq 2g }.\]
We will prove below the following lemma.
\begin{lemma} \label{lem:color-restriction-of-symplectic-is-cob-exp}
    Let \(I \in \Col\). Then \(h^1(C_g^I) \geq \Omega(1)\).
\end{lemma}

Ordering the colors of \(\Col\) by their size, we have that
\[\Col \supseteq \sett{\set{i_0,i_1,i_2}}{i_0 \in [0,0.01g], i_1 \in  [0.02g,0.1g], i_2 \in [0.3g,g]}\]
so \(|\Col| = \Omega \left ( \binom{n}{3} \right )\), or equivalently \(p = \frac{|\Col|}{\binom{n}{3}} = \Omega(1)\). Thus by \pref{thm:coboundary-expansion-from-colors} \(h^1(S) = \Omega(1)\).
\end{proof}

Our main effort will be proving \pref{lem:color-restriction-of-symplectic-is-cob-exp}. We will do so using non-abelian cones. For constructing them we will need the following claim, that will imply that the diameter of \(C_g^I\) is constant.

\begin{claim} \label{claim:symplectic-diameter}
    Let \(i < j\), and let \(G\) be the bipartite containment graph between \(C_g[i]\) and \(C_g[j]\), namely connecting $u\in C_g[i]$ to $v\in C_g[j]$ iff $u\subset v$ as subspaces. Let \(H\) be the bipartite containment graph between \(C_g'[i]\) and \(C_g'[j]\) where \(C_g'[x] = \sett{u \subseteq V}{dim(u)=x}\) (i.e. all subspaces, not just isotropic subspaces).
    Then for every \(v_1,v_2 \in S[i]\) it holds that \(\dist_G(v_1,v_2) \leq 2\dist_H (v_1,v_2)\). In particular, \(diam(G) \leq 2 diam(H)\).
\end{claim}
For intuition, the reader is encouraged to verify the claim for the case where $i=1$ and $j=2$.

\begin{proof}[Proof of \pref{claim:symplectic-diameter}]
The claim will follow if we show that for every \(v_1,v_2 \in C_g[i]\) that are of distance \(2\) in \(H\) have distance at most \(4\) in \(G\). 
   We will use the following easy fact that holds for every two subspaces \(u,u'\),
    \begin{equation}\label{eq:dimsum}
        \dim(u+u') = \dim(u)+\dim(u')-\dim(u\cap u').
    \end{equation}
    Let \(u_1\supset v_1\) be a $j$-dimensional isotropic subspace. Let $w = u_1+v_2$. By \eqref{eq:dimsum},
    \[\dim(w) = \dim(u_1) + \dim (v_2) - \dim (u_1\cap v_2) = j + i - \dim(v_1\cap v_2)\]
    We claim that there exists a \(j\)-dimensional isotropic subspace \(u_2 \subseteq w\) such that \(u_2 \supseteq v_2\):
    a corollary from Witt's theorem says that every maximal isotropic subspace has the same dimension (see e.g. \cite[Theorem 3.10]{Artin1957}). We will use this corollary on the bilinear form restricted to \(w\). In particular, \(u_1 \subseteq w\) is an isotropic subspace in \(w\) of dimension \(j\), so a maximal isotropic subspace inside \(w\) has dimension \(\geq j\). Thus there is also an isotropic subspace \(v_2 \subseteq u_2 \subseteq w\) of dimension \(j\) (that is contained in a maximal isotropic subspace that contains \(v_2\)).
    
    Next, observe that 
        \[\dim (u_1\cap u_2)= 2j-\dim(u_1+u_2) \geq 2j-\dim(w) \geq 2j - (j+i-\dim(v_1\cap v_2) =j-i+
    \dim(v_1\cap v_2)\geq i
    \]
    where the last inequality follows from the fact that $\dist_H(v_1,v_2)=2$ so $j\geq \dim(v_1+v_2) = 2i-\dim(v_1\cap v_2)$, again using \eqref{eq:dimsum}.
    This means that there some $i$-space $v\subseteq u_1\cap u_2$ and a length $4$ path $v_1\to u_1\to v\to u_2\to v_2$ in $G$ as needed.
\end{proof}

\begin{proof}[Proof of \pref{lem:color-restriction-of-symplectic-is-cob-exp}]
    The symplectic group induces a transitive action on the triangles of \(C_g^I\) (for every \(I \in \Col\)), therefore by \pref{lem:group-and-cones}, it is enough to find a constant sized cone for \(C_g^I\), and we can conclude that \(h^1(C_g^I) = \Omega(1)\).

    We define the following cones. Fix \(v_0 \in C_g[i_0]\) as a base point.

    For any $u\in C_g[i_0]$, $\dim(v_0+u) \leq 2i_0 \leq i_1$ so there is some $i_1$-space (not necessarily isotropic) that contains both $v_0$ and $u$. By \pref{claim:symplectic-diameter}, there must be a $4$-path between them in $G$. This implies that there exists a \(5\)-path from \(v_0\) to any \(u \in C_g^I\) so that for any \(u' \ne u\) in this path it holds that \(u' \notin C_g[i_2]\) (first go from \(u\) to some \(i_0\)-dimensional subspace of \(u'\) and then traverse from that subspace to \(v_0\) in four steps).

    Hence, for every \(u \ne v_0\) we fix an arbitrary shortest path \(P_u\) from \(v_0\) to \(u\) so that for every vertex \(u' \ne u\) in the path, \(u' \in C_g[i_0] \cup C_g[i_1]\).

    Now for every edge \(\set{w_1,w_2} \in C_g^I(1)\) we need to define a contraction of \(\cont = P_{w_1} \circ (w_1,w_2) \circ P_{w_2}^{-1}\) to the trivial loop around \(v_0\).
    We begin by showing how to contract \(\cont\) assuming that \(w_1,w_2 \in C_g[i_0] \cup C_g[i_1]\) and then we show how to reduce from the general case to this case.
    Observe that \(\cont\) is a cycle of length \(\leq 11\) with at most \(5\) subspaces of dimension \(i_1\) (and these contain all subspaces of dimension \(i_0\)), so the sum of all subspaces in $\cont$ has dimension \(\leq 5i_1\).
    We will contract $\cont$ in four steps:
    \begin{enumerate}
        \item We find an isotropic subspace \(u^*\) of dimension $6i_1$ that is perpendicular to every subspace participating in \(\cont\). Here we use the assumption that the sum of all subspaces in the cycle has low enough dimension \(\leq 5i_1\). It follows that for every \(v \in \cont\), \(u^* + v\) is isotropic.
        \item Fix an arbitrary ``middle vertex'' \(u^{**} \subseteq u^{*}\) such that \(u^{**}\in C_g[i_0]\).
        \item We will connect every vertex \(v_j\) of color \(i_0\) in the cycle to \(u^{**}\) by finding an isotropic subspace \(u_j \in C_g[i_1]\) that contains \(v_j+u^{**} \subseteq u_j\). We can do this since \(u^{**}\) is chosen to be isotropic and perpendicular to \(v_j\).
        \item We separately tile every cycle of the form \((u^{**},u_j,v_j,v_{j+1},v_{j+2},u_{j+2},u^{**})\). Here we use the fact (proven below) that the sum of the subspaces in this cycle is contained in the subspace \(u^* + v_{j+1}\) which is isotropic.
    \end{enumerate}
    The first step is the content of this claim, which we postpone to later.
    \begin{claim} \label{claim:large-perp-space}
        There exists an isotropic subspace \(u_{\bot}\) of dimension \(6i_1\) such that for every \(x \in \bigoplus_{v \in \cont}v\) and \(y \in u_{\bot}\), \(\iprod{x,y} = 0\).
    \end{claim}
    In particular, we can find an isotropic subspace \(u^* \subseteq u_{\bot}\) of dimension \(i_1\) that also intersects all subspaces in \(\cont\) trivially.
    Choose some arbitrary \(i_0\)-subspace \(u^{**} \subseteq u^*\) as in the second step. 
    
    Let us relabel \(\cont = (v_0,v_1,v_2,v_3,\dots,v_m,v_0)\) for \(m \leq 11\), where \(v_{2j} \in C_g[i_0]\) and \(v_{2j+1} \in C_g[i_1]\). We note that for any \(v_{2j}\), there exists some isotropic space \(u_{2j} \in C_g[i_1]\) so that \(v_{2j}, u^{**} \subseteq u_{2j}\): we start from \(u^{**} \oplus v_{2j}\) which is isotropic since \(u^{**},v_{2j}\) are both isotropic and perpendicular to one another, and then we add independent vectors to it from \(u^{*}\) until getting an \(i_1\)-dimensional subspace \(u_{2j}\).

    We denote by
        \begin{align}
            \cont' := &(v_0,u_0,u^{**},u_0,v_0,v_1, \\
            & v_2,u_2,u^{**},u_2,v_2,v_3 \\
            & v_4,u_4,u^{**},u_4,v_4,v_5, \\
            & \dots \\
            & v_m,v_0),
        \end{align}
        i.e. from every \(v_{2j} \in \cont\) we add \((v_{2j},u_{2j},u^{**},u_{2j},v_{2j})\) before going to \(v_{2j+1}\).
        and note that \(\cont \overset{(BT)}{\sim} \cont'\) so we can contract \(\cont'\) instead of \(\cont\). We have just completed the third step.
        
        We note that \(\cont'\) (shifted to start and end at \(w_0\)) is composed from a constant number of loops of the form \(\cont_j = (u^{**},u_{2j},v_{2j},v_{2j+1},v_{2j+2},u_{2j+2},u^{**})\). Thus if we can contract any such loop to the trivial loop with a constant number of steps, then we can find a contraction of all \(\cont\) with a constant number of steps.

        Indeed, fix \(\cont_j\), and note that by construction \(v_{2j+1} \oplus u^* \supset \bigoplus_{y \in \cont_j}y\). This is because \(v_{2j},v_{2j+2} \subseteq v_{2j+1}\) and by construction, the vectors in \(u_{2j},u_{2j+2}\) and \(u^{**}\) all lie in \(v_{2j+1} \oplus u^*\) as well. We note that \(u^* \oplus v_{2j+1}\) is isotropic since both subspaces are isotropic and perpendicular to one another. Thus there is an \(i_2\)-dimensional subspace \(x_j\) that contains all subspaces in \(\cont_j\) (there is a \(g\)-dimensional maximal isotropic space that contains \(v_{2j+1} \oplus u^*\) so we take an \(i_2\)-subspace of it containing this subspace as well. Here we have used $ i_0+i_1\leq 3i_1\leq i_2$). Hence for every edge \((a,b)\) in \(\cont_j\) the triangle \(\set{a,b,x_j}\) is in  \( S^I(2)\). We have shown that with a constant number of triangles \(\cont_j\) could be contracted to \[(w_0,x_j,u_{2j},x_j,v_{2j},x_j,v_{2j+1},x_j,v_{2j+2},x_j,u_{2j+2},x_j,w_0)\]
        which is equivalent to the trivial loop around \(w_0\) by a sequence of \((BT)\) relations.

        \medskip

        For the general case we can do the following contraction to a path that contains only subspaces from \(i_0,i_1\) as above (we also recommend looking at \pref{fig:general-case-reduction}).
        For the case where \(\set{w_1,w_2}\) is such that \(w_2 \in C_g[i_2]\) and \(w_1 \in C_g[i_0]\) we denote by \(w'\) the other neighbor of \(w_2\) in $\cont$ and note that \(w' \in C_g[i_0]\). Thus we can find an \(i_1\)-dimensional (isotropic) subspace \(w_2' \subseteq w_2\) such that \(w_2' \supseteq w_1 + w'\) (by assumption that \(i_1 \geq 2i_0\)). Thus using the triangles \(\set{w_1,w_2,w_2'}, \set{w',w_2,w_2'} \in C_g^I(2)\) we can contract \((w_1,w_2,w')\) to \((w_1,w_2',w')\) removing the subspace \(w_2\) resulting in the previous case (using \(2\) triangles). 
        
        For the case where \(w_1 \in C_g[i_1], w_2 \in C_g[i_2]\) we denote by \(w'' \in C_g[i_0]\) the other neighbor of \(w_1\) in $\cont$. By containment we have that \(\set{w'',w_1,w_2} \in C_g^I(2)\) so we can contract \((w'',w_1,w_2)\) to \((w'',w_2)\) using a single triangle. We then do the same contraction as above (using \(2\) triangles) to remove \(w_2\) and get to a cycle of the same length with vertices only from \(i_0,i_1\). 
        \begin{figure}
            \centering
            \includegraphics[scale=0.34]{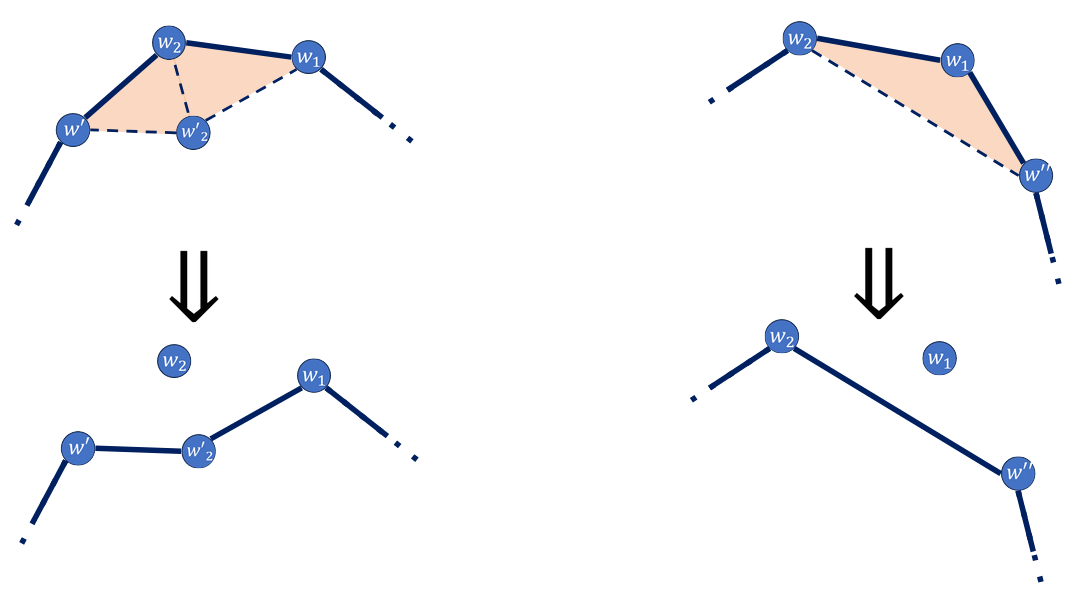}
            \caption{The case where \(w_1 \in C_g[i_0]\) is on the left. The case where \(w_1 \in C_g[i_1]\) is on the right.}
            \label{fig:general-case-reduction}
        \end{figure}
\end{proof}

\begin{proof}[Proof of \pref{claim:large-perp-space}]
    We find a basis for \(u_{\bot}\) one vector at a time as follows. Let \(B_0 = \emptyset\), and for \(j=1,2,\dots 6i_1\) the set \(B_j\) will denote the basis vectors that we found thus far. Let \(t_j = span(B_j) \oplus \left (\bigoplus_{v \in \cont}v \right )\). We note that \(dim(t_j) \leq j + dim(t_0)\). Moreover, \(t_0 = \bigoplus_{v \in \cont}v = \bigoplus_{v \in \cont: v \in S[i_1]}v\) since every subspace in this path is contained in an \(i_1\)-dimensional subspace. As we saw, there are at most \(5\) such spaces, therefore \(dim(t_j) \leq j + 5 i_1 \leq 11i_1 \).

    In particular, the subspace perpendicular to \(t_j\) always contains at least \(2g - 11i_1 \) independent vectors. This is greater or equal \(6i_1\) from the assumption that \(17i_1 \leq 2g\) which is due to the fact that $\set{i_0,i_1,i_2}\in \Col$. Thus given \(B_j\) we take \(B_{j+1} = B_j \dunion \set{x_{j+1}}\) where \(x_{j+1}\) is perpendicular to \(t_j\) and independent from \(B_j\). We note that by construction the vectors of \(B_j\) are perpendicular to one another, so \(u_{\bot}=span(B_{6i_1})\) is indeed an isotropic subspace of dimension \(6i_1\).
\end{proof}

\subsubsection{Bounds for additional color restrictions}
Towards proving swap cocycle expansion we show that various color restrictions also have coboundary expansion. Our goal is to prove the following lemma which we will use later on in \pref{sec:proof-of-faces-complex-lower-bound}. 
\begin{lemma} \label{lem:inductive-bound-on-S-I}
    Let \(I = \set{i_0<i_1<i_2<i_3} \in \binom{[2g]}{4}\) such that \(i_0<i_1-79\).
    Then \[h^1(C_g^I) \geq \exp \left (O \left (\log \left (\frac{i_3}{i_1-i_0} \right) \min \left (\log(\frac{i_1-i_0}{i_3-i_0}), \log (\frac{i_1-i_0}{i_0}) \right ) \right ) \right).\]
\end{lemma}

Our starting point is triples of colors as in \pref{lem:color-restriction-of-symplectic-is-cob-exp}, i.e. colors such that \(i_1 \geq 2i_0, i2 \geq 3i_1\) and \(17i_1 \leq 2g\). Then we will gradually relax the requirements from the the colors until we reach \pref{lem:inductive-bound-on-S-I}. For technical reasons, we will need to introduce a fourth color, since we intend to use \pref{thm:coboundary-expansion-from-colors}. The theorem \pref{thm:coboundary-expansion-from-colors} says that adding another color decreases the coboundary expansion by at most a multiplicative constant.

We start by removing the requirements from \(i_1\).
\begin{claim} \label{claim:color-restriction-lowest-smaller-than-highest}
    Let \(I=\set{i_0,i_1,i_2,i_3}\) be such that \(80i_0 \leq i_3\) then 
    \[ h^1(C_g^I) \geq \Omega \left (\min \set{\frac{i_1-i_0}{i_0}, \frac{i_3-i_1}{i_3}} \right ).\]
\end{claim}

\begin{proof}[Proof of \pref{claim:color-restriction-lowest-smaller-than-highest}]
    First, we notice that by the \pref{thm:coboundary-expansion-from-colors}, \(h^1(C_g^I) = \Omega(C_g^{\set{i_0,i_1,i_3}})\). The reason is that the set \(\set{i_0,i_1,i_3}\) is a constant fraction (\(\frac{1}{4}\)-fraction) of the subsets of \(I\) of size \(3\). Similarly, \(h^1(C_g^I) = \Omega(C_g^{\set{i_0,i_2,i_3}})\). Therefore, we obtain that if \(i_1 \geq 2i_0, i_3 \geq 3i_1\) and \(17i_1 \leq 2g\) then by \pref{lem:color-restriction-of-symplectic-is-cob-exp} \(h(C_g^I) = \Omega(1)\).   
    Let \(I'=\set{i_0,i_1'=2i_0,i_2,i_3}\).
    By \pref{claim:color-swap}, \(h^1(S^I) =\Omega( h^1(S^{I'})) \cdot \min_{v \in {S[i_1']}}h^1(S_v^{I})\).
    By \pref{lem:color-restriction-of-symplectic-is-cob-exp}, \(h^1(S^{I'}) = \Omega(1)\) (if \(i_1' =2i_0 \leq i_3/40\) then in particular \(17i_1' \leq 2g\) so the lemma applies in this case). As for \(S_v^{I}\), note that this is the join of the complex whose vertices are subspaces that are contained in \(v\), and isotropic subspaces that contain \(v\) (for clarity we emphasize that one of these joins may be one sided, i.e. contain only subspaces of \(v\) of dimension \(i_0\), or only isotropic subspaces of dimension \(i_3\)). \pref{claim:cob-exp-of-join} tells us that in this case, the coboundary expansion is bounded by \(O\) of one over the diameter, so let us find the diameter in this case.
    
    If \(i_1' < i_1\) then by \pref{claim:symplectic-diameter} the complex that has isotropic spaces that contain \(v\) has diameter that is at most \(2\) times the diameter of the complex of \emph{all} subspaces that contain \(v\). This complex is isomorphic to the complex that contains subspaces of dimension \(i_1-i_1', i_2-i_1'\) and \(i_3-i_1'\) and by \pref{claim:diam-of-spherical-building} it has diameter at most \(O(\frac{i_3-i_1'}{i_3 - i_1}) =  O(\frac{i_3}{i_3-i_1})\).
    
    Otherwise, the complex that has (all) subspaces contained in \(v\) has diameter \(O(\frac{i_0}{i_1-i_0})\): Indeed, this complex contains all subspaces of dimensions \(i_0,i_1\) inside a space of dimension \(i_1' = 2i_0\). This is isomorphic to the graph of subspaces of dimensions \(2i_0 - i_1, 2i_0-i_0\) inside a space of dimension \(2i_0\), the isomorphism goes from a subspace of dimension \(i_j\) to its perpendicular subspace with respect to the standard bilinear form (this isomorphism \(v \mapsto v^\perp\) takes subspaces of dimension \(x\) to dimension \(i_1'-x\), and reverses the containment relation). Thus its diameter is \(O(\frac{2i_0 - i_0}{(2i_0-i_0)-(2i_0-i_1)}) = O(\frac{i_0}{i_1-i_0})\) by \pref{claim:diam-of-spherical-building}.
    
    In both cases the diameter is at most the maximum between the two expressions. By \pref{claim:cob-exp-of-join}, \(h^1((C_g)_v^{I}) = \Omega(\min\set{\frac{i_0}{i_1-i_0},\frac{i_3-i_1}{i_3}})\). The claim follows.
\end{proof}

Our main effort in going from \pref{claim:color-restriction-lowest-smaller-than-highest} to \pref{lem:color-restriction-of-symplectic-is-cob-exp} is to relax the inequality on \(i_0\). To do so we will apply \pref{claim:color-swap} multiple times to go from \(i_0\) to \(i_0'\), to \(i_0''\) etc. until we finally reach a complex for which we can apply \pref{claim:color-restriction-lowest-smaller-than-highest}. Given \(i_0\), we would like to choose \(i_0'\) as small as possible while allowing us to use \pref{claim:color-swap}. It will be apparent in the proof of \pref{lem:color-restriction-of-symplectic-is-cob-exp} that the following function \(T(x)\) will give us the smallest possible index.

Let \(T(x)=T_{i_3}(x) = \max \set {1,\left \lceil \frac{80x-i_3}{79} \right \rceil}\). We denote by \(T^m(x)\) the $m$-fold composition of \(T\) (and let \(T^0(x)=x\)). Let us note that for every \(x < i_0\), \(T(x) < x\), which implies that \(T^{m+1}(x) \leq T^m (x)\) for every \(m \geq 0\), and the inequality is sharp while \(T^m(x) \ne 1\).

We record the following claim that bounds the number \(T\)'s needed to go from \(x=i_0\) to \(x=1\). We prove it after proving the lemma.
\begin{claim} \label{claim:bound-on-number-of-iterations}
    Let \(n = n(I)= \log \left ( \frac{i_3}{i_3-i_0} \right )\). Then \(T^{n}(i_0) = 1\).
\end{claim}
\begin{proof}[Proof of \pref{lem:inductive-bound-on-S-I}]
    Let \(n(I) = \log \left ( \frac{i_3}{i_3-i_0} \right )\). By \pref{claim:bound-on-number-of-iterations},  \(T_{i_3}^{n(I)}(i_0) = 1 \leq 80i_3\). 
    In every iteration we will show the following guarantee below.
    \begin{proposition}\label{prop:one-step-in-induction}~
        Let \(c = \Omega(\min \set{\frac{i_1-i_0}{i_0}, \frac{i_3-i_1}{i_3}})\). Then for every \(m \geq 0\),
        \[h^1(C_g^{\set{T^m(i_0),i_1,i_2,i_3}}) \geq c \cdot h^1(C_g^{\set{T^{m+1}(i_0),i_1,i_2,i_3}}).\]
    \end{proposition}
    By using \pref{prop:one-step-in-induction} \(n(I)\) times we have that
    \begin{align*}
        h^1(C_g^I) &\geq c^{n(I)} h^{1}(S^{1,i_1,i_2,i_3}) \\
        &\geq \frac{i_1}{i_3} \cdot c^{n(I)}\\
        &\geq \exp \left (\log \left (\frac{i_1}{i_3} \right ) + n(I)\log c \right) \\
        &\geq \exp \left (O \left (\log \left (\frac{i_3}{i_1-i_0} \right) \min \left (\log(\frac{i_3-i_1}{i_3}), \log (\frac{i_1-i_0}{i_0}) \right ) \right ) \right).        
    \end{align*}
        
The second inequality is due to \pref{claim:color-restriction-lowest-smaller-than-highest} (where we also use the fact that \(\frac{i_1}{i_3} \leq \frac{i_1-1}{1}\) is the minimum in the expression).

\begin{proof}[Proof of \pref{prop:one-step-in-induction}]
        Fix \(m \geq 0\). and let \(I_m = \set{T^m(i_0),i_1,i_2,i_3}\), \(I_{m+1} = \set{T^{m+1}(i_0),i_1,i_2,i_3}\). The complex \(S^{I_m}\) is a high dimensional expander so by \pref{claim:color-swap}
        \[h^1(C_g^{I_m}) \geq h^1(C_g^{I_{m+1}}) \Omega(\min_{v \in {C_g[T^{m+1}(i_0)]}}h^1((C_g)_v^{I_m})).\]
        It remains to show that for every \(v \in C_g[T^{m+1}(i_0)]\), \(h^1((C_g)_v^{I_m}) = \Omega(\min \set{\frac{i_1-i_0}{i_0}, \frac{i_3-i_1}{i_3}})\).
        Fix \(v \in C_g[T^{m+1}(i_0)]\) and denote by \(J = \set{j_0,j_1,j_2,j_3}\) where \(j_0=T^{m}(i_0) - T^{m+1}(i_0)\) and for 
        \(t=1,2,3\), \(j_t = i_t - T^{m+1}(i_0)\). Let us here recall that \(T^{m+1}(i_0) \leq T^m(i_0)\) so these constants are indeed positive.
        We have seen in the second item of \pref{prop:isomorphism-of-symplectic-link} that \((C_g)_v^{I_{m+1}} \cong (C_{g'})^J\) where \(C_{g'}\) for \(g' = g-T^{m+1}(i_0)\).
        
        Let us show that we can use \pref{claim:color-restriction-lowest-smaller-than-highest} on \(C_g'\), i.e. that \(80j_0 \leq j_3\).
        By definition of \(T\),
        \[79 T^{m+1}(i_0)=79 T(T^{m}(i_0)) = \max \set{ 79\left \lceil \frac{80T^m(x)-i_3}{79} \right \rceil, 79} \geq 80T^{m}(i_0)-i_3\]
        where the inequality comes from the first term in the maximum.
        Equivalently this implies that \(80T^m(i_0) \leq 79T^{m+1}(i_0)+i_3\) so 
        \[80j_0 = 80\left( T^{m}(i_0) - T^{m+1}(i_0) \right) \leq i_3 - T^{m+1}(i_0) = j_3.\]
        Thus
        \(80j_0 \leq j_3\) and we are justified to use \pref{claim:color-restriction-lowest-smaller-than-highest} and deduce that \(h^1((C_g)_v^I) = \Omega( \min \set{\frac{j_1-j_0}{j_0},\frac{j_3-j_1}{j_3}} ) = \Omega( \min \set{\frac{i_1-i_0}{i_0}, \frac{i_3-i_1}{i_3}} )\). 
        The inequality \(\min \set{\frac{j_1-j_0}{j_0},\frac{j_1}{j_3}} ) = \Omega( \min \set{\frac{i_1-i_0}{i_0}, \frac{i_3-i_1}{i_3}} )\) uses the following:
        \begin{enumerate}
            \item The inequality \(\frac{j_1-j_0}{j_0} = \frac{i_1-i_0}{j_0} \geq \frac{j_1-j_0}{i_0}\).
            \item The inequality \(\frac{j_3-j_1}{j_3} = \frac{i_3-i_1}{i_3 - T^{m+1}(i_0)} \geq \frac{i_3-i_1}{i_3}\).
        \end{enumerate}
        In both inequlities we just increased the denominator.
    \end{proof}
\end{proof}

It remains to prove \pref{claim:bound-on-number-of-iterations}.
\begin{proof}[Proof of \pref{claim:bound-on-number-of-iterations}]
        Let \(\tilde{T}(x) = \frac{80x-i_3}{79}+1\). We note the following facts that allow us to bound the number of iterations needed to apply on \(\tilde{T}\) instead of \(T\):
        \begin{enumerate}
            \item The inequality \(T(x) \leq \tilde{T}(x)\) holds for every \(x<i_0\) while \(T(x) >1\).
            \item while \(x \leq i_0\), \(\tilde{T}(x)<x\).
            \item One can show by induction on \(m\) that while \(T^m(x) > 1\) (and therefore for all \(m'<m\), \(T^{m'}(x) \geq T^m(x)>1\)), \(T^m(x) \leq \tilde{T}^m(x)\). Indeed, for \(m=1\) this holds by the above observation. Assuming for \(m\) the proof for \(m+1\) is direct-
            \[T^{m+1}(x) \leq \tilde{T}(T^m(x))\leq \tilde{T}(\tilde{T}^m(x))=\tilde{T}^{m+1}(x).\]
        \end{enumerate}
        Thus if we show that after \(m=O(\log \frac{i_3}{i_3-i_0})\), it holds that \(\tilde{T}^m(i_0) \leq 1\), then it follows that \(T^m(i_0) = 1\).
        
        Indeed, solving a recursion relation yields
        \(\tilde{T}^m(i_0) = (i_0+79-i_3) \left (\frac{80}{79} \right )^m + (i_3-79)\). The term \(i_0+79-i_3<0\) by the assumption on \(i_0\). Thus,
        \[\tilde{T}^m(i_0) \leq 1 \Leftrightarrow m \geq \log_{\frac{80}{79}}\frac{i_3-80}{i_3-i_0+79} .\]
        The constant \(m \geq \log_{\frac{80}{79}}\frac{i_3-80}{i_3-i_0+79}\) so we are done.
\end{proof}

\subsubsection{Coboundary expansion of bounds of joins of spherical buildings} \label{sec:spherical-building-color-complexes}
In the next subsection we will prove swap coboundary expansion of the faces complexes of links of the affine symplectic building. These links are (isomorphic to) joins of complexes, some of type \(A\) and some of type \(C\). Towards a bound on the faces complex of some link, we will need to bound the coboundary expansion of the link itself (and its sub-links).

Towards this, let us introduce some notation for joins of complexes. Let \(\S_1,\S_2,\dots,\S_k\) be spherical buildings of either type \(A\) or type \(C\), of dimensions \(\ell_1,\ell_2,\dots, \ell_m\). Let \(\S = \bigvee_{i=1}^k \S_i\). Every such building is \((\ell_i+1)\)-partite, so \( \S \) is also a partite complex. We denote its colors by \((i,j)\) where \(i\) indicates the subcomplex \(\S_i\) and \(j\) indicates the color inside \(\S_i\).
For two distinct colors $(i_1,j_1)$ and $(i_2,j_2)$ define their distance to be the number of colors between them, \[dist((i_1,j_1),(i_2,j_2)) = |\sett{(i,j)}{(i_1,j_1)< (i,j) \leq (i_2,j_2)}|.\] 

\begin{definition}
    A set of colors $\set{(i_1,j_1)<\ldots<(i_r,j_r)}$ is {\em $(a_1,a_2)$-spread} if the distance of every two consecutive colors is between $a_1$ and $a_2$. Namely, letting $(i_0,j_0) = (0,0)$,
\[\forall t=0,\ldots, r-1,\qquad
a_1< dist((i_t,j_t),(i_{t+1},j_{t+1})) \leq a_2.
\]
\end{definition}
Let us denote \(\sum_{i=1}^k (\ell_i+1) = \d_1\) and fix \(\S\) to be the above join.
We prove the following lemma.
 
\begin{lemma}\label{lem:join-links}
Let $a_1=\frac{n}{(m(\d_1+1))^3}$ and $a_2 =\frac{100n \log (\d_1+1)}{(\d_1+1) m}$ and let
$c=\set{(i_1,j_1)<\ldots<(i_r,j_r)}$ be an $(a_1,a_2)$-spread set of colors. Let  $I \subset c$ be a set of $5$ colors, and let $w\in \S$ be such that $col(w) \subset {c\setminus I}$.  
Then  $h^1(\S_{w}^I)\geq exp(-O(\log ^2(\d_1)))$.
\end{lemma}
\begin{proof}[Proof of \pref{lem:join-links}]
    It is enough to prove that for any \(I' = \set{(i_1,j_1) < (i_2,j_2)< (i_3,j_3)<(i_4,j_4)}\subset I\), \(\S_{w}^{I'}\) is a \(\beta\)-coboundary expander. By \pref{thm:coboundary-expansion-from-colors}, it follows that \(h^1(\S_{w}^{I}) \geq \Omega(\beta)\). So from now on we just re-annotate $I = \set{(i_1,j_1) < (i_2,j_2)< (i_3,j_3)<(i_4,j_4)}$.
    The coboundary expansion $h^1(\S^I_{w})$ depends on $I$ and $w$. We address first the easier ``direct'' cases, and then move to the general case which is gradually reduced to the easier cases, via decomposition steps.

    If \(\S_{w}^{I}\) is a join of three complexes (or more), i.e. we can write \(\S_{w}^{I} = A_1 \vee (A_2 \vee A_3)\). The diameter of \((A_2 \vee A_3)\) is constant so by \pref{claim:cob-exp-of-join} \(h^1(\S_w^I) \geq \Omega(\frac{1}{diam(A_2\vee A_2)}) = \Omega(1)\). If \(\S_{w}^{I}\) is a join of exactly two complexes then one of the complexes is a color restriction of a spherical building. Without loss of generality let us assume that \((i_1,j_1),(i_2,j_2)\) belong to the same complex (i.e. \(i_1=i_2\)) and let us assume that the complex has dimension \(t\). By \pref{claim:diam-of-spherical-building} and \pref{claim:symplectic-diameter} the diameter of this complex is \(O(\frac{j_2}{j_2-j_1})\). The dimension of the complex \(t\) is bounded from above by the maximal distance of two consecutive colors of \(\cup s\) which is at most \(a_2=\frac{100n \log (\d_1+1)}{(\d_1+1) m}\) (and the link of \(w\) can only split to more complexes of lower dimension). On the other hand $j_2-j_1$ is at least $a_1=\frac{n}{(m(\d_1+1))^3}$ by assumption. 
    Therefore the diameter is bounded by \(poly(\d_1)\). By \pref{claim:cob-exp-of-join} in this case \(h^1(\S_{w}^{I}) =\exp(-O(\log(\d_1)))\).

    If \(\S_{w}^{I}\) is a color restriction of a single type \(A\) spherical building then by \pref{lem:general-case-subspace-complex-sl} \[h^1(\S_w^I) \geq \exp \left (-O\left ( \log \left (\frac{j_3}{j_1-j_0} \right ) \cdot \log\left (\frac{j_3}{j_1}\right ) \right ) \right ).\]
    Similarly, if \(\S_{w}^{I}\) is a color restriction of a single type \(C\) spherical building then by \pref{lem:inductive-bound-on-S-I} \[h^1(\S_{w}^{I}) \geq \exp \left (O \left (\log \left (\frac{j_3}{j_1-j_0} \right) \min \left (\log(\frac{j_3-j_1}{j_3}), \log (\frac{j_1-j_0}{j_0}) \right ) \right ) \right).\]
    As before, the quantities \(\frac{j_3}{j_1},\frac{j_3}{j_1-j_0},\frac{j_3}{j_3-j_1}\) and \(\frac{j_1-j_0}{j_0}\) are \(\poly(\d_1)\) from the spreadness assumption, so in both cases \(h^1(\S_{w}^{I}) \geq \exp(-O(\log^2(\d_1)))\).
\end{proof}

%% file: swap.tex
\subsection[Swap coboundary expansion of the of tilde C]{Swap coboundary expansion of the of the links of \(\tilde{C}\)} \label{sec:proof-of-faces-complex-lower-bound}
In this section we modify the proof in \cite{DiksteinD2023swap} to show that the finite quotients of the affine symplectic building's faces complex is a \((\exp(-O(\sqrt{r})),r)\)-swap coboundary expander. The proof follows the same lines as \cite{DiksteinD2023swap}, where the main difference is that in some cases we need to use \pref{lem:inductive-bound-on-S-I} instead of \pref{lem:general-case-subspace-complex-sl}. 

As we saw in Section \pref{sec:preliminaries}, the link of every \(j\)-face in a quotient of an affine building is a join of \(j' \leq j+2\) spherical buildings (as in the definition of a join in \pref{sec:preliminaries}). If the building is symplectic then these buildings are either symplectic or special linear. The following theorem deals with swap coboundary expansion of such complexes.

\begin{theorem}\label{thm:coboundary-expansion}
   Let \(\d_1\) be an integer. There is some \(p_0=p_0(d)\). Let \(p > p_0\) be any prime power. Let \(k \geq 1\) and let \(\set{\S_i}_{i=1}^k\) be so that for every \(i=1,2,\dots,k\), \(\S_i\) are either \(SL_{\ell_i}(\mathbb{F}_p)\) or \(Sp(2\ell_i,\mathbb{F}_p)\) spherical buildings. 
   Assume that \(\sum_{i=1}^k \ell_i = n \geq \d_1^5\). \(\S = \bigvee_{i=1}^k \S_i\). Let \(\X = \FS[\d_1]\) be its faces complex. Then \(\X\) is a coboundary expander and \(h^1(\X) \geq \exp(-O(\sqrt{\d_1}))\).
\end{theorem}
From this theorem we immediately derive swap cocycle expansion of the quotients of the affine symplectic building.
\restatetheorem{thm:swap-coboundary-expansion}
\begin{proof}[Proof of \pref{thm:swap-coboundary-expansion} from \pref{thm:coboundary-expansion}]
    Let \(\X = \FX[\d_1]\). For every \(s \in \X(0)\), \(\X_s\), being itself a faces complex of \(X_s\), is a faces complex of a complex that satisfies the conditions of \pref{thm:coboundary-expansion} so it is a \(\exp(-O(\sqrt{\d_1}))\)-coboundary expander. In addition, for \(p_0\) large enough \(\X\) is a sufficiently good spectral high dimensional expander, so by \pref{lem:trick} (applied to the two skeleton of \(\X\) for \(i=0\)) , it holds that \(h^1(\X) \geq \exp(-O(\sqrt{\d_1}))\) (as a cocycle expander).
\end{proof}

\subsubsection*{Notation for this section}
Fix $\S = \bigvee_{i=1}^k \S_i$ and $d,n$ as in the theorem statement. Denote by \(\X = \FS[\d_1]\) and \(\tilde{\X} = \FS[]\). Fix $m = \sqrt{\d_1+1}$.

All \(\S_i\)'s are partite and come with colors associated with the dimension of the subspace.
The color of a vertex $v$ in $\S$ (denoted by $col(v) = col_{\S}(v)$), is a pair \((i,j)\) such that \(v \in S_i\) is a subspace of dimension \(j\) (i.e. \(j=col_{\S_i}(v)\)). We let \(C_0 = \sett{(i,j)}{i \in [k], j \in [\ell_i]}\) be the possible colors of \(\S\) and we order the colors lexicographically, that is \((i,j) \leq (i',j)\) if \(i \leq i'\) or \(i=i'\) and \(j\leq j'\).
We set $\mathcal{C}= \binom{C_0}{\d_1+1}$ be the set of possible colors of vertices of $\X$.

We use $u,v$ to denote vertices of $\S$, and $w$ to denote vertices of $\X$, which are faces of $\S$. Faces of $\X$ are denoted by $s$.
We denote subsets of colors of $\FS[]$ that are mutually disjoint by the letters $J,I$ (so $J,I\in\FD[]$).

\subsubsection*{Well spread colors}
\label{sec:well-spread-colors}
In light of \pref{lem:colorest}, it suffices to show swap coboundary expansion of certain color sets of the quotients of the affine building, in order to deduce swap coboundary expansion of the quotients themselves. We now describe which colors interest us. These are exactly the same colors which were used in \cite{DiksteinD2023swap}.

Let \(N\) be an ordered set. Let \(c = \set{i_1 < i_2 < \dots < i_T} \subseteq N\) be any subset. A \(c\)-bin is one of the following sets
\[B_0 = \sett{i \in N}{i < i_1}, \; B_{T} = \sett{i \in N}{i > i_T}\]
or
\[\forall t = 1,2,\dots,T-1 \qquad B_t = \sett{i \in N}{i_t < i < i_{t+1}}.\]
Let \(J = \set{c_1,c_2,\dots,c_m}\) be mutually disjoint and disjoint from \(c\). We say that a \(c\)-bin \(B_t\) is \(J\)-crowded if there are two distinct \(c_{j_1},c_{j_2} \in J\) such that \(B_t \cap c_{j_1}, B_t \cap c_{j_2} \ne \emptyset\). If there is only a single \(c_j \in J\) such that \(B_t \cap c_j \ne \emptyset\) we say that \(B_t\) is \(J\)-lonely. Otherwise, if for all \(c_j \in J\), \(B_t \cap c_j = \emptyset\) we say that \(B_t\) is \(J\)-empty.

We define a well-spread color to have good pseudo-random properties, that is, all indices are roughly equally spaced, and interlaced with one another so that many colors will be isolated. This will facilitate the lower bounds in the next sections.

Let \(\mathcal{C} = \binom{N}{d+1}\). Let \(J \subseteq \mathcal{C}\). Recall that \(\dunion J = \bigcup_{c \in J}c\).

\begin{definition}[Well-spread subsets of colors] \label{def:good-colors}
    Let \(m > 5\) and let \(J\) be a set of \(m\) colors in \(\mathcal{C}\). We say that \(J\) is \emph{well-spread} if the following properties hold.
    \begin{enumerate}
        \item Every \(c_1,c_2 \in J\) are disjoint.
        \item Renaming the colors \(N = \set{0,1,\dots,n}\) (with the usual order), for every \(\ell_1,\ell_2 \in (\cup J) \cup \set{0,n}\) it holds that \(\Abs{\ell_1 - \ell_2} \geq \frac{n}{(m(\d_1+1))^3}\).
        \item For every \(J' \subseteq J\),  \(\Abs{J'} = 5\), let \(c^*=\cup (J \setminus J')\):
        \begin{enumerate}
            \item Every \(c^*\)-bin has size at most \(\frac{100n \log (\d_1+1)}{(\d_1+1) m}\). 
            \item For every \(c \in J'\), the number of colors \(i \in c\) that are in \(J'\)-crowded \(c^*\)-bins is at most \(\frac{100(\d_1+1) \log (\d_1+1)}{m \log m}\).
            \item For every \(c \in J'\) and every \(c^*\)-bin \(B\), it holds that  \(\abs{B \cap c } \leq 20\frac{\log (\d_1+1)}{\log m}\).
        \end{enumerate}
    \end{enumerate}
We denote by $\mathcal{J}\subset\FD[\d_1]$ the set of well-spread color sets.
\end{definition}
The following proposition was proven in \cite{DiksteinD2023swap}.
\begin{proposition}[{\cite{DiksteinD2023swap}}] \label{prop:prob-of-good-colors-tend-to-one}
    Let \(\d_1\) be an integer. Let \(6 \leq m \leq \d_1+1\). The probability that \(m\) uniformly chosen colors out of \(n\) colors are \emph{well-spread} tends to \(1\) as \(\d_1,n \to \infty\) as long as \(\d_1^5 \leq n\).
\end{proposition}

\subsubsection*{Roadmap and proof of \pref{thm:swap-coboundary-expansion}}
Let us explain the proof idea. Swap coboundary expansion of $\S$ is, by definition, coboundary expansion of $\X$. We begin with two reductions.
\begin{enumerate}
    \item Using \pref{lem:colorest} we deduce that to show \(\exp(-\Omega(\sqrt{d}))\)-coboundary expansion of \(\X\) it is enough to show coboundary expansion of \(\X^J\) where \(J\) is a set of well-spread colors as in \pref{def:good-colors}.
    \item Using \pref{lem:trick} we deduce that in order to show \(\exp(-\Omega(\sqrt{d}))\)-coboundary expansion of \(\X^J\) it is enough to show coboundary expansion of \(\X_w^J\) for every \(w \in \X^J(m-6)\).  
\end{enumerate}
The reason we reduce to a link of a complex that has well spread colors, is because we can decompose such a link to a tensor product of a complete partite complex and a ``remainder'' complex which itself is a faces complex with \(5\)-colors of size \(O(\sqrt{\d_1})\) (see \pref{claim:tensoring-corollary}).  \pref{claim:triangle-complex} allows us to bound the expansion of the complete tensor part, and \pref{prop:colored-exponential-decay-bound} gives an exponential lower bound on the remainder part. Note that the exponent is \(\sqrt{\d_1}\) since this is the size of the remaining color sets. In order to use \pref{prop:colored-exponential-decay-bound} we need to bound the coboundary expansion of  (color-restrictions of) the links of \(\S\). For this we turn to \pref{lem:join-links} proven in the previous subsection.

Using the quasi-polynomial bounds we get in \pref{lem:join-links} together with \pref{prop:colored-exponential-decay-bound}, we get a bound of
\begin{equation} \label{eq:weaker-bound}
    h^1(\X) \geq \exp(-\sqrt{\d_1} \poly(\log \d_1)).
\end{equation}
This bound is almost as strong as claimed in \pref{thm:swap-coboundary-expansion}, and is already strong enough for proving \pref{thm:main}. However, we claimed in \pref{thm:coboundary-expansion} (and in \pref{thm:swap-coboundary-expansion}) a slightly stronger bound that doesn't suffer from the the poly-logarithmic factors. Proving this stronger bound requires a more complicated version of \pref{prop:colored-exponential-decay-bound}. We do this in full in \pref{sec:full-version}. We encourage the readers to go over the proof below of the weaker bound \eqref{eq:weaker-bound}, and save \pref{sec:full-version} for a second read.

\begin{proof}[Proof of Equation \pref{eq:weaker-bound}]
To bound $h^1(\X)$ we follow the steps of the decomposition.
Let $\J$ be the set of well-spread $J$'s per \pref{def:good-colors}. By \pref{prop:prob-of-good-colors-tend-to-one}, at least half of the sets $J$ are in $\J$. Therefore, by \pref{lem:colorest},
\begin{equation}\label{eq:colors}
        h^1(\X) \geq \Omega(1)\cdot \min_{J\in\J} h^1(\X^J) .
\end{equation}
Fix $J\in\J$. We note that every link of $\X$ is simply connected by \pref{claim:trivial-simple-connectivity} (proven in \pref{sec:simple-connectivity}). We use \pref{lem:symplectic-local-spectral-expansion} to deduce that such a complex is a sufficient local spectral expander for large enough prime \(p\), so by \pref{lem:trick}
\begin{equation}\label{eq:trickle}
        h^1(\X^J) \geq \exp(-O(m))\cdot \min_{s\in \X^J(m-6)} h^1(\X_s^J) .
\end{equation}
The well-spreadness of $J$ allows us to bound the right hand side by a much smaller color-restriction,
\begin{claim} \label{claim:tensoring-corollary}
    Let \(J\) be a set of well-spread colors and let \(s\in \X^J(m-6)\). Then there exists a set of colors \(\tilde{J} = \set{\tilde{c}_1,\tilde{c}_2,\dots,\tilde{c}_5} \leq J\) so that \(\sum_{j=1}^5 \abs{\tilde{c}_j} = O \left (\frac{\d_1 \log \d_1}{m \log m} \right )\) and
    \begin{equation}\label{eq:tensor}
        h^1(\X_s^J) \geq const \cdot  h^1(\tilde \X^{\tilde J}_s).
    \end{equation}
\end{claim}
Next, we wish to use \pref{prop:colored-exponential-decay-bound} to bound the coboundary expansion of $\tilde \X_s^{\tilde J} = F\S_{\cup s}^{\tilde J}$. For this we need to bound $\beta = \min_{w,I} h^1(\S_{\cup s \cup w}^I)$ 
where $I = \set{i_1,\ldots ,i_5}$ such that $i_j\in \tilde c_j$ for $j=1,\ldots,5$ and where $w\in\S_{\cup s}^{\cup \tilde J \setminus I}$.
By \pref{prop:colored-exponential-decay-bound},
\begin{equation}\label{eq:GK}
        h^1(\tilde \X_s^{\tilde J}) \geq const\cdot(\beta_1)^R
\end{equation}
where \(\beta_1 = \Omega(\beta)\) and $R = \sum_{j=1}^5 |\tilde{c}_j| =  O \left (\frac{\d_1 \log \d_1}{m \log m} \right )$. 

By \pref{lem:join-links},
\begin{equation}
    \beta = \min_{w,I} h^1(\S_{\cup s \dunion w}^I) \geq \exp(-O(\log ^2 \d_1))
\end{equation}
We are justified in applying this lemma since the well-spreadness of $J$ implies that $\cup  J$ is $(\frac{n}{(m(\d_1+1))^3},\frac{100n \log (\d_1+1)}{(\d_1+1) m})$-spread. 
We now plug in each equation into the previous one, to get the desired bound,
\[ 
    h^1(\X) \geq const \cdot \exp \left (- O\left(m + \frac{\d_1\poly\log \d_1}{m\log m}\right) \right ) = \exp(-O(\sqrt{\d_1}\log^2 \d_1)).
\] 
\end{proof}
\subsubsection{Links of the spherical building} \label{sec:links-of-colored-faces-complex} 
In the next two subsections we proceed towards proving \pref{claim:tensoring-corollary}. Towards this, let us understand how a link of a join of spherical bulding looks like.

Let \(v \in \S_i(0)\) be a vertex, 
 and let $j=col(v)$. We can write \((S_i)_v = (\S_i)_v[ {[j-1]}]  \vee (\S_i)_v[ {[\ell_i] \setminus [j]}] \). The reason is that \((\S_i)_v[ {[j-1]}]\) consists of the subspaces contained in \(v\), and \((\S_i)_v[ {[\ell_i] \setminus [j]}]\) consists of the subspaces that contain \(v\), and containment is transitive. This holds for spherical buildings of types both $A$ and $C$, where in the later case we only consider isotropic subspaces.

Therefore, it is immediate to see that the join \(\S\) has the same structure, joined with the other complexes, i.e.
\[ \S_v = \left ( \bigvee_{t=1 ,t \ne i}^k S_t \right ) \vee (\S_i)_v^{ {[j-1]}}  \vee (\S_i)_v^{ {[\ell_i] \setminus [j]}}).\]
 We observe that in particular, if \(col(v)=(i,j)\), then we can write \(\S_v\) as a join of two complexes: one contains vertices whose colors are \(<(i,j)\) and the other contains vertices whose colors are \(> (i,j)\).

Let us understand how this generalizes to links of arbitrary faces. Fix a general face $w = \set{ v_1,\cdots,v_T}\in \S$ and let us study $\S_w$. Let the colors of $w$ be \(c = \set{(i_1,j_1) <(i_2,j_2) < \dots < (i_T,j_T)}\). Recall the notion of a \(c\)-bin from \pref{sec:well-spread-colors}.  

We can write \(\S_w\) as 
\begin{equation}\label{eq:decomposition-of-link}
    \S_w = \bigvee_{t=0}^T \S_w^{B_t}
\end{equation}
where the \(B_t\)'s are \(c\)-bins as above (and it is possible that \(\S_w^{B_t}\) is itself also a join of complexes).

\subsubsection{A tensor decomposition of the faces complex of a join}\label{sec:tensor}
Let \(w \in \S(T-1)\), let $c= col(w)$ and let the $c$-bins be $B_0,\ldots,B_T$ as in \pref{sec:well-spread-colors}. Let \(\S_w^{B_t}\) for $t=0,\ldots,T$ be the components of the decomposition of \(\S_w\) as in \eqref{eq:decomposition-of-link}. Let \(J = \set{c_1,c_2,\dots,c_m}\) be subsets of mutually disjoint colors in \(\X\) that are also disjoint from \(c\). We denote by \(c_j^t = c_j \cap B_t\), and let \(J_t = \set{c_1^t,c_2^t,\dots,c_m^t}\) be the corresponding subsets of mutually disjoint colors in \(\S_w^{B_t}\) (technically this should be a multiset but only the empty set can appear more than once). 
\begin{claim} \label{claim:tensor-decomposition-of-faces-complex}
    \begin{equation} \label{eq:join-faces-decomposition}
        \tilde{\X}_w^J = \bigotimes_{t=0}^T \tilde{\X}_w^{J_t}.
    \end{equation}    
\end{claim}
Verifying this is a direct calculation:
\begin{proof}
    Recall that by \eqref{eq:decomposition-of-link}, \(\S_w = \bigvee_{t=0}^T \S_w^{B_t}\). In particular, for every \(c_j\),
    \[\S_w[c_j] = \prod_{t=0}^T \S_w^{B_t} [c_j^t]\]
    since specifying a \(c_j\)-colored face in \(\S_w\) corresponds to independently specifying a face of color \(c_j^t\) in each part of the join (and taking disjoint union). 
    
    Next consider a top-level face \(s = \set{w_1,w_2,\dots, w_m} \in \tilde{\X}_w[J]\) so that $col(w_j) = c_j$ for each $j\in [m]$. Let $\bar w = \dunion_{j\in[m]} w_j$ and let $\bar c = \dunion_{j\in[m]} c_j = \cup J$.
    Specifying $s$ corresponds to specifying a \(\bar c\)-colored face \(\bar w \in \S_w[\cup J]\) and then partitioning it to $m$ parts according to the colors $J$. By properties of the join, this is the same as sampling \(s_t=\set{w_1^t,w_2^t,\dots,w_m^t} \in \S_w[J_t]\) for every \(t=0,1,\dots,T\) and taking the partite disjoint union of these faces (i.e. \(w_j = w_j^0 \dunion w_j^1 \dots \dunion w_j^T\)). This is, in turn, the same as sampling a top-level face in \(\bigotimes_{t=0}^T \tilde{\X}_w^{J_t}\). Complexes with the same vertex sets and same distribution on top level faces are equal and the claim is proven.
\end{proof}

We can refine this decomposition by separating the bins to empty, lonely, and crowded, as defined in \pref{sec:well-spread-colors}. Recall that a \(c\)-bin \(B_t\) is \(J\)-crowded if there are two distinct \(c_{j_1},c_{j_2} \in J\) such that \(c_{j_1}^t,c_{j_2}^t \ne \emptyset\). If there is exactly one \(c_j^t \ne \emptyset\) then we say that \(B_t\) is \(J\)-lonely and if all \(c_j^t = \emptyset\) we say that \(B_t\) is \(J\)-empty. 

Let us use this separation to crowded and lonely/empty bins to prove \pref{claim:tensoring-corollary}.
\begin{proof}[Proof of \pref{claim:tensoring-corollary}]
    Let \(s \in \X^J(m-6)\), let $w = \dunion s$ and recall that \(\tilde\X_s = \tilde\X_{w}\). Let $c=col(w)$ and let $B_0,\ldots,B_T$ be the $c$-bins which we partition into crowded and not crowded.
    Let
        \[I_1 = \sett{0 \leq t \leq T}{B_t \text{ is \(J\)-crowded}}\] 
    and let \(I_2 =  \set{0,\ldots,T} \setminus I_1\). We can write \eqref{eq:join-faces-decomposition} as
        \[ \tilde{\X}_{w}^J = \left (\bigotimes_{t \in I_1} \tilde{\X}_{w}^{J_t} \right ) \otimes \left ( \bigotimes_{t \in I_2} \tilde{\X}_{w}^{J_t} \right ).\]
    For every \(t \in I_2\), \(J_t\) has at most one non-empty set of colors, so \(\tilde{\X}_{w}^{J_t}\) is a complete partite complex. Therefore, \(\bigotimes_{i \in I_2} \tilde{\X}_w^{J_i}\) is also a complete partite complex. 
    For every \(c_j \in J \setminus col(s)\), Let \(\tilde{c}_j = c_j \cap \bigcup_{t \in I_1} B_t\) and let \(\tilde{J} = \set{\tilde{c}_1,\tilde{c}_2,\dots \tilde{c}_5}\). 
    By \pref{claim:triangle-complex} (and the fact that this is a sufficiently local spectral expander for large enough primes \(p\) \pref{lem:symplectic-local-spectral-expansion}),
        \[h^1(\tilde{\X}_{s} ) \geq const \cdot h^1 \left ( \bigotimes_{t \in I_1} \tilde{\X}_{w}^{J_t} \right ) = const  \cdot h^1(\tilde{\X}_{s}^{\tilde{J}}).\]
        
        Finally, by definition of well spread colors, the number of crowded bins for every \(\tilde{c}_j\) is \( O \left ( \frac{\d_1 \log \d_1}{m \log m} \right )\) and $|\tilde C_j\cap B_t|=O(1)$ by \pref{def:good-colors}, item 3(c). Thus \(\sum_{j=1}^5 \abs{\tilde{c}_j} = O \left ( \frac{\d_1 \log \d_1}{m \log m} \right )\).
\end{proof}

\subsubsection{Simple connectivity of the links} \label{sec:simple-connectivity}
\begin{claim} \label{claim:trivial-simple-connectivity}
    Let \(J\) be a set of well-spread colors. For every \(i \leq \d_1\) and \(s \in \X^J(i)\), the complex \(\X_s^J\) is simply connected. 
\end{claim}
\begin{proof}[Proof of \pref{claim:trivial-simple-connectivity}]
Showing simple connectivity is equivalent to showing that the complex is a coboundary expander with some positive constant. We do so using \pref{prop:colored-exponential-decay-bound} on \(\X_s^J\).  

Recall that \(\X_s^J= F^J(\S_{ w_1})\), where $w_1=\dunion s$. By \pref{lem:join-links}, for every \(I = \set{i_1,i_2,i_3,i_4,i_5}\) such that \(i_j \in c_j\) and \(w \in \S_{w_1}^{\cup J \setminus I}\), we have \(h^1(\S^I_{w_1 \cup  w}) > \beta\). 
We are justified in applying this lemma because the well-spreadness of $J$ implies that $\cup J$ is $(\frac{n}{(m(\d_1+1))^3},\frac{100n \log (\d_1+1)}{(\d_1+1) m})$-spread.

In addition, for large enough \(q\), \(\S_{w_1}\) is a \(\frac{1}{2\d_1^2}\)-local spectral expander and hence we are justified to apply \pref{prop:colored-exponential-decay-bound} deducing that \(h^1(F^J(\S_{\cup w})) = h^1(\X_s^J) > 0\) which implies $\X_s^J$ is simply connected.
\end{proof}
We remark that the well-spreadness is most probably unneeded, but it shortens our proof.

\subsubsection{Proof of the full version of the theorem} \label{sec:full-version}
In the beginning of \pref{sec:proof-of-faces-complex-lower-bound} we have proven \((\d_1,\exp(-\tilde{\Omega}(\sqrt{\d_1})))\)-swap cocycle expansion. In this subsection we shave off the log factors.

For this we need a stronger version of \pref{prop:colored-exponential-decay-bound}, for which we need some terminology.
Let \(q \leq R\) be an integer. \(\mathcal{J}_q= \mathcal{J}_q(J)\) be all the \(J' = \set{c_1',c_2',\dots,c_\ell'} \leq J\) such that \( \sum_{j\in J'} |c_j'|= q\). Let 
\[T_q(X,J) = \min_{(J',X_w), J' \in \mathcal{J}_q, w \in X[\cup J \setminus \cup J']}  \left ( \max_{i_1,i_2,\dots,i_\ell \text{ s.t. } i_j \in c_j'}  \left ( h^1(X_{w}^{\set{i_1,i_2,\dots ,i_\ell}}) \right ) \right ).\]
It may not be clear why this quantity is the correct thing to look at to improve \pref{prop:colored-exponential-decay-bound}. A full explanation appears in \cite[Section 8]{DiksteinD2023swap}.

\begin{proposition} \label{prop:improved-generic-lower-bound-for-colors}
    Let \(X\) be a partite \(\lambda\)-one sided local spectral expander for \(\lambda \leq \frac{1}{2\d_1^2}\). Let \(J = \set{c_1,c_2,\dots ,c_\ell}\) and let \(R = \sum_{j=1}^\ell |c_j| \leq \d_1\). Then \(h^1(\FX[J]) \geq \prod_{q=1}^R \Omega_{\ell}(T_q(X,J))\).
\end{proposition}

The following claim is the same as \cite[Claim 8.8.1]{DiksteinD2023swap} for this setup. The proof is also exactly the same so we omit it.
\begin{claim} \label{claim:good-tqs}
Let \(J  \subseteq \FD[](4)\). Let \(w \in \S\) be such that \(col(w) \cap (\cup J) = \emptyset\). Let 
\[q_0 = \max_{B, c}|c \cap B|\]
where \(B\) is a \(col(w)\)-bin and \(c \in J\).
Then for all \(q > 10q_0\), \(T_q(\S_{w},J) = \Omega(1)\).
\end{claim}
Note that when \(J\) is well spread, and \(w= \cup s\) for \(s \in \X^J(m-6)\), then \(q_0 = O(1)\).

\begin{proof}[Proof of \pref{thm:coboundary-expansion} (full version)]
Our starting point is, combining \eqref{eq:colors},\eqref{eq:trickle}, and \eqref{eq:tensor} in the proof of the weaker version, is that
\[h^1(\X) \geq \dots \geq \exp(-O(m)) \min_{J \in \J, s \in \X^{J}(m-6)} h^1(\tilde{\X}^{\tilde{J}}_s).\]
Recall that $J$ is a set of well-spread colors, $s\in \X^J(m-6)$, and
$\tilde J$ is a set of five colors from \pref{claim:tensoring-corollary}.
Note that by \pref{claim:link-of-a-faces-complex}, $\tilde \X_s^{\tilde J} = F^{\tilde J}(S_{\cup s})$.
By \pref{prop:improved-generic-lower-bound-for-colors}
    \[h^1(\FS[\tilde{J}]_{\cup s}) \geq \exp(-O(R)) \cdot \prod_{q=1}^R T_q(\S_{\cup s},\tilde{J})\]
where \(R = \sum_{j=1}^5 |\tilde{c}_j|\). By \pref{claim:tensoring-corollary}, \(R = O(\frac{\d_1 \log \d_1}{m \log m})\). By definition, 
\[T_q(\S_{\cup s},\tilde{J})\geq \min_{w,I} h^1(S_{\cup s \dunion w}^I).\]
We could use \pref{lem:join-links} to deduce  
\(T_q(\S_{\cup s},\tilde{J}) \geq \exp(-O(\log ^2 \d_1))\). 

However, by \pref{claim:good-tqs} we can obtain a tighter bound on \(T_q(\S_{\cup s},\tilde{J})\). 
Let \[q_0 = \max_{B, c}|c \cap B|\]
where \(B\) is a \(col(\cup s)\)-bin and \(c \in \tilde{J}\).
By \pref{def:good-colors} \(q_0 = O \left (\frac{\log \d_1}{\log m} \right )\) and by \pref{claim:good-tqs} for every \(q > 10q_0\),
\(T_q(\S_{\cup s},\tilde{J}) = \Omega(1)\). Thus 
\[\FS[\tilde{J}]_{\cup s} \geq \exp(-O(R)) \cdot \exp(-O(\log ^2 \d_1))^{10q_0} \cdot \exp(-O(R-10q_0)).\]
Plugging in \(m=\sqrt{\d_1}\) we have that \(q_0=O(1)\) so this is at least
\(\exp(-O(R + \log^2 \d_1)) = \exp(-O(\sqrt{\d_1}))\).
In conclusion, we have that \(h^1(\X) = \exp(-O(\sqrt{\d_1}))\).
\end{proof}

%% file: spectral.tex
\subsection{Local Spectral Expansion of the Symplectic Buildings} \label{sec:local-spectral-expansion-C}
In this section we prove \pref{lem:symplectic-local-spectral-expansion} and \pref{thm:expansion-of-building-quotients}, showing local spectral expansion of \(C\) and quotients of \(\tilde{C}_g\). The main proposition we need is the following one.
\begin{proposition} \label{prop:the-actual-spectral-bound}
        Let \(0\leq j<\ell\leq g\), then \(C_g[j],C_g[\ell]\) is a \(\frac{c^{\ell-j}}{\sqrt{p}^{\ell-j}}\)-spectral expander where \(c > 1\) is some universal constant independent of \(p\). In particular this is a \(O(\frac{1}{\sqrt{p}})\)-spectral expander.
\end{proposition}

The theorem follows quite easily from \pref{prop:the-actual-spectral-bound} and the following elementary claim:

The proof of the theorem and corollary rely on the following elementary claim (see e.g. \cite[Claim 2.9]{Dikstein2022} for a proof of a more general statement).
\begin{claim} \label{claim:sides-expand-imply-partite-expands}
    Let \(X\) be a \(d\)-partite complex and suppose that for every two parts \((X[i],X[j])\) the induced bipartite graph is a \(\lambda\)-spectral expander. Then the underlying graph of \(X\) is a \(\lambda\)-spectral expander.
\end{claim}

\begin{proof}[Proof of \pref{lem:symplectic-local-spectral-expansion}]
    Fix \(g\) and denote by \(C:=C_g\). Let \(w \in C_g\) and we consider \(C_w[i], C_w[j]\). By \pref{cor:link-of-symplectic-link} \(C_w\) is a join of lower dimensional spherical buildings. If \(C_w[j], C_w[\ell]\) belong to different complexes with respect to the join, then the graph between the two sides is a complete bipartite graph which is a \(0\)-one sided spectral expander. If \(C_w[j],C_w[\ell]\) both belong to belong to a building of type \(A\), then the vertices of \(C_w[j]\) are (isomorphic to) subspaces of dimension \(j'\) in \(\mathbb{F}_p^m\), the vertices of \(C_w[\ell]\) are subspaces of dimension \(\ell'\) in \(\mathbb{F}_p^m\). There is an edge between \(u_1\) and \(u_2\) if and only if \(u_1 \subseteq u_2\). It was shown by e.g. \cite{DiksteinDFH2018} that this graph is an \(O(\frac{1}{\sqrt{p}})\)-expander (where the constant is independent of \(p\)).

    The remaining case is when \(C_w[j],C_w[\ell]\) belong to a part in the join which is itself isomorphic to a spherical building associated with \(Sp(2m,\mathbb{F}_p)\) for some \(m \leq g\). In this case the graph is a \(O(\frac{1}{\sqrt{p}})\)-one sided spectral expander by \pref{prop:the-actual-spectral-bound}.

    By \pref{claim:sides-expand-imply-partite-expands} the itself if \(\frac{c}{\sqrt{p}}\)-one sided spectral expander so \(C_g\) is also a \(\frac{c}{\sqrt{p}}\)-one sided local spectral expander.
\end{proof}

The \pref{thm:expansion-of-building-quotients} is just a consequence of the theorem.
\begin{proof}[Proof of \pref{thm:expansion-of-building-quotients}]
    By \pref{fact:basic-symplectic-building} \(\tilde{C}_g\) is connected, thus \(X\) is also connected. The links of \(\tilde{C}_g\) are joins of spherical buildings of type \(C\). For every \(v \in X(v)\) the bipartite graph between two colors \((X_v[i],X_v[j])\) is either the complete bipartite graph (if \(i,j\) belong to differend parts of the join) or a graph as in \pref{lem:symplectic-local-spectral-expansion}. In both cases this graph is a \(\frac{c}{\sqrt{p}}\)-spectral expander. By \pref{claim:sides-expand-imply-partite-expands} the links are are \(\frac{c}{\sqrt{p}}\)-spectral expanders so by \pref{thm:trickle-down} \(X\) is a \(\lambda\)-one sided spectral expander for \(\lambda = \frac{\frac{c}{\sqrt{p}}}{1-\frac{c}{\sqrt{p}}} = O (\frac{1}{\sqrt{p}})\). By \pref{cor:skel-two-sided-hdx} the \(k\)-skeleton is a \(\max\set{O(\frac{1}{\sqrt{p}}),\frac{1}{d-k+1}}\)-two sided local spectral expander.
\end{proof}

The proof of \pref{prop:the-actual-spectral-bound} follows from the theory developed in \cite{DiksteinDFH2018} regarding expanding posets. We give a brief discussion of the parts of the theory we need. 

\subsubsection{Sub Posets of the Grassmann} \label{sec:subposets-of-grassmann}
The \((n,p,d)\)-Grassmann poset is the poset
\[Gr(n,p,d) = \sett{u \subseteq \mathbb{F}_p^n}{dim(u) \leq d},\]
where the order is by containment.

A simplicial sub-poset of \(Gr(n,p,d)\) is a subset \(P \subseteq Gr(n,p,d)\) such that for every \(v \in P\) and \(u \subseteq v\). We denote the \(i\)-dimensional subspaces in \(P\) by \(P(i)\). A simplicial sub-poset is \emph{pure} if for every \(u \in P\) there exists some \(v \in P(d)\) such that \(u \subseteq v\).

The measure on flags in \(P\) is via sampling a uniform \(v_d \in P(d)\) and then a uniform flag \(\set{v_0,v_1,\dots,v_d}\) where \(v_i \in P(i)\).

For every \(i < j\) we consider the containment graph \(C(P,i,j)\) between \(P(i)\) and \(P(j)\) where the probability of an edge \(\set{v_i,v_j}\) is the probability of sampling a flag containing \(\set{v_i,v_j}\).

Fix \(i\). We denote the bipartite graph operator of the containment graph between \(P(i)\) and \(P(i+1)\) by \(U_{i}\). That is, for every \(f:P(i) \to \RR\), \(U_i f: P(i+1) \to \RR\) is given by \(U_i f(v) = \Ex[u \in P(i), u \subseteq v]{f(u)}\). Denote its adjoint by \(D_{i+1}\).

The bipartite graph operator of the containment graph between \(P(i)\) and \(P(j)\) is the composition \(U_{j-1} \circ \dots U_{i+1}\circ U_i\). Therefore 
\begin{equation} \label{eq:lambda-of-containment-graph-decomposed}
    \lambda_2(C(P,i,j)) \leq \prod_{t=i}^{j-1}\lambda_2 (C(P,t,t+1)).
\end{equation}

There are two natural two-step walks on \(P(i)\) using these containment graphs.
\begin{enumerate}
    \item The upper walk that chooses a pair \(v,v' \in P(i)\) by choosing \(u \in P(i+1)\) and then two \(v,v' \subseteq u\). The graph operator for this walk is \(D_{i+1}U_i\). We also denote its non-lazy version by \(M_i\) (i.e. the walk that samples \(v,v'\) conditioned on \(v \ne v'\)). It holds that \(D_{i+1}U_i = \frac{p-1}{p^{i+1}-1}I + \left (1- \frac{p-1}{p^{i+1}-1} \right )M_i\). 
    \item The lower walk is the one that chooses a pair \(v,v' \in P(i)\) by choosing \(u \in P(i-1)\) and then two \(v,v' \supseteq u\). The graph operator for this walk is \(U_i D_{i}\).
\end{enumerate}

The following notion generalizes graph expansion to posets.
\begin{definition}[eposet]
    Let \(\overline{\gamma} = (\gamma_0,\gamma_1,\dots)\) be a vector of non-negative numbers. A sub-poset of the Grassmann is a \(\overline{\gamma}\)-eposet if for every \(i=1,\dots,d-1\)
    \[\norm{M_i - U_i D_{i}} \leq \gamma_i.\]
\end{definition}

The following theorem is by \cite{DiksteinDFH2018}.
\begin{theorem}[Theorem 8.23 in \cite{DiksteinDFH2018}] \label{thm:eposet-of-grassmann}
    Let \(P\) be a pure \(d\)-dimensional sub-poset of the \((n,p,d)\)-Grassmann. Then if \(P\) is a \(\overline{\gamma}\)-eposet then.
    \[\lambda(D_{i+1}U_i) \leq \sum_{t=1}^{i} \frac{1}{p^t} + \sum_{t=0}^i \gamma_t.\]
\end{theorem}
Work in \cite{DiksteinDFH2018} also proposes a criterion for showing \(\gamma\)-eposetness.

Let \(w \in P(i-1)\). Its link graph \(P_w\) is the graph whose vertices are all 
\(P_w(0) = \sett{w' \in P(i)}{w' \supseteq w}\).
The edges are sampled by sampling some \(u \in P(i+1), u \supseteq w\) and then sampling \(w \subseteq w',w'' \subseteq u\) conditioned on \(w' \ne w''\). We say that a poset \(P\) is a \(\overline{\gamma}\)-expander if for every \(i=0,1,\dots,d\) and every \(w \in P(i)\) \(\lambda(P_w) \leq \gamma_i\).

\begin{theorem}[Theorem 8.21 in \cite{DiksteinDFH2018}] \label{thm:grassmann-local-to-global}
    Let \(P\) be a \(\overline{\gamma}\)-link expander. Then \(P\) is a \(\overline{\gamma}\)-eposet.
\end{theorem}

\subsubsection{Proof of \pref{prop:the-actual-spectral-bound}}
Recall that the graph between the two parts is the containment graph between istropic spaces of dimension \(j,\ell\) respectively, inside some \(2g\)-dimensional space \(V\). Hence the poset \(P\) we consider is the poset of isotropic subspaces, with respect to some non-degenerate skew symmetric bilinear form.

We observe that \(P\) is a pure \(n\)-dimensional sub-poset of the \((2n,q,n)\)-Grassmann poset. Also the measure on edges in any containment graph is uniform.

\begin{proof}[Proof of \pref{prop:the-actual-spectral-bound}]
With the notation \(U_\ell\) in \pref{sec:subposets-of-grassmann}. To prove the proposition it suffices to show that there exists a universal constant \(c>1\) such that \(\lambda(U_\ell) \leq \frac{c}{\sqrt{p}}\). If we show that \(P\) is a \(\bar{\gamma}\)-eposet for \(\gamma_\ell = \frac{c'}{p^{g-\ell-1}}\) then this follows from \pref{thm:eposet-of-grassmann}. To prove this we will show that \(P\) is a \(\bar{\gamma}\)-link expander and invoke \pref{thm:grassmann-local-to-global}.

Fix \(w \in P(\ell)\). By \pref{prop:isomorphism-of-symplectic-link} the link graph is isomorphic to the link of \(\set{0}\) in \(C_{g-\ell}\). That is, the vertices are all \(1\)-dimensional subspaces inside a \(2(g-\ell)\)-dimensional space (which are all subspaces since the bilinear form is skew-symmetric), and we connect two subspaces via traversing through a \(2\)-dimensional isotropic subspace - i.e. two subspaces are connected if and only if their sum is an isotropic subspace. If \(v \oplus u\) is isotropic we write \(v \bot u\). Recall we denote the adjacency operator of this graph by \(M_0\).
It will be more convenient to analyze this graph once we add a self loop to every vertex, i.e. add laziness. This corresponds the graph whose matrix is \(M_0' = \frac{1}{\Delta}I + \frac{\Delta-1}{\Delta}M_0\) where \(\Delta\) is the regularity of the graph. As we will see shortly, in this case \(\Delta = \Omega(p^{g-\ell-1})\), so \(\norm{M_0'-M_0} = O(\frac{1}{p^{g-\ell-1}})\) and we can analyze \(M_0'\) instead of \(M_0\). We note that this is not \(D_0U_0\), since the amount of laziness we add is much smaller. It corresponds to the number of neighbors a one-dimensional space has, and not the number of one-dimensional spaces are contained in a two-dimensional space.

Let us consider the double cover of \(M_0'\), i.e. the bipartite graph whose vertices are all \(V \times \set{0,1}\) and there is an edge between \((v,i)\) and \((u,j)\) if \(v \bot u\) and \(i\ne j\). This graph is isomorphic to the containment graph of the Grassmann between subspaces of dimension \(1\) and \(2(g-\ell)-1\), and \(v \sim u\) if and only if \(v \subseteq u^{\bot}\). The isomorphism is given by \((v,0) \mapsto v\) and \((u,1) \mapsto u^{\bot} = \sett{x \in V}{\forall y \in u, \iprod{x,y}=0}\). It is well known that this graph is an \(O(1/\sqrt{p}^{2(g-\ell-1)})=O(1/p^{g-\ell-1})=\gamma_\ell\) one-sided spectral expander (see e.g. \cite{DiksteinDFH2018}). 

As for the degree of every vertex, one observes from the double cover that the degree \(\Delta\) of a subspace \(v\) (in either \(M_0'\) or the double cover), is the number of co-dimension \(1\) subspaces that contain \(v\), i.e. \(\Delta = \Omega(p^{g-\ell-1})\).
\end{proof}

%% file: appendix.tex
\section{Affine building proofs} \label{app:building-proofs}
Before proving \pref{claim:basic-properties-of-affine-symplectic-building}, we show the following.
\begin{claim} \label{claim:unique-reps}
    Let \(\set{[L_0],[L_1],\dots,[L_g]} \in \tilde{C}_g(g)\). The representatives \(L_i \in [L_i]\) such that \eqref{eq:isotropic-flag-Qp} holds (i.e. such that \(L_0\) is primitive and such that \(L_i/pL_0\) are isotropic), are unique.
\end{claim}
\begin{proof}[Proof of \pref{claim:unique-reps}]
    First let us prove that there are no two primitive lattices \(L_1,L_2\) such that \(L_1 \subsetneq L_2\). Indeed, assume otherwise. By translating via \(Sp(2g,\mathbb{Q}_p)\), we can assume that \(L_2 = L_{std}\). That is 
    \[L_{2} = \sp_{\mathbb{Z}_p}(e_1,e_2,\dots,e_g,f_1,f_2,\dots,f_g)\]
    the standard basis. We also write 
    \[L_1 = \sp_{\mathbb{Z}_p}(e_1',e_2',\dots,e_g',f_1',f_2',\dots,f_g')\]
    so that \eqref{eq:asymmetric-bilinear-form} holds for this basis. There exists \(v \in \set{e_1,e_2,\dots,e_g,f_1,f_2,\dots,f_g}\) such that in the linear combination \(v = \sum_{j=1}^g \alpha_j e_j' + \beta_j f_j'\) one of the \(\alpha_j\) or \(\beta_j\) are in \(\mathbb{Q}_p \setminus \mathbb{Z}_p\). Without loss of generality \(v=e_1\) and \(\alpha_1 \in \mathbb{Q}_p \setminus \mathbb{Z}_p\) (the proof is the same for every choice). On the one hand, from primitivity of \(L_2\), \(\iprod{e_1,f_j'} = \alpha_1 \notin \mathbb{Z}_p\). On the other hand, as \(f_j' \in L_2\) then \(f_j' = \sum_{j=1}^g \gamma_j e_j + \delta_j f_j\) with \(\gamma_j,\delta_j \in \mathbb{Z}_p\). It follows that \(\iprod{e_1,f_j'} = \delta_1 \in \mathbb{Z}_p\), a contradiction.

    This shows that if \([L_0]\) is primitive, then no other \([L_i]\) is primitive. This is because there are representatives such that \(L_i \subsetneq L_0\), and if  there were some primitive representative \(L_i' \in [L_i]\), i.e.\ \(L_i'=p^j L_i\), then \(L_0 \subsetneq L_i'\) or vice versa, which contradicts the above. In addition, the above also shows that there is a unique \(L_0 \in [L_0]\) that is primitive.

    Thus, by the definition of the relation, there is a unique representative \(L_i \in [L_i]\) such that \(pL_0 \subsetneq L_i \subsetneq L_0\) and in particular, there is a unique choice for the flag to hold. The claim is proven.
\end{proof}

\restateclaim{claim:basic-properties-of-affine-symplectic-building}
\begin{proof}[Proof of \pref{claim:basic-properties-of-affine-symplectic-building}]~
\paragraph{First item}First we note that \(L_1 \sim L_2\) if and only if for every \(A \in Sp(2g,\mathbb{Q}_p)\), \(AL_1 \sim AL_2\). Thus \(Sp(2g,\mathbb{Q}_p)\) acts on \emph{lattice classes}. Moreover, it is easy to verify that for every \([L] \in \tilde{C}_g(0)\) and \(A \in Sp(2g,\mathbb{Q}_p)\), \([AL] \in \mathbb{C}_g(0)\): if \(pL_0 \subseteq L \subseteq L_0\), then \(pAL_0 \subseteq AL \subseteq AL_0\) where \(AL_0\) is primitive. Finally, we note that by definition of \(Sp(2g,\mathbb{Q}_p)\), for any \(A \in Sp(2g,\mathbb{Q}_p)\) and \(u_1,u_2 \in L\), \(\iprod{Au_1,Au_2} = \iprod{u_1,u_2}\). In particular for every \(u_1,u_2 \in L\) and primitive lattice \(L_0\), 
\[\iprod{u_1 + pL_0, u_2 + pL_0} = \iprod{Au_1 + pAL_0, Au_2 + pAL_0}.\]
Thus \(L/pL_0\) is isotropic if and only if \(AL/pAL_0\) is isotropic. Finally, as \(A \in Sp(2g,\mathbb{Q}_p)\) clearly preserves containment between lattices, it also preserves flags as in \eqref{eq:isotropic-flag-Qp}, and even sends the \emph{unique} representatives of the flag \(\set{[L_0],[L_1],\dots,[L_g]}\) to the unique representatives of \(\set{[AL_0],[AL_1],\dots,[AL_g]}\).

\paragraph{Second item}The group \(Sp(2g,\mathbb{Q}_p)\) acts transitively on primitive lattices since these are exactly the orbit of a \(L_{std}\). Thus they also act transitively on primitive lattice classes. Let us show that the stabilizer of \([L_{std}]\) acts transitively on \((\tilde{C}_g)_{[L_{std}]}\). This shows that \(Sp(2g,\mathbb{Q}_p)\) acts transitively, since given a pair \(s_1,s_2 \in \tilde{C}_g(g)\) we can send them to a pair of faces \(s_1',s_2'\) containing \([L_{std}]\) via some \(A_1,A_2 \in Sp(2g,\mathbb{Q}_p)\). Then we find some \(B \in Stab([L_{std}])\) that sends \(s_1'\) to \(s_2'\). This will show that \(A_2^{-1}BA_1s_1 = s_2\). 

By \pref{claim:unique-reps} the space of flags is isomorphic to all the lattice flags \(\set{L_1 \subsetneq L_2 \subsetneq \dots \subsetneq L_g}\) such that 
\[L_1/pL_{std} \subsetneq L_2/pL_{std} \subsetneq \dots \subsetneq L_g/pL_{std} \subsetneq L_{std}/pL_{std}\]
is a flag of isotropic subspaces. There is a bijection between the flags \(\set{L_1 \subsetneq L_2 \subsetneq \dots \subsetneq L_g}\) with this isotropic quotient, and the quotent itself 
\[\set{L_1/pL_{std} \subsetneq L_2/pL_{std} \subsetneq \dots \subsetneq L_g/pL_{std}}.\]
So it is enough to show that the action \(A. (L_1/pL_{std}) := (A.L_1)/pL_{std}\) is transitive on flags of isotropic subspaces\footnote{I the proof of the first item, we saw that this action is well defined.}.

It is easy to verify that \(Stab(L_{std}) = Sp(2g,\mathbb{Z}_p)\). Moreover, using the projection from \(\mathbb{Z}_p\) to \(\mathbb{F}_p\) on the matrices \(Sp(2g,\mathbb{Z}_p)\) entrywise, gives a surjective homomorphism to \(\psi: Sp(2g,\mathbb{Z}_p) \to Sp(2g,\mathbb{F}_p)\). Finally, one can verify directly that \(A. (L_1/pL_{std}) = (\psi(A).L_1)/pL_{std}\). The action of \(Sp(2g,\mathbb{F}_p)\) is transitive on flags, so we can conclude the proof.

\paragraph{Third item}We note that every vertex \([L]\) has some primitive \(L'\) so that \(pL' \subseteq L \subseteq L'\) (for some \(L \in [L]\)). Inside the link of \([L']\) there is a flag containing \(L\) and \(L'\) (because there is a maximal flag in \(L/pL\) containing \(L'\)). The action of \(Sp(2g,\mathbb{Q}_p)\) is transitive on \(g\)-faces and preserves the dimension of every \([L]\) with respect to every primitive lattice that it shares a face with. Thus the colors are well defined.

\paragraph{Fourth item}This is a direct conclusion of the above; we already concluded that set of flags containing some primitive \([L_0]\) is (isomorphic to) the set of \(\set{L_1/pL_{0},L_2/pL_{0},\dots,L_g/pL_{0}}\) such that 
\[L_1/pL_{0} \subsetneq L_2/pL_{0} \subsetneq \dots \subsetneq L_g/pL_{0}\]
is an isotropic flag inside \(L_0/pL_0\) (where \(L_0 \in [L_0]\) is the primitive element).
\end{proof}